\documentclass{article}
\usepackage{amssymb,amsmath,amsthm}
\usepackage[utf8]{inputenc}
\usepackage{amsmath}
\usepackage{bbm}
\usepackage{color}
\usepackage{comment}
\newcommand{\td}{\text{d}}
\newcommand{\ti}{\text{i}}
\newcommand{\calC}{\mathcal{C}}

\newcommand{\supp}{\text{supp}}

\newtheorem{defn}{Definition}
\newtheorem{rmk}{Remark}
\newtheorem{lem}{Lemma}
\newtheorem{col}{Corollary}
\newtheorem{claim}{Claim}
\newtheorem{thm}{Theorem}

\newtheorem{info}{Informal Theorem}

\title{Wilson Loop Expectations for Non-Abelian Gauge Fields Coupled to a Higgs Boson at Low and High Disorder}
\author{Arka Adhikari \\ Stanford University\\ arkaa@stanford.edu }

\usepackage{graphicx}
\usepackage{appendix}

\usepackage{fancyhdr}
\begin{document}
\maketitle
\begin{abstract}
    We consider computations of Wilson loop expectations to leading order at large  $\beta$ in the case where a non-abelian gauge field interacts with a Higgs boson. By identifying the main order contributions from minimal vortices, we can express the Wilson loop expectations via an explicit Poisson random variable. This paper treats multiple cases of interests, including the Higgs boson at low and high disorder, and finds efficient polymer expansion like computations for each of these regimes.  
\end{abstract}

\tableofcontents
\section{Introduction}
\subsection{Background and History}

Quantum field theory and the Standard Model are one of the greatest successes of modern physics; they are able to compute the behavior of the smallest particles to remarkable accuracy. However, these physical theories have yet to be given a rigorous formulation due to multiple difficulties in trying to define the Hamiltonians used on a continuous space. Glimm and Jaffe \cite{GJ1987} attempted to define quantum field theories by first defining a Euclidean lattice gauge theory and then applying a Wick rotation to the proper complex space. Seiler \cite{SEI1982} pointed out that there were issues when trying to follow this path in general; Seiler instead proposed that the basic object one should study are appropriate random functions on an appropriate space of closed curves.

One promising strategy to do this was via the method of lattice gauge theories.
Inspired by ideas from statistical physics, Wilson \cite{Wilson} proposed lattice gauge theories as a means of computing some quantities of interest in lattice gauge theory; more specifically, he wanted to explain quark confinement. Lattice gauge theories are essentially quantum field theories defined on a discrete lattice; on this discrete space, the Hamiltonians can be well-defined. The ultimate hope is that one could take the limit of quantities defined on the lattice as the lattice size of the lattice goes to $\infty$ and the lattice spacing goes to $0$ to define a continuous theory.

This approach is promising, but yet there is still no proof that in the interesting physical dimension of $4$ there is a way to take the lattice spacing to $0$ and still obtain obtain a meaningful probability distribution in the end. However, it is still a very interesting question to determine whether there are important physical quantities of interest such that there is a meaningful limit when computed on lattice gauge theories as the size of the lattice goes to $\infty$. The most important quantity of interest are Wilson loop observables. These Wilson loop expectation will be formally defined at the end of  subsection \ref{subsec:Prelim}, but we can describe the physical meaning of these values here. In the Standard Model, many larger subatomic particles, called baryons, are theorized to be composed of smaller constituents called quarks. An unusual fact about quarks is that individual quarks are not found alone in nature. It is argued that in real world scenarios, the quarks are tightly bound to each other. Wilson in \cite{Wilson} argues that if Wilson loop expectations satisfy appropriate decay conditions, namely the Area Law, then this would be sufficient to establish quark binding.

Though the most physically relevant case of lattice gauge theories occur when studying the group $SU(4)$, there is substantial literature in both the math and physics literature studying lattice gauge theories on finite groups \cite{AF1984,Borgs1984,CJR1979,LMR1989,MP1979,PolandMan,WEG1971}.
In more recent history, there were multiple works computing the values of Wilson loop expectation for various discrete groups. In particular the papers \cite{Ch2019} and \cite{SC20} compute Wilson loop expectations in the case of pure gauge field. The paper \cite{Garbon} computes Wilson loop expectations for a pure $U(1)$ gauge group. However, in the Standard Model, it is expected that the gauge field interacts with multiple other particles of interest; the interaction with other particles makes the analysis substantially more completed. This model considers the computation of Wilson loop expectations when our Hamiltonian includes an interaction with a Higgs boson. Under certain conditions, we would expect that the analysis in this work could be adapted to the case of multiparticle interactions without too much difficulty. We remark here that the work \cite{Viklund} treats the problem of computing Wilson loop expectations with a Higgs boson, but only in the case that the gauge group is abelian and there is low disorder in the Higgs field. This work substantially generalizes the analysis to non-abelian groups and to low disorder as well, amongst other generalizations. In the proceeding subsections, we give further mathematical background on the problem of computing Wilson loop expectations as well as our main results.
\subsection{Organization of the Paper}
 In this short subsection, we describe the organization of the paper.
 In the rest of the introduction, we present the notation we will use as well as further details about the problem of computing Wilson loop expectations on models containing a Higgs boson. We discuss some of the history behind the problem and introduce our two major new regimes, the low disorder and the high disorder regime.
 
 These two regimes have different polymer expansions an each have their own section. Section 2 analyzes the low disorder regime with a toy model, where $G$
 and $H$ are both abelian $Z_2$ groups. Within this section, we introduce the model in Subsection 2.1, introduce our polymer expansion in subsection 2.2, perform our identification of the main order contributions in Subsection 2.3, and finally compare our model to a computation involving a Poisson random variable in Subsection 2.4.

 Wondrously, there are no substantial differences in the low disorder regime between the abelian and non-abelian cases. Thus, Section 3 has the simultaneous job 
of introducing the notation of the non-abelian Higgs boson model in subsection 3.1 and perform the brief modifications necessary to analyze the low disorder regime in Subsection 3.2.

In Section 4, we give the more involved analysis of the lower disorder regime. The first two Subsections 4.1 and 4.2 introduce elements of our polymer like expansion. Subsection 4.3 identifies the main order contribution and Subsection 4.4 performs a comparison to a Poisson random variable.

In Section 5, we perform a far more delicate analysis of the error terms in Section 4 involving particular properties of the minimal vortices that form the crux of our main-order contribution in Section 4. There are delicate decorrelation estimates established in Subsection 5.1 which are helpful in improving the error term and then improving the approximation by Poisson random variables in Subsection 5.2. 

\subsection{Preliminary Notation and Discussion} \label{subsec:Prelim}

Lattice gauge theories, as their name implies, are studied in some lattice $\Lambda_N:=[-N,N]^d$ which is a sublattice of $\mathbb{Z}^d$. As a graph, the structure is inherited as a subgraph of $\mathbb{Z}^d$.  We let $V_N$ denote the set of vertices of $\Lambda_N$ and we let $E_N$ be the set of oriented edges of $\Lambda_N$; the orientation distinguishes the edge $e=(x,y)$ from $-e=(y,x)$. Also important in lattice gauge theories are the set of plaquettes $P_N$. If we let $x_1,x_2,x_3,x_4$ be a set of four vertices such that $x_{i+1}$ is adjacent to $x_i$ (with $x_5 = x_1$), then the square bounded by these four vertices will be called a plaquette. Plaquettes have a natural orientated inherited from the order in which we traverse the boundary vertices. (Thus, if we let $p$ be the plaquette with boundary vertices traversed in the order $x_1,x_2,x_3$ and then $x_4$, then $-p$ is the plaquette with boundary vertices traversed in the order $x_4,x_3,x_2$ and then $x_1$.)

A gauge configuration $\sigma$ with gauge group $G$ on a lattice, $\Lambda_N \subset \mathbb{Z}^d$ is an assignment $\sigma: E_N \to G$ such that $\sigma_{-e} = (\sigma_{e})^{-1}$.  The gauge group $G$ also comes with a unitary representation $\rho: G \to U_{D}$ to the appropriate space of $D$ by $D$ unitary matrices.

A Higgs boson configuration is the assignment $\phi: V_N \to \mathbb{C}$ of a complex valued function at each vertex. Later in this introduction, we will provide some restrictions on the values that $\phi$ could take. Associated to the Higgs boson is a covariant derivative that mixes the action of Higgs boson and the gauge field. Let $e$ be the edge that connects the vertex $x$ to the vertex $y$. Then, we have,

\begin{equation}
    D_e \phi =  \phi_x \text{Tr}[\rho(\sigma_{e})]  - \phi_y.
\end{equation}

With this, we have the following interaction term between the Higgs boson and the gauge field,
\begin{equation}
    S^{D}(\sigma,\phi):= \frac{1}{2} \sum_{e \in E_N} |D_{e} \phi|^2 = 2d \sum_{x \in V_N} |\phi_x|^2 - \sum_{e=(x,y) \in E_N} \phi_x \text{
    Tr}[\rho(\sigma_e)] \overline{\phi}_y.
\end{equation}

A gauge field configuration also has a self-energy, corresponding to the curvature form, or Wilson action functional. Associated to any oriented plaquette $p$, whose boundary vertices are $x_1,x_2,x_3$ and $x_4$ ( with edges $e_i=(x_i,x_{i+1})$), we can define the action $\psi_p$ along the plaquette to be $\text{Tr}[\rho(\sigma_{e_1}) \rho(\sigma_{e_2}) \rho(\sigma_{e_3}) \rho(\sigma_{e_4})]$. Note that due to the cyclic property of the trace, it is indeed the case that the value of $\psi_p$ does not depend in particular on which vertex we start the ordering of the boundary vertices( though the actions will differ if we consider $\psi_p$ or $\psi_{-p}$).
The Wilson action function will be the sum of the actions $\psi_p$ over all oriented plaquettes $\psi_p$.
\begin{equation}
    S^W(\sigma) := - \sum_{p \in P_N} \psi_p(\sigma). 
\end{equation}

For an abelian group and a 1-dimensional representation $\rho$, we can represent this in a much more direct form as follows:
\begin{equation}
    S^W(\sigma):= -\sum_{p \in P_N} \rho((\td \sigma)_p).
\end{equation}

If we consider the oriented plaquette $p$ with oriented boundary edges $e_1$,$e_2$,$e_3$ and $e_4$, then $(\td \sigma)_p$ is the well-defined element $\sum_i \sigma_{e_i}$. We needed to introduce more general notation for the non-abelian case since we are no longer able to assign a canonical element to each plaquette.


The Higgs boson also has a self- interaction given by the following form,
\begin{equation}
    V(\phi):= \sum_{x \in V_N}(|\phi_x|^2 -1)^2.
\end{equation}

Consider the following interaction Hamiltonian between the Higgs boson and a gauge field,
\begin{equation} \label{eq:basicinter}
    H_{N, \beta,\kappa,\zeta}(\sigma,\phi):=  -\beta S^{W}(\sigma) - \kappa S^D(\sigma,\phi) - \zeta V(\phi).
\end{equation}

We consider the formal limit when we take $\zeta \to \infty$. In this limit, the values of $|\phi_x|$ are confined to $1$, so $\phi_x$ can only take values along the unit circle. We will let $\phi_x$ take values in some finite subgroup, $H$, of the multiplicative group of the unit circle. The group $H$ is isomorphic to $Z_k$( the additive group of integers mod $k$) for some $k$; namely, $\phi_x$ can only take values of the form $e^{2 \pi \ti \frac{j}{k} }$ for some $0 \le j <k$. We also remark that after fixing $|\phi|=1$, we have that $\bar{\phi}=\phi^{-1}$. This replacement will be used many times when expressing the Hamiltonian.

We consider the following measure of gauge and Higgs boson configurations on the lattice $\Lambda_N$.

\begin{equation}
    \mu(\sigma, \phi) = \frac{1}{Z} \exp[H_{N,\beta,\kappa}(\sigma,\phi)] ,
\end{equation}
where $Z$ is the partition function coming from the sum of $\exp[H_{N,\beta,\kappa}(\sigma,\phi)]$ over all configurations of $\sigma$ and $\phi$.

On loop $\gamma$, the Wilson loop action $W_{\gamma}(\sigma,\phi)$ is the product $\text{Tr}[\prod_{e\in \gamma} \rho(\sigma_e)]$  of the $\sigma$'s along $\gamma$. As before, the cyclic property of the trace ensures that this quantity is well-defined.

\subsubsection{Previous Work: Wilson Loops in Abelian Lattice Higgs Model}

In this previous work, the authors considered the case that $G$ was abelian, $\rho$ was a one-dimensional representation, and $H$ is contained inside the image of $\rho(G)$(which will be a subgroup of the unit circle). This simplification allows one to effectively gauge out the effect of the Higgs boson.

Before describing this construction, we remark that we can always assume that the only element $g\in G$ such that $\rho(g)=1$ is the identity element. If this were not true, then we can remove a trivial gauge invariance that does not change the Hamiltonian action, but modifies the group $G$ so that there is only the only element with $\rho(g)=1$ is the identity. The construction is as follows: first let $G'$ be the subgroup of elements $g'$ in $G$ consisting of those elements with $\rho(g')=G$. Since multiplying a configuration $\sigma_e$ by some element $g' \in G'$ at any edge does not change the Hamiltonian, we see that this is a trivial gauge invariance. We might as well consider the quotient group $G/G'$ as our gauge group. This $G/G'$ has the property that the only element $g$ with $\rho(g)=1$ is the identity element. 

Now let us return to our simplification of the Higgs boson action.
We can make the following map that preserves the action, but simplifies the Higgs boson term. 
Let $\eta_x$ be the term such that $\rho(\eta_x) = \phi_x$ (this is unique since $\rho(g)=1$ only if $g=1$), then we see that under the map $(\phi_x, \sigma_{e}) \to (1, \eta_x^{-1} \sigma_e \eta_y)$ where $(x,y)$ are the endpoints corresponding to the edge $e$. We let the map $\sigma_{e} \to \eta_x^{-1} \sigma_e \eta_y$ be $F_{\phi}$. For any choice of $\phi$, this $F_{\phi}$ is a bijection between the set of gauge field configurations $\sigma_e$ and itself.

This argument shows that when considering the expectation of a gauge invariant function, such as a Wilson loop expectation, one can merely consider the expectation with respect to the net Hamiltonian,
\begin{equation} \label{eq:simpHam1}
    H_{N, \beta, \kappa} = -\kappa \sum_{e  \in E_N} \rho(\sigma_e) - \beta \sum_{p \in P_N} \rho( (\td \sigma)_p).
\end{equation}

This reduces the question of the Higgs boson to a slightly modified question about pure gauge configurations. To some extent, the paper \cite{Viklund} only dealt with case under the assumption that $\kappa$ is sufficiently large.

\subsection{Extension of the Analysis of Wilson Loop Expectation}

One can imagine multiple changes to the model that make the model proposed in \eqref{eq:simpHam1} substantially more difficult to analyze. 

One question that one could ask is the following:

\textit{What happens if the group $H$ for the Higgs boson is not contained in the image $\rho(G)$ of the gauge group? }

The above question is well-defined for abelian groups with a 1-dimensional representation. For more general groups and representations, the analogue is follows. Let $\tilde{G}$ be the group of elements whose image under the representation $\rho$ is a multiple of the identity, i.e. $g \in \tilde{G}$ if $\rho(g) = c I$, where $c$ is some constant. We let $X$ be the group of these coefficients $c$ from elements of $\tilde{G}$; since $\rho$ is a unitary representation, $X$ is a subgroup of the unit circle. The analogue we have to consider is that $H$ is not a subgroup of these elements $X$.

If $H$ were not contained in $X$, one immediately sees that one cannot immediately gauge out the effects of the Higgs boson. If one now has to consider a non-trivial Higgs boson group, one must also need to consider the excitations of the Higgs field.

For example, if the Higgs field is constant, one automatically knows that the lowest energy states are those that set $\sigma_e =1$ and any other configuration on edges can be reasonably understood to be an exponentially suppressed excitation.

However, if we do not start with a constant Higgs field, it may be possible that we could get a lower energy configuration by setting $\sigma_e \ne 1$ instead of $\sigma_e =1$.

\subsubsection{The Excitations of the Higgs field}



To describe this phenomenon, we will first give a simple calculation. 
Consider the case that $G = Z_3$ while $H= \{ e^{2 \pi \ti \frac{k}{9}}\}$ for $k$ between $0$ and $9$ and $\beta=0$ (or at least $\beta \ll \kappa$). Observe that $H$ is isomorphic to $Z_9$.

In this case, one can apply a gauge transformation  similar to one applied in the previous subsection. However, $\rho(\sigma_e)$ can only take values $\{1, e^{2\pi \ti \frac{3}{9}}, e^{2 \pi \ti \frac{6}{9}} \}$ rather than all of $H$. Thus, the furthest we are able to simplify the values of $\phi_x$ under a gauge transformation are to $\{1, e^{2\pi \ti \frac{1}{9}}, e^{2 \pi \ti \frac{2}{9}}\}$. Let $\theta_x=\{1, e^{2\pi \ti \frac{1}{9}}, e^{2 \pi \ti \frac{2}{9}}\}$ be the simplified value of $\phi_x$ under the gauge transformation.

Now consider an edge $e= (x,y)$ such that $\theta_x= e^{2\pi \ti \frac{2}{9}}$ and $\theta_y =1$.
We see that the maximum value of $\exp[ \theta_x \rho(\sigma_{e}) \bar{\theta_y} + \bar{\theta_x} \rho(\sigma_{e}^{-1}) \theta_y ]$  is $\exp[2 \cos( 2 \pi \frac{8}{9}) ]$ which occurs when $\sigma_{e} = e^{2 \pi \ti \frac{6}{9}}$. Observe that the maximum is not $\exp[2]$, as one would get if $\theta_x=\theta_y$.

What this illustrates is the fact that if one fixes the configuration of the Higgs boson at every lattice point, one could possibly have substantial changes to the lowest energy configuration. Thus, in order to analyze the Higgs boson model in general, one would necessarily need a method to analyze the excitations of the Higgs boson beyond just reducing it to a case of a modified pure gauge field.
We will show that various scales of $\kappa$ will change the behavior of the excitations of the Higgs boson. One will have to treat these different regimes through different types of analysis.
\subsection{The Low Disorder Regime: High $\kappa$}

In the case that $\kappa$ is large, much like the case dealt with in \cite{Viklund}, one finds that the large value of $\kappa$ prevents excitations of $\sigma_e \ne 1$ for many edges. They can show that in this regime, one can reduce the computation of Wilson loop expectations to the computation of the contribution from minimal vortices, e.g. excitations that consist of single edge with $\sigma_e \ne 1$ with all other edges in a neighborhood around it set to $\sigma_e=1$.

When we introduce a Higgs boson field, we can assert a slightly similar statement. Though it is not the case that, given a fixed Higgs boson field configuration, the lowest energy configuration fixes all values of $\sigma$ to $\sigma_e=1$, we do see that an edge $e=(x,y)$ with $\sigma_e \ne 1$ and arbitrary Higgs boson values $\phi_x$,$\phi_y$ at the boundary vertices will certainly have less energy at that edge than a configuration that sets $\phi_x=\phi_y$ and $\sigma_e=1$.
Thus, if one is willing to modify the Higgs field in addition to the gauge field configuration, one may be able to treat edges with $\sigma_e \ne 1$ as an excitation.

This is exactly what is done in Section \ref{sec:ToyModel}.
In this section, beyond just generalizing the analysis of \cite{Viklund} to a more general Higgs boson group, we simplify the analysis of \cite{Viklund} by providing a true polymer expansion. One further consequence of this polymer expansion is that it helps express the Wilson loop expectation as a computation involving an explicit Poisson random variable. Roughly speaking, if we define the support of a configuration $(\sigma,\phi)$ of a joint gauge field and Higgs field configuration as those plaquettes that bound excited  edges $e$ with $\sigma_e \ne 1$ or $\phi_x \ne \phi_y$, then we derive the desired properties behind a true polymer expansion. The details of the polymer expansion are given in Definition \ref{def:supportlowdisorder}. The main principle behind a polymer expansion is that it is easy to analyze probabilities if one can always perform the following procedure. If a configuration has support $P_1 \cup P_2$, where $P_1$ and $P_2$ are disjoint, then configuration can be bijectively split into a configuration supported on $P_1$ and another supported on $P_2$.

There is a rather straightforward way to do this splitting for the gauge field configurations $\sigma$, but one has to exert effort in order to split the Higgs field configurations into the disjoint supports. Section \ref{sec:ToyModel} performs this analysis in a toy case that $G=Z_2$ and $H=Z_2$, where the splitting is easier to describe. The analysis for more general gauge group $G$ and Higgs boson group $H$ is similar, but more cumbersome in the description of the division.

We informally describe our main theorem in this section as follows.

\begin{info}[Informal Version of Theorems \ref{thm:reductiontominimal} and \ref{thm:MainThm1} of Section \ref{sec:ToyModel} ]
For sufficiently large $\beta$ and $\kappa$,
the main order contributions to the Wilson loop observables are by the number of minimal vortices (roughly, those plaquette exciations centered around a single edge with $\sigma_e \ne 1$) that are centered around edges of $\gamma$.
Furthermore, the number of minimal vortex excitations along the edges of $\gamma$ can be treated as roughly a Poisson random variable with parameter $O(|\gamma|)$.

Define the quantity $$A_{\beta,\kappa}:= \frac{\sum_{g \ne 1} \rho(g) \exp[12\beta\text{Re}[(\rho(g) - \rho(1))]] \exp[2\kappa \text{Re}[\rho(g) - \rho(1)]]}{\sum_{g \ne 1}  \exp[12\beta\text{Re}[(\rho(g) - \rho(1))]] \exp[2\kappa \text{Re}[\rho(g) - \rho(1)]]}$$

\begin{equation}
    \mathbb{E}[W_{\gamma}] \approx \mathbb{E}[A_{\beta,\kappa}^{X} ],
\end{equation}
where $X$ is a Poisson random variable with expectation $|\gamma| \exp[12 \beta\text{Re}[\rho(g)-\rho(1)] \exp[2\kappa \text{Re}[\rho(g)- \rho(1)]]$.
\end{info}

Wondrously, our polymer expansion is robust enough to treat the case of non-abelian $G$ with few changes to the proof.

\subsection{The High Disorder Regime: $\kappa $ Low}

A difficulty in the case of small $\kappa$ is one can no longer consider edges with $\sigma_e \ne 1$ as excitations of the configuration $(\sigma,\phi)$. Instead, the predominant suppression to the probability is caused by excited plaquettes with $\rho((\td \sigma)_p) \ne 0$. Due to this fact, we see that it would be better to separate the gauge field excitations as follows. We first create an auxiliary field $\eta:V_N \to G$ satisfying the property that $\tilde{\sigma}_e= \eta_x \sigma_e \eta_y^{-1}$ has as `few' nontrivial edges with $\tilde{\sigma}_e \ne 1$ as possible. The effects of the Higgs boson interaction on the gauge field configuration can be understood as a fluctuation of the auxiliary field $\eta_x$.  What follows is an informal discussion of the procedure being described.
\subsubsection{Informal Discussion}
To do this, we introduce a spanning tree $T$ with base point $b$ of the lattice $\Lambda$. Our goal is to find the auxiliary field $\eta: V_N \to G$ such that the field $\tilde{\sigma}_e= \eta_x^{-1} \sigma_e \eta_y$ for $e=(x,y)$ satisfies some particular properties. The property we try to ensure is the following: for each edge $e \in T$ we assign the value $\tilde{\sigma}_e=1$. 

The auxiliary field $\eta$ can be constructed inductively as follows. We first set $\eta_b=1$. Now, for any other vertex $v$ on the lattice $\Lambda_N$, let $p = (b=v_0),v_1,v_2,\ldots,(v_n=v)$ be the path on the spanning tree $T$ connecting $b$ to $v$. Assume that we have already assigned the values $\eta_{v_0}$ to $\eta_{v_{n-1}}$. Then, we choose $\eta_{v_n}:= \sigma_{e}^{-1} \eta_{v_{n-1}}$, where $e$ is the edge $(v_{n-1},v_n)$. Clearly, with this choice of $\eta_{v_m}$, we have $\tilde{\sigma}_e=1$ for $e=(v_{n-1},v_n)$. Since $T$ is a spanning tree, this fixes all values of $\eta$ on $V_N$. Finally, we can define $\tilde{\sigma}$ as $\tilde{\sigma}_e = \eta_x^{-1} \sigma_e \eta_y$ for $e=(x,y)$. We remark that in the course of the proof, $T$ will not be fixed a-priori and will be chosen as is most appropriate for that part of the proof.


Let us give some heuristics for the partition function and other associated quantities for this new Hamiltonian with the auxiliary field $\eta$. In the computation of the partition function
to highest order in $\beta$ (for $\beta$ sufficiently large and regardless of the specific value of $\kappa$), we can ignore the fluctuation of $\tilde{\sigma}_e$ and set $\tilde{\sigma}_e=1$ for all edges $e$. The partition function can be treated as a sum over all possible configurations of $\eta_v$ and $\phi_v$ for the Higgs boson.

Namely,
$Z \approx \sum_{\eta_v,\phi_v} \exp[ \sum_{e=(v,w)} \kappa \phi_v \text{Tr}[\rho(\eta_v \eta_w^{-1})] \phi^{-1}_w ]$. The configuration of the values of $\tilde{\sigma}_e$ are fixed to $1$, since these excitations would be lower order in $\beta$.

When computing the probability of seeing a single excited minimal vortex $P(e)$ in the Wilson action, the probability can be computed as a fraction with some numerator and $Z$ as the denominator.  To leading order, we see that the numerator can roughly be expressed as a sum over configurations with $\tilde{\sigma}_e' = 1$ for $e'$ not the center of the single excited vortex and $\sigma_e\ne 1$ for $e=(a,b)$ the center of the single excited vortex. We see that it would be important to specifically consider the values of $\eta_a,\eta_b,\phi_a,\phi_b$.

We see that one way to write the numerator as,
\begin{equation}
\begin{aligned}
&  \sum_{\tilde{\sigma}_e \ne 1} \exp[12 \beta (\text{Tr}[\rho(\tilde{\sigma}_e)] - \text{Tr}[\rho(1)])] \\&\sum_{\hat{\eta}_a,\hat{\eta}_b, \hat{\phi}_a ,\hat{\phi}_b} \exp[\kappa \text{Re}[\hat{\phi}_a \text{Tr}[\rho( \hat{\eta}_a \tilde{\sigma}_e \hat{\eta}_b^{-1})] \hat{\phi}^{-1}_b -  \hat{\phi}_a \text{Tr}[\rho(\hat{\eta}_a \hat{\eta}_b^{-1})]\hat{\phi}^{-1}_b ]]\\&\sum_{R} \exp[\sum_{e=(v,w)} \kappa \phi_v \text{Tr}[\rho(\eta_w \eta_w^{-1}) ]\phi^{-1}_w],
\end{aligned}
\end{equation}
where $\sum_R$ denotes the sum over all elements not specified by $\eta_a,\eta_b,\phi_a,\phi_b$. Additionally, for the quantity inside the second sum, we will use the convention $\phi_a= \phi_a$ and so on.
The reason we introduce this notation is to specifically indicate that we fix the values of $\phi_a$ to $\hat{\phi}_a$ and so on.

Now, the ratio
\begin{equation}\label{eq:exampure}\frac{\sum_{R} \exp[\sum_{e=(v,w)}  \kappa \phi_v \text{Tr}[\rho(\eta_w \eta_w^{-1})] \phi_w^{-1}]}{\sum_{\eta_v,\phi_v} \exp[ \sum_{e=(v,w)}\kappa \phi_v \text{Tr}[\rho(\eta_v \eta_w^{-1})] \phi_w^{-1} ]}
\end{equation}is the probability that under the  Ising type interaction $\phi_v \text{Tr}[\rho(\eta_v \eta_w^{-1})] \phi_w^{-1}$
on edges, we will get $\eta_a= \hat{\eta_a}, \eta_b= \hat{\eta_b}, \phi_a= \hat{\phi}_a, \phi_b= \hat{\phi}_b$ for some specified values of $\hat{\phi_a}, \hat{\phi_b}, \hat{\eta}_a,\hat{\eta}_b$. These can be considered to be some type of magnetization.

We see at low $\kappa$, we see that understanding the behavior of the $\eta$'s and the $\phi$ fluctuations are non-trivial questions.  Furthermore, we see that the introduction of the $\eta$ and $\phi$ fields can lead to highly non-trivial correlations.


From the example above, we can imagine what would happen if we consider a more general family of excited plaquettes. For example, let us consider two excited minimal vortices centered at edge $e_1=(a_1,b_1)$ and $e_2=(a_2,b_2)$. To do this, we have to adjust the computation of the numerator to specifically fix the values of $\eta_{a_i,b_i}$ and $\phi_{a_i,b_i}$.

Namely, we see that we can better express the numerator as,
\begin{equation}
\begin{aligned}
    &\sum_{\tilde{\sigma_{e_i}}} \prod_i \exp[12 \beta (\text{Re}[\rho(\tilde{\sigma_{e_i}})]- \text{Re}[\rho(1)])]\\
    &\sum_{\hat{\eta}_{a_i,b_i}, \hat{\phi}_{a_i,b_i}} \prod_i \exp[\kappa \text{Re}[\hat{\phi}_{a_i} \text{Tr}[\rho(\hat{\eta}_{a_i}  \tilde{\sigma}_{e_i}\hat{\eta}^{-1}_{b_i})] \hat{\phi}_{b_i} -  \hat{\phi}_{a_i} \text{Tr}[\rho(\hat{\eta}_{a_i}  \hat{\eta}^{-1}_{b_i})] \hat{\phi}_{b_i}]] \\&\sum_{R} \exp[\sum_{e=(v,w)}  \kappa \phi_v \text{Tr}[ \rho(\eta_w \eta_w^{-1})] \phi_w],
\end{aligned}
\end{equation}
where as before $\sum_{R}$ is the sum over variables that are not $\hat{\eta}_{a_i,b_i}$, $\hat{\phi}_{a_i,b_i}$ while $\phi_{a_i}$ is set equal to $\hat{\phi_v}$ in the second sum.

The important ratio we have to consider in this example is,
\begin{equation} \label{eq:Examcorrelation}
    \frac{\sum_{R} \exp[\sum_{e=(v,w)}  \kappa \phi_v \text{Tr}[\rho(\eta_w \eta_w^{-1})] \phi_w]}{\sum_{\phi_v,\eta_v} \exp[\sum_{e=(v,w)}  \kappa \phi_v \text{Tr}[\rho(\eta_w \eta_w^{-1})] \phi_w]}.
\end{equation}

We see that the computation involves understanding  the correlation between $\hat{\eta}_{a_i}, \hat{\eta}_{b_i}$ and $\hat{\phi}_{a_i},\hat{\phi}_{b_i}$ at different sites.  Though at small $\kappa$ we would expect exponential decay of correlations,  it is the very presence of these correlations that makes it difficult to define a polymer expansion in the case of small $\kappa$ purely from knowing the values of $\sigma$ and $\phi$, even in the abelian case.

For simplicity of notation, let us consider the abelian case. The main issue can be observed by looking at the following example. Consider a gauge field configuration $\sigma_1$ such that the set of excited plaquettes , those with $\td(\sigma)_p \ne 0$ is a small set with size $k$, but the set of edges with $\sigma_e$ with $\sigma_e \ne 0$ is much larger, say of $O(k^4)$. This is an obstruction to a polymer expansion.

Naively, what this means is that if we find another configuration $\sigma_2$ whose support contains one of these $O(k^4)$ edges, then there will be correlations between observing both $\sigma_1$ and $\sigma_2$. Namely, $\mathbb{P}(\sigma_1 \sigma_2)$ is far different from $\mathbb{P}(\sigma_1) \mathbb{P}(\sigma_2)$. This is in contrast to the case of a pure abelian gauge field, where configurations can easily be split as the set of excited plaquettes associated to $\sigma_1$ and $\sigma_2$ are disjoint.

For small $\kappa$, we found an auxiliary Hamiltonian based on the random currents representation of the Ising model \cite{Copin-Tassion} that exactly allows us to characterize when the effects of $\kappa$ cause two disjoint configurations to be correlated with each other.
The random current expansion introduces a new field $I(e)$ for each edge $e$ and couples our original Hamiltonian, $H_{N}(\sigma,\phi)$ to a new Hamiltonian $\mathcal{H}(\sigma,\phi,I)$ , equation \eqref{def:mathcalHam}, in the new variables $I$. One can define a polymer expansion using these new variables, as in Definition \ref{def:clustnonabel}. Unfortunately, even in the abelian case, the introduction of the Higgs Field leads to knotting problems, as in the non-abelian pure gauge field case \cite{SC20}. Through this new polymer expansion, we show that the Wilson loop expectation can be reduced to understanding the effects of minimal vortices lying along the edges of $\gamma$ in section \ref{sec:nonabelianhighdisorder}.
There is a simple argument with a worse error rate showing that the number of minimal vortices along $\gamma$ can be treated as a Poisson random variable in the end of the same section \ref{sec:nonabelianhighdisorder}. However, by carefully understanding the decorrelation between the contribution of the Wilson loop action from different minimal vortices, we can improve our error analysis. Roughly speaking, this involves relating quantities of the form \eqref{eq:Examcorrelation}
to products of those of the form \eqref{eq:exampure}. This rather delicate task was performed in Section \ref{sec:decorrelation}.

We have the following informal version of our main result,
\begin{info}[Based on Theorems \ref{thm:mainthmnonabelian} and \ref{thm:mainthm3} of Section \ref{sec:nonabelianhighdisorder} and Theorem \ref{thm:MainThm2} of Section \ref{sec:decorrelation}]
For $\beta$ sufficiently large and $\kappa$ sufficiently small, the main order contribution to Wilson loop expectations come from minimal vortices centered around the edges of $\gamma$ .

One can compute the Wilson loop expectation as,
\begin{equation}
    \mathbb{E}[W_{\gamma}] \approx \mathbb{E}[\text{Tr}[\mathcal{D}_{\beta,\kappa}^{X}]],
\end{equation}
where $\mathcal{D}_{\beta,\kappa}$ is a matrix defined in Theorem \ref{thm:WilsonLoopReplowdisorder} and $X$ is a Poisson Random Variable whose expecation is $O(|\gamma|)$.
\end{info}





\section{
The Low Disorder Regime: A Toy Model} \label{sec:ToyModel}

\subsection{Introduction and Informal Discussion}

As we have discussed in the introduction, when the gauge group is not the same as the symmetry group of the Higgs field, the fluctuations of the Higgs boson will affect the lowest energy configurations of the gauge field. To illustrate the main ideas in the low disorder regime(large $\kappa$), we consider the following slightly simplified model.

On the lattice $\Lambda_N=[-N, N]^4$, we allow  the Higgs boson $\phi_p$ at each lattice point $p \in \Lambda_N$ to take 2 values, either $+1$ or $-1$ (which we will later refer to $+$ and $-$, respectively, and will be called charges). The gauge group $G$ of the gauge field $\sigma$ is $Z_2$; as such, we will talk about setting edges $e$ to $\sigma_e=1$ or $\sigma_e=-1$. We will also let $\rho$ be a 1-dimensional representation of $Z_2$.

We remark that since for all elements $g\ \in Z_2$, we have $g=g^{-1}$, we do not need to concern ourselves with the orientations of edges or of plaquettes. Thus, we do not need to take into account the orientation of edges. Furthermore, since Higgs boson values $\phi$ are assigned to vertices and gauge field values $\sigma$ are assigned to edges, we may use language as `assign + to a vertex v' to unambigously mean assigning the Higgs boson value $\phi_v=+$ at the vertex $v$ and `assign -1 to an edge e' to mean assigning the gauge field value $\sigma_e,\sigma_{-e}=-1$ to either orientation of the edge $e$.   
\subsubsection{Higgs boson Configuration}
We start by describing, in words, the type of interaction between the gauge group and the Higgs field.
 Consider an edge $e \in E_N$ of the lattice. If it connects two points of the same charge, then assigning $\{+1\}$ to $e$ will give energy $E_1$ and assigning $\{-1\}$ to $e$ will give energy $E_2<E_1$.   If instead, the edge $e$ connects two vertices of opposite charges, then assigning $\{-1\}$ as the gauge field to edge $e$ will give energy $E_3$ and assigning $\{1\}$ to the edge $e$ will give energy $E_4< E_3$. The energies are ordered as follows
\begin{equation*}
    E_1> E_3, E_2, E_4.
\end{equation*}
We will have a standard Wilson loop energy of the form $\beta \rho((\td \sigma)_p)$, where for the abelian group $Z_2$ we can understand $(\td \sigma)$ as the sum of all the gauge group elements along the edges that bound $p$.


After subtracting an appropriate constant, we can formally write our Hamiltonian as follows,
\begin{equation} \label{eq:toyHam}
    H_{N,\beta,\kappa}(\phi,\sigma) =  \sum_{e=(v,w) \in E_N}  \kappa [f( \sigma_e,\phi_v \phi_w^{-1}) - f(1,1)]+ \beta \sum_{p \in P_N} [\rho((\td \sigma)_p) - \rho(1)].
\end{equation}
where $f(1,1)= E_1, f(1,-1) = E_2, f(-1,1) = E_3, f(-1,-1)= E_4$ and our measure on the lattice is
\begin{equation}
     \mu_{N,\beta,\kappa}(\sigma, \phi) =  \frac{1}{Z_{\Lambda,\beta}} \exp[H_{N,\beta,\kappa}(\sigma,\phi)] .
\end{equation}
Here $Z_{\Lambda,\beta}$ is the partition function,
\begin{equation}
   Z_{\Lambda,\beta}= \prod_{v \in V_N}\sum_{ \phi_v \in \{-1,1\}}  \prod_{e \in E_N}\sum_{ \sigma_e \in Z_2}  \exp[H_{N,\beta,\kappa}(\sigma_,\phi)].
\end{equation}


The action considered above is rather similar to the action in \eqref{eq:basicinter} if the  gauge group $G$ were chosen to be $Z_2$ and the Higgs boson group $H$ is $Z_4$. In this case, one would only be able to gauge out the Higgs boson field to take only two values ( $+$ or $-$ ) here; + corresponds to $\phi_x= 1$ or $-1$ while $-$ corresponds to $\phi_x= i$ or $\phi_x= -i$. In this case, if there is an edge $e= (v,w)$ such that  $\phi_v$ and $\phi_w$ are assigned the same value, then the lowest energy configuration would assign $\{+1\}$ to $\sigma_e$. By contrast, if there is an edge $e=(v,w)$ such that $\phi_v$ and $\phi_w$ are assigned different values, then assigning $\sigma_e$ to be $\{-1\}$ and $\{+\}$ would give the same energy to the edge; this energy $E$ would be less than the maximum energy assigned when $\phi_v$ and $\phi_w$ are the same sign and $\sigma_e$ is assigned the value $\{+1\}$.


With the Hamiltonian in \eqref{eq:toyHam} and ignoring the effect of the Wilson action for now, we expect the following behavior. Each of the sites would have the same sign (either + or -),  while we would expect that each of the gauge fields would be assigned the value $\{+1\}$. The fact that we would expect nearly all vertices to be assigned the same charge is due to the constraints $E_3,E_4 <E_1$ and a Peierl's argument.

Consider an assignment of Higgs boson charges in which the $\phi$'s are not constant. WLOG, we can consider the case that there is a connected neighborhood $N$ of vertices assigned a negative charge, but this neighborhood $N$ is surrounded by an ocean of positive charge. We can let $E(N)$ be the edges that connect $N$ to its complement $N^c$. We now argue that we can obtain a less excited configuration by flipping the charges of the vertices in $N$ and changing the assignment of all the edges in $E(N)$ to 1. Clearly, the change in the energy of this configuration is given by at least $\kappa |E(N)|(E_1 - \max(E_3,E_4))$, where $|E(N)|$ is the size of the set $E(N)$.  Thus, we expect that deviations of the Higgs field from the constant will be suppressed exponentially for large enough $\kappa$. 

The previous discussion shows one important concept. When dealing with a pure gauge field, one only needs to considering fluctuations of the edges $\sigma_{e}$ from the identity. This leads to an understanding of Wilson loops expectations via the associated plaquette computations. Here, the fluctuations comes in two ways; as before, we must still consider fluctuations of edges from the low energy configurations, but we must further consider the fluctuations of the Higgs field configurations from the identity.

\subsubsection{Wilson Loop Action}

As is standard, we want to compute the expected value of the Wilson loop action as the size of the lattice goes to infinity. We consider a closed non-intersecting loop $\gamma$ in $\Lambda_N$ consisting of edges $e_1, e_2,\ldots, e_m$.

We can define the Wilson loop action as on a configuration $\calC = (\phi_v, \sigma_e)$ as
\begin{equation}
    W_{\gamma}(\calC) =  \sum_{i=1}^m \rho(\sigma_{e}) = \rho( \sum_{i=1}^m \sigma_e).
\end{equation}
We use the notation $\langle \gamma,\sigma_e \rangle = \sum_{i=1}^m \sigma_e$.

Recall that on abelian groups we can consider $\sigma_e$ to be a one-form supported on the edges, $E_N$; for reference, see the discussion in Section 3 of \cite{SC20}.  Thus, we can interpret $\langle \gamma, \sigma \rangle$ as an integral on the set of $1$-forms. More importantly, we can apply Stokes' theorem and can find a surface $q$ such that $\delta q = \gamma$ and $(q)_p \in \{-1,0,1\}$ for all plaquettes $p \in P_N$ such that
\begin{equation}
    \langle \gamma, \sigma\rangle = \sum_{p \in P_N} q_p (\td \sigma)_p.
\end{equation}

When we write the Wilson loop action in terms of the integration of a two-form over a surface, we see that we are able directly see the effect of nontrivial plaquettes those with $(\td \sigma)_p \ne 1$; later, we will see that this decomposition will allow us to determine whether some excitation of the Higgs and gauge field $\calC$ would be independent of the Wilson action on the loop.

In the next section, we will start rigorously describing some notions we can use to characterize fluctuations of the Higgs field and the gauge field.

\subsection{Rigorous Definitions}

There is a trivial symmetry that preserves the Hamiltonian; we flip the signs of all the Higgs boson configurations ( '+' to '-' and '-' to '+'), while the gauge field configurations $\sigma_e$ are unchanged. For some simplicity in the proof, we may assume that $N$ is odd. When $N$ is odd, we can always apply a global flip of the Higgs boson charges to get a unique configuration satisfying the condition that $\sum_{v \in V_N} \phi_v > 0$. (All this condition means is that the majority of charges are $+$ instead of $-$.) At this point, we can start giving definitions of the types of excitations we will consider so that we can consider a cluster expansion.

\begin{defn} \label{def:supportlowdisorder}[Support of Configurations]


{

We define $\calC_e$, our excited edges, as
$$
\calC_e=\{e=(v,w) \in E_N: \phi_v \ne \phi_w\} \cup \{e \in E_N: \sigma_e = -1 \}. 
$$

The support of our configuration is 
\begin{equation}
    \text{supp}(\calC) =\{p \in P_N: \exists e \in \calC_e \text{ s.t. } e \in \delta p \}.
\end{equation}
}

\end{defn}

\begin{rmk}
Without the Higgs field action, an edge with $\sigma_e \ne 1$ is not necessarily excited. For example, one could find a plaquette whose boundary edges are all $-1$. This plaquette would not be of lower energy in the Hamiltonian and we would not consider this plaquette to be excited. By contrast, with the Higgs action, a single edge with $\sigma_e =-1$ drives it to lower energy.


\end{rmk}

\begin{lem}
Consider a configuration $\calC=(\sigma_e,\phi_v)$
If $p$ is a plaquette such that $(\td \sigma)_p \ne 1 $, then $p \in \supp(\calC)$.
\end{lem}
\begin{proof}
Clearly, if $(\td \sigma)_p \ne 1$, then there clearly is some edge $e$ in the boundary of $p$ such that $\sigma_e \ne 1$. Thus, the only non-trivial computations to the probability value from our clusters come from what we define to be the support of our distribution.
\end{proof}

We will now define the function $\Phi(P)$ where $P$ is a set of plaquettes in $P_N$; this is a crucial part of our cluster expansion. 

\begin{defn}
We let $E(P)$ denote the set of edges $e$ such that there exists a plaquette $p\in P$ with $e \in \delta p$. Similarly, we let $V(P)$ denote the set of vertices such that there is a plaquette $p \in P$ such that $v$ is a boundary vertex of $p$.
\begin{equation}
\begin{aligned}
    \Phi(P) = \sum_{\substack{\calC= (\sigma_e, \phi_v) \\ \supp(\calC) = P \\ \sum_v \phi_v > 0}} &\prod_{p \in \supp(\calC)}  \exp[\beta (\rho((\td \sigma)_p) - \rho(1))]\\
    &\prod_{\substack{e \in E(P)\\ e = (v,w)}} \exp[\kappa (f(\sigma_e, \phi_v \phi_w^{
    -1})-f(1,1))].
\end{aligned}    
\end{equation}

We also have a similar cluster expansion formula incorporating the effects of the Wilson action.
\begin{equation}
\begin{aligned}
    \Phi_{W}(P) = &\sum_{\substack{\calC= (\sigma_e, \phi_v) \\ \supp(\calC) = P \\ \sum_v \phi_v > 0}} \prod_{p \in \supp(\calC)}  \exp[\beta( \rho((\td \sigma)_p)-\rho(1))]\\
    & \times \prod_{\substack{e \in E(P)\\ e = (v,w)}} \exp[\kappa (f(\sigma_e, \phi_v \phi_w^{-1})- f(1,1))] \prod_{e \in E(P) \cap \gamma} \rho(\sigma_e).
\end{aligned}
\end{equation}
\end{defn}
\begin{rmk}
Though after this section, we will not use this specific definition of $\Phi(P)$, we will use $V(P)$ and $E(P)$ frequently throughout the course of this paper.
\end{rmk}

Our second definition gives us criterion to determine whether our clusters would interact with each other, or could be considered to split from each other.
\begin{defn} \label{def:graphg2}
We define the following graph $G_2$ with vertex set $P_N$ as follows. We say that there is an edge between two plaquettes $p_1$ and $p_2$ if there exists a $3$-cell such that both $p_1$ and $p_2$ are on the boundary of said $3$-cell. 

We call a set of plaquettes $V$ a vortex if the plaquettes of $V$ form a connected set in the graph $G_2$. 

Furthermore, we say that two sets of plaquettes $P_1$ and $P_2$ are compatible if the set $P_1$ and $P_2$ are disconnected in $G_2$. Otherwise, we say that $P_1$ and $P_2$ are incompatible.

The decomposition of some plaquette set $P$ into its maximal connected subcomponents $P = V_1 \cup V_2 \ldots \cup V_N$ is called a vortex decomposition of $V$. It is clear to see that $V_i$ and $V_j$ are compatible for any distinct pair $i$ and $j$.

\end{defn}

We have the following Lemma that allows us compute the function $\Phi(P)$ in terms of $\Phi$ evaluated on the elements of its vortex decomposition.
\begin{lem} \label{lem:compsplit}
Let $P_1$ and $P_2$ be two compatible sets of plaquettes. Then,
\begin{equation}
    \Phi(P_1 \cup P_2) = \Phi(P_1) \Phi(P_2).
\end{equation}
As a consequence, if $V_1 \cup V_2 \ldots \cup V_N$ is a vortex decomposition of $P$, then we have,
\begin{equation}
    \Phi(V_1 \cup V_2 \ldots \cup V_N) = \Phi(V_1) \Phi(V_2)\ldots \Phi(V_N).
\end{equation}
The function $\Phi_W$ would satisfy a similar property; namely,
\begin{equation}
    \Phi_W(V_1 \cup V_2 \ldots \cup V_N) = \Phi_W(V_1) \Phi_W(V_2) \ldots \Phi_W(V_N).
\end{equation}

\end{lem}
\begin{proof}
 Our goal is to find a bijection $\calC=(\sigma_e,\phi_v) \to (\calC_1= (\sigma^1_e,\phi^1_v), \calC_2= (\sigma^2_e,\phi^2_v))$ where $\calC$ is a configuration whose support is $P_1 \cup P_2$ and $\calC_i$ is a configuration whose support is $P_i$ that satisfies the following equation.
 
 \begin{equation} \label{eq:multequation}
 \begin{aligned}
   & \prod_{p \in \calC}  \exp[\beta (\rho((\td \sigma)_p - \rho(1))]\prod_{\substack{e \in E(P_1 \cup P_2) \\ e = (v,w)}} \exp[\kappa (f(\sigma_e, \phi_v \phi_w^{-1})-f(1,1))] = \\
   &\prod_{p_1 \in \calC_1}  \exp[\beta (\rho((\td \sigma)_{p_1}- \rho(1))]\prod_{\substack{e \in E(P_1)\\ e_1 = (v_1,w_1)}} \exp[\kappa (f(\sigma_{e_1}, \phi_{v_1} \phi_{w_1}^{-1}) - f(1,1))]   \\
   & \times \prod_{p_2 \in \calC_2}  \exp[\beta (\rho((\td \sigma)_{p_2})-\rho(1))]\prod_{\substack{e_2 \in E(P_2)\\ e_2 = (v_2,w_2)}} \exp[\kappa (f(\sigma_{e_2}, \phi_{v_2} \phi_{w_2}^{-1}) - f(1,1))].
\end{aligned}
\end{equation}

We will assume that $P_2$ is a vortex for simplicity. 

Recall that $V(P_1)$ and $V(P_2)$(resp. $E(P_1)$ and $E(P_2)$) denote the set of vertices(resp. edges) that form a boundary vertex(resp. edge) of some plaquette in $P_1$ and $P_2$. 

Consider a configuration $\calC$ whose support is $P_1 \cup P_2$. Let $e$ be an edge in the complement of $E(P_1 \cup P_2)$. Then, $\sigma_e=1$ for this configuration; otherwise, if $\sigma_e \ne 1$, then a plaquette that contains $e$ as a boundary edge would be in the support of $\calC$ by our definition of support. This would imply $e \in E(P_1 \cup P_2)$, which is a contradiction.

Now consider an edge $e \in E(P_1) \cap E(P_2)$. If $\sigma_e \ne 1$, then all the plaquettes $p$ that contain $e$ as a boundary edge are in $\supp(\calC)$ and are connected to each other in $G_2$. However, since $e$ is in $E(P_1)$, at least one of these plaquettes must be in $P_1$. For the same reason, one of these plaquettes must also be in $P_2$. However, this implies that $P_1$ and $P_2$ are connected to each other; this is a contradiction.

Thus, there is a simple way to describe the gauge field configurations $\sigma_e$ for the configurations $\calC_1$ and $\calC_2$. For edges $e$ in $E(P_1)$, we will set $\sigma^1_e = \sigma_e $, $\sigma^2_e = 1$. for edges $e$ in $E(P_2)$, we will set $\sigma^1_e = 1$ and $\sigma^2_e = \sigma_e$. For edges $e$ not in either $E(P_1)$ or $E(P_2)$, we merely set $\sigma_e =1$. Note that this is well defined since for edges $e \in E(P_1) \cap E(P_2)$, we know that $\sigma_e =1$.


It is more complicated to describe the map on the Higgs field. Just as in the Ising model, one can expect to see islands of charges that include each other. Our procedure for mapping the Higgs field configurations from $\calC$ to $\calC_1$ and $\calC_2$ does not merely involve applying a restriction map as we have done for the edges. Instead, one must apply appropriate charge flips in order to ensure the gauge constraints $\sum \phi^1_v \ge 0, \sum \phi^2_v \ge 0$ are satisfied.

We will describe the problem formally as follows. Divide $V(P_1 \cup P_2)^c$ into connected components as $B_1 \cup B_2 \cup \ldots \cup B_N$ and, similarly,  $V(P_i)^c$ as $B_{1}^i \cup B_2^i\cup \ldots \cup B_{N_1}^i$. We see from Lemma \ref{lem:compcharg} in the Appendix that each set $B_i$ is monocharged. If we propose an charge assignment $\calC_1$ with support $P_1$, then we too must make sure that each of the sets $B_i^1$ are monocharged.

Note that each set $B_i^1$ can be decomposed as follows $B_i^1 \subset V(P_2) \cup B_{i_1} \cup \ldots B_{i_{m_i}}$, where the union is minimal in $|m_i|$. If it were the case that all the sets $B_{i_k}$ were assigned the same charge, e.g. $c_{i_l} = c_{i_1}$ for all $l \in \{1,\ldots,m_{i_1}\}$, then we would be able to safely assign the charge $c_{i_1}$ to all vertices in $B_i^1$ in the configuration $\calC_1$.

The problem occurs when there are two components $B_{i_{k_1}}$ and $B_{i_{k_2}}$ in the decomposition of $B^1_{i}$ that are assigned two different charges $c_{i_{k_1}} \ne c_{i_{k_2}}$. For simplicity of notation, WLOG, we are considering the set $B^1_1$ with two subcomponents $B_{1}$ and $B_2$ of different charges. For simplicity of later description, let us also assume that $B_1$ and $B_2$ do not contain a boundary vertex of $V_N$. 

In the remainder of this proof, we will frequently use notation from Lemma \ref{lem:bndry} in the Appendix. Let $\mathcal{V}(B_1)$ and $\mathcal{V}(B_2)$ be as in Lemma \ref{lem:bndry}, the set of vertices connected to $B_1$ and $B_2$ respectively, having the same charge. From the same Lemma, the exterior boundaries $EB(\mathcal{V}(B_1))$ and $EB(\mathcal{V}(B_2))$ are connected sets. Thus, $EB(\mathcal{V}(B_i))$ must be entirely contained in either $P_1$ or $P_2$. If both of these external boundaries were in $P_1$, then this would imply that $B_1$ and $B_2$ would be in different connected components in the splitting of $B_i^1$. Therefore, it must be the case that one of the external boundaries, say $EB(\mathcal{V}(B_2))$ was in $P_2$.

In this case, what one can do is the following. We can embed our lattice inside $\mathbb{Z}^d$. In the full lattice $\mathbb{Z}^d$, find the set $S$ containing infinity whose unique(internal) boundary is $EB(V_2)$. Now consider the complement of $S$ in $\mathbb{Z}^d$. We see that $\mathcal{V}(B_2) $ is a subset of $S^c$. Now, what one can do is to completely flip the signs of all vertices in $S^c$. Under this transformation, the only edges whose Hamiltonian action value changes are those edges in $\mathcal{E}(S^c)$. However, all of these edges must belong to $E(P_2)$ rather than $E(P_1)$. Thus, this flip applied to $\calC$ will not change the Hamiltonian action when restricted to edges of $E(P_2)^c \cup E(P_1)$. As a consequence of this, if we consider an edge connecting a vertex of $V(P_1)$ to its complement, then this edge will still have its adjacent vertices assigned the same charge  even after this flip. 
Furthermore, this flip changes the charges so that the charges of $B_2$ will match those of its closest neighbors in $S^c$. 

Through a careful ordering, one can come up with a series of charge flips that do not affect the Hamiltonian action's value on $E(P_2)^c \cup E(P_1)$ and will ensure that for each component $B^1_i \subset V(P) \cup B_{i_1} \cup B_{i_2} \cup \ldots \cup B_{i_{i_m}}$ the charges associated to the $B_{i_l}$'s are the same. Intuitively, this procedure involves flipping the outermost island, and then further applying flips to correct the internal islands. At this point, one can assign this charge to all the vertices in $B^1_i$ in $\calC_1$. When restricted to edges in $E(P_2)^c \cup E(P_1)$, the Hamiltonian actions of $\calC_1$ and $\calC$ are the same. If after this procedure $\sum \phi_v <0$, then one can perform a global flip of the Higgs boson charge.

Through this procedure, we get a configuration $\calC_1$ whose Hamiltonian action on the set of edges $E(P_2)^c \cup E(P_1)$ are the same as those of $\calC$. We can similarly construct $\calC_2$. Furthermore, if we consider an edge in $(E(P_1)^c \cup E(P_2)) \cap (E(P_2)^c \cup E(P_1))$, then the action of this edge is trivial; the vertices attached to it are assigned the same color and $\sigma_e=1$.

One can see for this map $\calC \to (\calC_1,\calC_2)$, the desired relationship \eqref{eq:multequation} holds. In addition, given two different configurations $(\calC_1,\calC_2)$ supported on $P_1$ and $P_2$ respectively, one can reverse the construction here to get a configuration $\calC$ whose support is $P_1 \cup P_2$. This procedure is fairly tedious, but is still based on finding the appropriate sets to flip iteratively until one gets the desired configuration. The construction of this bijection completes the proof of Lemma \ref{lem:compsplit}.

\end{proof}

In summary,
the previous lemma has shown that the cluster expansion splits based on a vortex decomposition. We will now present the link between our cluster expansion function $\Phi$ and probabilities under our Hamiltonian. 
\begin{lem}
Let $Z_{\Lambda,\beta}$ be the partition function associated to the Hamiltonian \eqref{eq:toyHam} on the lattice $\Lambda$. We have the following relations,
\begin{equation} \label{eq:partfunct}
     Z_{\Lambda,\beta,\kappa} =  2 \sum_{P \subset P_N} \Phi(P),
\end{equation}
\begin{equation}
    P(\supp(\calC) = P) = \frac{\Phi(P)}{\sum_{P \subset P_N} \Phi(P)}.
\end{equation}
\end{lem}
\begin{proof}
In the sum that appears in the definition of $Z_{\Lambda,\beta,\kappa}$, one can first sum over all subsets $P$ that could possibly be a support of the configuration, and then over all configurations $\calC$ that have $P$ as its support. After possibly applying a trivial global flip to ensure that $\sum_{v \in V_N} \phi_v >0$, the first equation is merely the definition of $\Phi(P)$. The global flip is the reason for the existence of the factor of $2$ on the right hand side of equation \eqref{eq:partfunct}. The second equation comes from the definition of the probability distribution and the previous equation \eqref{eq:partfunct} on $Z_{\Lambda,\beta,\kappa}$. Note that we have implicitly canceled a factor of $2$ that comes from the global flip symmetry.
\end{proof}

\subsection{Analysis of the Wilson Loop Action}

Our ultimate goal in this section is to argue that the main contribution to the Wilson loop expectations come from minimal vortices, $P(e)$.

\begin{defn}[Minimal Vortices]
Minimal vortices $P(e)$ have the following structure: they consist of the edge $e$ and the $2(d-1)$ plaquettes that use the edge $e$ as a boundary edge. This should be thought of as the smallest excitations, as the easiest way to create a minimal vortex is to excite the single edge $\sigma_e \ne 1$ and to set all other surrounding edges, $e'$, to $\sigma_{e'}=1$. 

\end{defn}

\begin{rmk} \label{rmk:minimal}
Minimal vortices will be a common structure appearing in the previous works \cite{SC20} and \cite{Viklund} and will appear in later sections as well. We remark that in this section, that a configuration $\calC$ that has a minimal vortex $P(e)$ in its support must assign all the vertices of $V(P(e))$ the same Higgs boson charge. Otherwise, one would be able to find at least $2$ edges in $V(P(e))$ connecting opposite charges. One of these edges will be $e' \ne e$. Under our definition of support, all of the plaquettes that use $e'$ as a boundary vertex will be in the support of $\calC$. This means that $P(e)$ cannot be a vortex in the support of $\calC$. Furthermore, all edges $e'$ in $E(P(e))$ aside from $e$ itself must be assigned $\sigma_{e'}=1$.
\end{rmk}

We let $W_{\gamma}$ represent our Wilson loop action on the loop $\gamma$. In addition, we let $S_\gamma$ be a surface whose boundary is $\gamma$.

The analysis in this section depends on describing which plaquette configurations form the main contribution to the Wilson loop equation and bounding away the contributions of rare configurations by the function $\Phi$ constructed earlier. 

\subsubsection{ Trivial Configurations}

One difference from the case of a pure gauge field is that there are excited configuration $\calC$ whose support would intersect the loop $\gamma$, but would not have any contribution to the Wilson loop action. In this subsection, we find some we give some conditions on the configuration which will ensure that the configuration will not have any contribution to the Wilson loop action. By `removing' the contribution of these trivial configurations, we can improve our error analysis later. 


Our first definition describes a geometric condition for a vortex $V$ to have no contribution to the Wilson loop action.
\begin{defn}\label{def:non-contributing}
We call a vortex $V$ non-contributing to the Wilson loop action $\gamma$ if it satisfies one of the following two properties: \begin{itemize}
    \item 
    There is no edge $e \in \gamma$ that is simultaneously a boundary edge of some plaquette $p \in V$.
    \item It is a minimal vortex not centered around an edge $e \in \gamma$.
\end{itemize} 
\end{defn}

We will show later that if a vortex is non-contributing then, as its name states, it will not have a contribution to the Wilson loop action.

To appropriately account for the Wilson loop contribution in orders of $e^{-\beta}$, we need to describe other configurations $\calC$ that are not immediately removed due to the geometric condition above. In order to do this,
we start by first describing a modification of $\phi$ that only takes into account the 'nontrivial' plaquette excitations with respect to the Wilson loop functional.

\begin{defn} \label{def:WilsonLoopNontrivial}
We call a configuration $\calC = (\phi_v,\sigma_e)$ Wilson loop non-trivial if there is a plaquette $p$ in $\supp \text{ } \calC$ such that there is a boundary edge in $p$ that intersects the loop $\gamma$ and there is another plaquette $p'$ in $\supp \text{ } \calC$  with $(\td\sigma)_{p'} \ne 1$. A configuration that is not Wilson loop non-trivial is Wilson loop trivial.

We define the functional $\phi_{NT}(V)$ which takes values on sets of plaquettes as follows,
\begin{equation}
\begin{aligned}
\Phi_{NT}(P) = &\sum_{\substack{\calC= (\sigma_e, \phi_v) \\ \supp(\calC) = P \\ \sum_v \phi_v > 0\\ \calC \text{ is Wilson Loop Nontrivial}}} \prod_{p \in \supp(\calC)}  \exp[\beta (\rho((\td \sigma)_p) - \rho(1))]\\
& \times \prod_{\substack{e \in E(P)\\ e = (v,w)}} \exp[\kappa (f(\sigma_e, \phi_v \phi_w^{-1}) - f(1,1))].
\end{aligned}
\end{equation}

\end{defn}

  The purpose of introducing the above Definition \ref{def:WilsonLoopNontrivial} is to separate out those configurations which will have a non-trivial effect on the Wilson loop action $W_{\gamma}$.  Namely, there needs to be some plaquette with $(\td \sigma)_p \ne 1$ in order to ensure that $W_{\gamma} \ne 1$.

\begin{rmk}
We remark that if there is a configuration $\calC$ with support $P_1 \cup P_2$ with $P_1$, $P_2$ compatible such that $\calC$ is Wilson loop non-trivial when restricted to $P_2$, then after the splitting construction in Lemma \ref{lem:compsplit}, then the configuration $\calC_2$ is Wilson loop non-trivial on $P_2$. This is due to the fact that our construction is merely a restriction map applied to the gauge field configuration $\sigma$. This allows us to establish multiplicative analogues of our result in \ref{lem:compsplit} with $\Phi_{NT}$. For example, the sum of $\exp[H_{N,\beta,\kappa}(\sigma,\phi)]$ over configurations with support $P_1 \cup P_2$  and are Wilson loop non-trivial when restricted to $P_2$ is $\Phi(P_1)\Phi_{NT}(P_2)$.
\end{rmk}

The following lemma takes into account both the sources of trivialities to the Wilson loop action.

\begin{lem}
Consider a vortex $V$ whose associated edge set $E(V)$ has no edge in common with $\gamma$. Then,
\begin{equation} \label{eq:firstcondition}
    \mathbb{E}[ W_{\gamma} | \supp \text{ } \calC = V] =1.
\end{equation}

For similar reasons, we also have the following statement.
\begin{equation} \label{eq:secondcondition}
    \mathbb{E}[W_{\gamma}|  \calC \text{ is  Wilson Loop Trivial}] =1.
\end{equation}
\end{lem}
\begin{proof}

Due to our definition of support, it is clear that if $V$ is a vortex whose edge set $E(V)$ shares no edge with $\gamma$, then $\sigma_e =1$ for all edges $e \in \gamma$. Otherwise, if $\sigma_e \ne 1$ for an edge in $\gamma$, then the plaquettes surrounding that edge would be in the support of our configuration. This would imply that $E(V)$ contains an edge of $\gamma$. Finally, it is clear that if $\sigma_e=1$ for all edges of $\gamma$, then the Wilson loop action $W_{\gamma}$ is clearly 1. This deals with the first case of non-contributing vortex from Definition \ref{def:non-contributing}.

The second case is that the non-contributing vortex is minimal $P(e)$, but the central edge $e$ is not in $\gamma$. From our discussion in \ref{rmk:minimal}, we know that all edges $e' \in E(P(e))$ that are not $e$ itself are set to $1$. Again, this would imply that $\sigma_e$ is $1$ along all edges of $\gamma$ and, furthermore, the Wilson loop action is $1$.
This deals with the case of \eqref{eq:firstcondition}.

Now, we deal with \eqref{eq:secondcondition}. Note that if a configuration $\calC$ is Wilson loop trivial then either $E(\supp \text{ } \calC)$ will not contain an edge of $\gamma$ or $(\td \sigma)_p=1$ for all plaquettes $p \in P_N$. We have already dealt with the former possibility in the last paragraph. We now deal with the second possibility.

A benefit of the abelian case is that we can rewrite our Wilson loop expectation as the integral of $\gamma$ with respect to a $1$-from $\sigma$. As we have mentioned earlier, Stokes' theorem allows us to express this Wilson loop integral as the integral of a surface $S_{\gamma}$ whose boundary is $\gamma$ with respect to the $2$-from $(\td \sigma)$. Now, if $(\td \sigma)_p=1$ for all $p \in P_N$, then clearly the integral $\langle S_{\gamma}, \td \sigma \rangle =1$. Thus, the Wilson loop action is still $1$. This completes the proof of equation \eqref{eq:secondcondition}.

\end{proof}

\subsubsection{The main term contribution}


 The following lemma describes the contribution to the Wilson loop action from the minimal vortices. These should be the simplest excitations.
\begin{lem} \label{lem:minimalCont}

    Consider an excitation $\calC$ and
    let 
    $\supp \text{ }\calC = V_1 \cup V_2 \cup \ldots V_N$ be the vortex decomposition of $\supp \text{ } \calC$.
    
    We say that $\calC$ belongs to the event $E$ if satisfies certain properties regarding the vortices in the support. To be in the set $E$,
     at least one of the following holds for each $i$:
    \begin{itemize}
    \item The vortex $V_i$ is minimal and is centered around an edge $e$ in $\gamma$ .
    \item $V_i$ is non-contributing to the Wilson loop action.
    \item $V_i$ is not non-contributing to the Wilson loop action, but $\calC$ restricted to $V_i$ is Wilson loop trivial.
    \end{itemize}
    
    Restricted to the event $E$, we have
    \begin{equation} \label{eq:restcomp}
        \mathbb{E}[W_\gamma| E] = \rho(-1)^{M_{\gamma}},
    \end{equation}
    where $M_{\gamma}$ is the number of minimal vortices centered along an edge of $\gamma$.
\end{lem}

\begin{proof}

By definition, 
$$
\mathbb{E}[W_\gamma| \supp(\calC) = P ] = \frac{\Phi_W(P)}{\Phi(P)},
$$.

Consider a decomposition into compatible sets of plaquettes of the form,
$\mathcal{V}_1 \cup \mathcal{V}_2 \cup \mathcal{V}_3$, where each $\mathcal{V}_i$ has the vortex decomposition
$\mathcal{V}_1 = V^1_1 \cup \ldots \cup V^1_a$,
$\mathcal{V}_2 = V^2_1 \cup \ldots \cup V^2_b$,
$\mathcal{V}_3 = V^3_1 \cup \ldots \cup V^3_c$.

Define $E_{\mathcal{V}_1,\mathcal{V}_2,\mathcal{V}_3}$ to be the set of configurations $\calC$ on which $\mathcal{V}_1$ is non-contributing $\mathcal{V}_2$ are minimal and centered on some edge of $\gamma$, and $\mathcal{V}_3$ are not non-contributing, but $\calC$ is trivial on $\mathcal{V}_3$. Assume that $E_{\mathcal{V}_1,\mathcal{V}_2,\mathcal{V}_3}$ is non-empty.

We explicitly see that,
\begin{equation}
\begin{aligned}
    \mathbb{E}[W_{\gamma}|\calC \in E_{\mathcal{V}_1, \mathcal{V}_2,\mathcal{V}_3}] &= \frac{\prod_{i=1}^a\Phi(V^1_i) \prod_{j=1}^b \Phi_W(V^2_j) \prod_{k=1}^c [\Phi(V^3_k)-\Phi_{NT}(V^3_k)]}{\prod_{i=1}^a\Phi(V^1_i) \prod_{j=1}^b \Phi(V^2_j) \prod_{k=1}^c [\Phi(V^3_k)-\Phi_{NT}(V^3_k)]}\\
    &= \rho(-1)^{|\mathcal{V}_2|}.
\end{aligned}
\end{equation}

Noting that $E = \cup_{\substack{\mathcal{V}_1, \mathcal{V}_2, \mathcal{V}_3\\\text{all compatible}}} E_{\mathcal{V}_1,\mathcal{V}_2,\mathcal{V}_3}$, a disjoint union, this gives the result \eqref{eq:restcomp} after a removal of the conditioning.

\end{proof}

\subsubsection{Bounding the probability of large vortices}

The following lemma bounds the probability of observing a vortex $V$ in the support by the functions $\Phi$ we have defined earlier.

\begin{lem}\label{lm:vortex}
Let $V$ be a vortex of cardinality $|V|=m$, where the cardinality counts the number of plaquettes in $V$. Then, we have the following equation,
\begin{equation}
    \mathbb{P}( V \text{is a vortex } \in \supp \text{  }  \calC) \le \Phi(V).
\end{equation}

Similarly,

\begin{equation}
    \mathbb{P}(V \text{ is a vortex } \in \supp \text{ } \calC, \calC \text{ is non-trivial on }V) \le \Phi_{NT}(V).
\end{equation}

\end{lem}

\begin{proof}
Let $\mathcal{P}_V$ be the set of possible supports of configurations $\calC$ that include  $V$ in the vortex decomposition.  If $P$ is a vortex decomposition $\mathcal{P}_V$, we use $P \setminus V$ to denote the remaining vortices. We remark that $P \setminus V$ is still a legitimate vortex decomposition corresponding to some configuration $\calC$. By applying Lemma \ref{lem:compsplit},
we see that,
\begin{equation}
    \sum_{P \in \mathcal{P}_V} \Phi(P) = \sum_{P \in \mathcal{P}_V} \Phi(V) \Phi(P \setminus V).
\end{equation}
Since each vortex decomposition $p \setminus V$ can be constructed to be the support of some configuration $Z$, we see that we immediately have the inequality,
\begin{equation}
    Z \ge  2\sum_{P \in \mathcal{P}_V } \Phi(P \setminus V).
\end{equation}

Thus, 
\begin{equation}
\begin{aligned}
    \mathbb{P}(V \text{ is a vortex of } \supp\text{ } \calC)& = \frac{2\sum_{P \in \mathcal{P}_V} \Phi(P)}{ Z}\\
    &\le \Phi(V) \frac{\sum_{P \in \mathcal{P}_V } \Phi(P \setminus V)}{\sum_{P \in \mathcal{P}_V } \Phi(P \setminus V)} = \Phi(V),
\end{aligned}
\end{equation}
as desired.

For the second part of the theorem, one can see that we can count the contribution of excitations whose support contains $V$ but is non-trivial on $V$ as
\begin{equation}
   2 \Phi_{NT}(V) \sum_{P \in P_V} \Phi(P \setminus V).
\end{equation}
This is due to the fact that the construction in Lemma \ref{lem:compsplit} merely projects the value of $\sigma_e$ to its appropriate subset. Thus, if $(\td \sigma)_p \ne 1$ for some plaquette, it will hold true for the appropriate projection. Since we assumed $\calC$ was nontrivial on $V$, the same must hold true for the projection to $V$.

\end{proof}


Our next lemma uses our previous lemma to prove that large vortices are exponentially suppressed in probability.

\begin{lem} \label{lem:boundlarge}

Define the constant $C_1$ as
\begin{equation}
\begin{aligned}
    C_1:= 2^8  \max\bigg(&\exp[-2 \beta \text{Re}( \rho(1) - \rho(-1))],\\
    &\exp[-\frac{\kappa}{(d-1)}  (\text{Re}f(1,1) - \max_{(a,b) \ne(1,1)}  \text{Re}f(a,b))] ]\bigg).
\end{aligned}
\end{equation}

Let $V$ be a vortex with $2k$ oriented plaquettes (alternatively, $k$ pairs of plaquettes $p,-p$), then we have the following bound.
\begin{equation} \label{eq:bndsinglevortex}
    \phi(V) \le (C_1)^{k},
\end{equation}

As a consequence of this bound and the previous Lemma \ref{lm:vortex}, the probability that there is a configuration $\calC$ whose support contains a vortex of size $k$ that is not non-contributing to the Wilson loop action is bounded above by
\begin{equation} \label{eq:boundovergamma}
    2(d-1) |\gamma| (C(d) C_1)^k,
\end{equation}
where $C(d)$ is a constant that only depends on the dimension $d$.

In the case that,
$$
2\beta (\rho(1) - \rho(-1)) > \frac{\kappa}{(d-1)} [\text{Re}f(1,1) - \max_{(a,b) \ne (1,1)} \text{Re} f(a,b)],
$$
then we have a better estimate when we consider Wilson loop non-trivial configurations.

\begin{equation}
    \Phi_{NT}(V) \le \exp[-4(d-1) \beta \text{Re}[\rho(1) - \rho(-1)]] (C_1)^{k-2(d-1)}.
\end{equation}
Similarly,
the probability that there is a configuration $\calC$ such that there is a vortex $V$ of size $2k$ that is not non-contributing to the Wilson loop and such that $\calC$ is nontrivial when restricted to $V$ is bounded above by 
\begin{equation} \label{eq:probkbound}
     2(d-1) |\gamma|( C(d)  C_1)^{k-2(d-1)} (2^8 C(d)  \exp[-2\beta\text{Re}(\rho(1) - \rho(-1))])^{2(d-1)}.
\end{equation}

\end{lem}
\begin{proof}

Let $\calC$ be a configuration with support $V= \{P_1,P_2,\ldots,P_k\}$. 

If $2k_1$ of these plaquettes $p$ satisfy  the property that $(\td \sigma)_p \ne 1$ (here, we implicitly include a plaquette and its negative), then the probability of this configuration is immediately reduced by $\exp[-k_1 2\beta \text{Re} (\rho(1) - \rho(-1)) ]$. 

Now, if the plaquette $p$ does not satisfy $(\td \sigma)_p=1$, then necessarily, either one of the edges on its boundary satisfies $\sigma_e =-1$ or there is an edge $e=(v,w)$ on its boundary with $\phi_v\ne \phi_w$.
In any case, the presence of such an edge (and thus also its negative) reduces the probability of the configuration by a factor of,
\begin{equation}
    \exp[-2 \kappa (\text{Re}f(1,1) - \max_{(a,b) \ne (1,1)} \text{Re} f(a,b))].
\end{equation}

Now a single edge can be the cause of the excitation of at most $4(d-1)$ oriented  plaquettes. Thus, there must be at least $\frac{2(k-k_1)}{4(d-1)}$ excited oriented edges pairs $e$ and $-e$.

Thus, the probability of seeing this configuration is at most
\begin{equation}
    \exp[- k_1 2\beta \text{Re}( \rho(1) - \rho(-1)) - \frac{k-k_1}{(d-1)} \kappa (\text{Re}f(1,1) - \max_{(a,b) \ne(1,1)} \text{Re}f(a,b))].
\end{equation}

At this point, we now need to calculate the number of configurations with support $V$. Note that there are at most $4k$ vertices that are boundary vertices of plaquettes in $V$. These are the only vertices that could be assigned the Higgs boson charge $-1$. There will be at most $2^{4k}$ ways to assign the Higgs boson charge. Similarly, there are $2^{4k}$ ways to assign the value of $\sigma_e$ at the edges that bound some plaquette in the support. 

Finally, we must find a way to bound the number of ways to get a connected set of $k$ plaquettes that would intersect one of the edges of $\gamma$.

Lemma 3.4.2 of \cite{SC20} gives a useful structure theorem. In dimension $4$, there are at most $(20 e)^k$ vortices of size $2k$ that could contain any given plaquette $p$. 
A modification to general dimension will  show that there is some constant $C(d)$ such that the number of vortices of size $2k$ that contain any given plaquette is $(C(d))^k$. Here, $C(d)e^{-1}$ is the number of plaquettes that are adjacent to any given plaquette in the adjacency graph of plaquette $G_2$ from Definition \ref{def:graphg2}.

The logic applied is similar. Each vortex is a connected subgraph of $G_2$ and thus has a connected spanning tree. The problem of counting the number of vortices of size $2k$ can now be divided into two parts. The first part is to count the number of non-isomorphic  rooted spanning trees of size $k$. Now given a spanning tree $T$, the second part is to count the number of ways to embed this spanning tree in $G_2$ with the root of the spanning three the same as our special plaquette $p$. Since no two vortices of size $k$ can have the same spanning tree when embedded in $G_2$, our counting procedure will clearly be an upper bound on the number of vortices containing $p$.

The number of non-isomorphic spanning trees of size $k$ can be bounded by $k (k!)^{-1}k^{k-2}$. To get this, we divide the number of labeled spanning trees from Cayley's Theorem by the number of permutations of the labels. We also have $k$ choices for the root. This is less than $e^k$ and is good enough for our purposes.

Now to embed the graph in $G_2$, we use the following graph. Assume that we have already embedded the vertex $v$ in $G_2$. Let $w$ be a neighbor of $v$ that we have not yet fixed. There are $C(d)e^{-1}$ ways to choose to embed the vertex $w$; it has to be one of the $C(d)e^{-1}$ neighbors of $v$ in $G_2$. If we inductively perform this procedure starting from the root, we see that for every tree $T$, there are at most $(C(d)e^{-1})^k$ ways to embed the spanning tree in $G_2$. Combining our two estimates show that there are at most $(C(d))^k$ different vortices of size $k$ that could contain any given plaquette $p$.

Now, if we consider the $2(d-1)|\gamma|$ possible plaquettes that are adjacent to edges of $\gamma$, then we see that there are at most $2(d-1)|\gamma| (C(d) )^k$ vortices of size $k$ that are not-noncontributing to the Wilson loop action.

By a union bound, the total probability of seeing one of these excitations is bounded from above by,

\begin{equation}
\begin{aligned}
2(d-1) |\gamma| (2^8 C(d) )^k \max_{k_1} &\exp\bigg[- k_1 2\beta\text{Re}(\rho(1) - \rho(-1))\\& - \frac{k-k_1}{(d-1)} \kappa(\text{Re}f(1,1) - \max_{(a,b) \ne (1,1)}\text{Re} f(a,b))\bigg].
\end{aligned}
\end{equation}

If we consider the case that 
\begin{equation} \label{eq:cond}
    \beta\text{Re}(\rho(1) - \rho(-1)) < \frac{\kappa (\text{Re}f(1,1) - \max_{(a,b) \ne (1,1)} \text{Re} f(a,b))}{2(d-1)},
\end{equation}
then the maximum is located at $k_1 = k$. Otherwise, the maximum is located at $k_1 = 0$.

By considering only configurations that are Wilson loop non-trivial, we can get further cancellations in the case detailed in \eqref{eq:cond}. 

A non-trivial configurations necessarily has $k_1>0$. In fact, $k_1\ge 2(d-1)$ for any configuration $\sigma$ such that there exists a plaquette $p$ with $(\td \sigma)_p \ne 0$. This is Lemma 3.4.6 of \cite{SC20} generalized to higher dimensions.

Thus, we could say that the probability of seeing a configuration whose support contains a vortex $V$ of size $k$ that is not non-contributing and  is a vortex of size $k$  and, simultaneously, is non-trivial when restricted to this $V$ is at most,
\begin{equation}
\begin{aligned}
2(d-1) |\gamma| (2^8 C(d) )^k  &\exp\bigg[- 4(d-1) \beta \text{Re}(\rho(1) - \rho(-1)) \\ &- \frac{k-2(d-1)}{(d-1)} \kappa (\text{Re}f(1,1) - \max_{(a,b) \ne (1,1)} \text{Re} f(a,b))\bigg].
\end{aligned}
\end{equation}
\end{proof}

As a consequence, we can now state the main result of our section. This theorem shows that the main contribution to the Wilson loop expectation comes from the minimal vortices we considers.
\begin{thm} \label{thm:reductiontominimal} Recall the constants $C(d)$ and $C_1$ from our previous Lemma \ref{lem:boundlarge}. Assume that $C(d)C_1<1$ and 
define the constant $D$ as
\begin{equation}
\begin{aligned}
    D:=\max\bigg[& \frac{e^{-2\beta\text{Re}(\rho(1) - \rho(-1))}}{1- 2^8  C(d)  e^{-2\beta\text{Re}(\rho(1)- \rho(-1))}},\\& \frac{e^{-\frac{\kappa}{(d-1)}(\text{Re}f(1,1) - \max_{(a,b) \ne (1,1)} \text{Re} f(a,b))}}{1 - 2^8  C(d)e^{-\frac{\kappa}{(d-1)}(\text{Re}f(1,1) - \max_{(a,b) \ne (1,1)} \text{Re}f(a,b))}} \bigg] .
\end{aligned}
\end{equation}
 We have the following bounds on the expectation of the Wilson loop action,

\begin{equation}
\begin{aligned}
    &|\mathbb{E}[W_{\gamma}] - \mathbb{E}[\rho(-1)^{M_{\gamma}}]| \le\\&
    \hspace{1 cm} 2(d-1)|\gamma|  (2^8 C(d) )^{2(d-1)+1} \exp[-4(d-1) \beta \text{Re} (\rho(1) - \rho(-1)]D, 
    \end{aligned}
\end{equation}
where $M_{\gamma}$ is the number of minimal vortices that are centered around an edge $e$ of $\gamma$.
\end{thm}
\begin{proof}
Recall the event $E$ from Lemma \ref{lem:minimalCont}.

The result of Lemma \ref{lem:boundlarge} shows that the probability of the event $E^c$ is bounded by the union of the events described by probability \eqref{eq:probkbound} for $k$ varying between 7 and $\infty$. Adding all the terms together and performing the union bound gives the desired result.
\end{proof}

\subsection{Approximating $M_{\gamma}$ as a Poisson Random Variable}
\label{sec:PoissonRV}
We can write $M_{\gamma}$ as the sum of random events,
\begin{equation}
    M_{\gamma}:= \sum_{e \in \gamma} \mathbbm{1}[F_{P(e)}],
\end{equation}
where $F_{P(e)}$ is the event that $e$ is the center of a minimal vortex $P(e)$ in the support of our configuration.

To show $M_{\gamma}$ behaves approximately as a Poisson random variable, we would want to show that for $e \ne e'$, the events $F_{P(e)}$ and $F_{P(e')}$ are approximately independent from each other. For each edge $e$, we let $\mathcal{B}_e$ to be the set of edges $e'$ in $\gamma$ such that a minimal vortex centered around an edge $e'$ connect to the minimal vortex centered around $e$ based on adjacency in the graph $G_2$ of plaquettes from Definition \ref{def:graphg2}.

The following Theorem from \cite{Chen} details exactly the degree of approximation of some some of random variables to the Poisson distribution.

\begin{thm} \label{thn:PoisApproxGen}
    Consider the following constants.
    \begin{equation}
    \begin{aligned}
       & b_1:= \sum_{e \in \gamma} \sum_{e' \in B_e} \mathbb{P}(F_{P(e)}) \mathbb{P}(F_{P(e')}),\\
    &    b_2:= \sum_{e \in \gamma} \sum_{e' \in B_e \setminus e} \mathbb{P}(\mathbbm{1}(F_{P(e)}) \mathbbm{1}(F_{P(e')})),\\
    & b_3:= \sum_{\gamma} \mathbb{E}\left[| \mathbb{E}[\mathbbm{1}(F_{P(e)})| \mathbbm{1}(F_{P(e')}) e' \not \in B_e] -\mathbb{P}(F_{P(e)})| \right].
    \end{aligned}
    \end{equation}
    Let $\mathcal{L}(M_{\gamma})$ denote the law of $M_{\gamma}$ and let $\lambda= \mathbb{E}[M_{\gamma}]$. Then 
    \begin{equation}
    d_{
    TV}(\mathcal{L}(M_{\gamma}), \text{Poisson}(\lambda))\le min(1, \lambda^{-1})(b_1 + b_2) +\min(1, 1.4 \lambda^{-1/2})b_3.
    \end{equation}
\end{thm}

We start with the following corollary of the proof of Lemma \ref{lm:vortex}.

\begin{col}[of Lemma \ref{lm:vortex}] \label{col:boundcondition}
Let $E_1$ and $E_2$ be two set of edges. Let $M(E_1,E_2)$ be the event that our configuration has a minimal vortex centered on each edge of $E_1$ and no minimal vortices centered on any edge of $E_2$. Assume that $M(E_1,E_2) >0$ (this amounts to assuming that the minimal vortices centered around $E_1$ do not intersect each other). Now consider some other union of vortices $U$. The probability that some union of vortices $U$ appears in the support of some configuration, conditionally on being in the event $M(E_1,E_2)$ is less than $\Phi(U)$. 
\end{col} 
\begin{proof}
Let $(\sigma,\phi)$ be a configuration with support $U \bigcup_{e_i \in E_1} P(e_i) \cup R $ where $P(e_i)$ is the minimal vortex centered around $e_i$. In addition, $R$ is some union of vortices that do not intersect $U$ or $P(e_i)$ and does not contain a minimal vortex centered at an edge of $E_2$.

The sum of $\exp[H_{N,\beta,\kappa}(\sigma,\phi)]$ of all configurations with this support is given by $2\Phi(U) \prod_{e_i \in E_1} \Phi(e_i) \Phi(R)$.

Thus, the sum of $\exp[H_{N,\beta,\kappa}(\sigma,\phi)]$ over all configurations such that the support contains $U$ will be
\begin{equation} \label{eq:numer}
2\Phi(U) \prod_{e_i \in E_1} \Phi(e_i)\sum_{R} \Phi(R),
\end{equation}
where the sum in $R$ goes over all unions of vortices that do not intersect $U$, the $P(e_i)$'s and do not contain a minimal vortex centered around an edge of $e_2$.

All configurations with support $\bigcup_{e_i \in E_1} P(e_i) \cup R$ will be in $M(E_1,E_2)$. 

Thus, the sum of $\exp[H_{N,\beta,\kappa}]$ for all configurations found in $M(E_1,E_2)$ (we will call this $Z(M(E_1,E_2))$ the partition function restricted to the event $M(E_1,E_2)$) will be satisfy the relation
$$
Z(M(E_1,E_2)) \ge 2 \prod_{e_i \in E_1} \Phi(P(e_1)) \sum_{R} \Phi(R).
$$

The probability that restricted to $M(E_1,E_2)$ that $U$ will be in the support is the ratio of the quantity in \eqref{eq:numer} with the partition function $Z(M(E_1,E_2))$. This ratio is clearly less than $\Phi(U)$.
\end{proof}

This is one of our major tools in computing $\mathbb{E}[\mathbbm{1}[F_{P(e)}]| \mathbbm{1}[F_{P(e'}], e' \not \in B_e]$.

\begin{lem} \label{lem:comparwpoisson}
Let $\gamma$ be a loop that has no self-intersection. In addition, recall the dimension dependent constant $C(D)$ and the  constant $D$ from Theorem \ref{thm:reductiontominimal}. If $2(d-1) C(d)  D< 1$,  
we have the following bounds on the constants $b_1,b_2$ and $b_3$ and $\lambda$.
\begin{equation}
\begin{aligned}
    & b_1 \le 8(d-1) |\gamma| \Phi(P(e))^2,\\
    & b_2 =0,\\
    &b_3  \le \frac{2|\gamma|\Phi(P(e))(d-1)  C(d)  D }{1 - 2(d-1)  C(d)  D},\\
    & |\lambda - |\gamma| \Phi(P(e))| \le \frac{2|\gamma|\Phi(P(e))(d-1)  C(d)  D }{1 - 2(d-1)  C(d)  D}.
\end{aligned}    
\end{equation}

As a consequence of Theorem \ref{thn:PoisApproxGen}, we see that 
\begin{equation} \label{eq:totalvarfromPoisson}
\begin{aligned}
    &d_{TV}(\mathcal{L}(M_{\gamma} ), \text{Poisson}( |\gamma| \Phi(P(e)))) \\& \hspace{1 cm}\le 8(d-1) |\gamma|(\Phi(P(e)))^2 + \frac{4|\gamma|\Phi(P(e))(d-1)  C(d)  D }{1 - 2 (d-1)  C(d)  D}.
\end{aligned}    
\end{equation}

\end{lem}
\begin{rmk}
The function $\Phi$ evaluated on a minimal plaquette $P(e)$ does not depend on the minimal plaquette chosen. Thus, in a minor abuse of notation, we treat it like a constant. Explicitly, we may always substitute,\begin{equation}
    \Phi(P(e))= \exp[-4(d-1) \text{Re}(\rho(1) - \rho(-1))] \exp[-2\kappa\text{Re}(f(1,1) -f(-1,1)) ]
\end{equation}
\end{rmk}

\begin{proof}
We start with the most difficult part: computing a bound on $b_3$.
Fix an edge $e$. We now attempt to compute 
$$ \mathbb{E}\left[| \mathbb{E}[\mathbbm{1}(F_{P(e)})| \mathbbm{1}(F_{P(e')}) e' \not \in B_e] - \mathbb{P}(F_{P(e)})|\right],$$ for our fixed edge $e$. Recall the notation $M(E_1,E_2)$ and $Z(M(E_1,E_2))$ from the proof of Corollary \ref{col:boundcondition}. Let $E_1$ be some set of edges in $\gamma \setminus B_e$ and $E_2$ be the remaining edges in $\gamma \setminus B_e$ that are not contained in $E_1$.

We know compute $\mathbb{E}[\mathbbm{1}[F_{P(e)}]|M(E_1,E_2)]$. Events in $M(E_1,E_2)$ can be divided into two parts,
\begin{enumerate}
    \item $G(E_1,E_2)$: These are configuration $(\sigma,I)$ in $M(E_1,E_2)$ whose support has a vortex decomposition $V_1 \cup \ldots \cup V_m$ such that $P(e)$ does not intersect any $V_i$. Thus, $P(e) \bigcup_{i=1}^m V_i$ is a valid vortex decomposition.
    \item $B(E_1,E_2)$: These are configurations $(\sigma,\phi)$ in $M(E_1,E_2)$ whose vortex decomposition $V_1\ldots V_m$ does contain a vortex $V_j$  such that $V_j$ contains a plaquette that is in $P(e)$. Thus, $P(e) \bigcup_{i=1}^m V_i$ is not a valid vortex decomposition.
\end{enumerate}

If we let $Z(G(E_1,E_2))$ and $Z(B(E_1,E_2))$ be the sum of $\exp[H_{N,\beta,\kappa}]$ for configurations in them ( the partition functions), then we see that the probability of seeing $P(e_i)$ as a plaquette in the support conditional on the event $M(E_1,E_2)$ is given by,
\begin{equation}
    \frac{\Phi(V) Z(G(E_1,E_2))}{Z(G(E_1,E_2)) + Z(B(E_1,E_2))}.
\end{equation}
If we show that $Z(B(E_1,E_2)) \le c Z(M(E_1,E_2))$ for some constant $c <1$, then we see that the probability of seeing $P(e_i)$ in the support conditional on $M(E_1,E_2)$ is greater than $(1-c) \Phi(P(e))$ and less than $\Phi(P(e))$.

By applying Corollary \ref{col:boundcondition}, we can bound $\frac{Z(B(E_1,E_2))}{Z(M(E_1,E_2))}$ by the sum of $\Phi(V)$ over all configurations $V$ that contain a plaquette of $P(e_i)$ in its support. We have performed a version of his sum when trying to compute the probability of the event $E^c$ from  \ref{lem:minimalCont}. Recall the bound \eqref{eq:bndsinglevortex} on the probability of observing a vortex excitation of size $2k$. Notice that a minimal vortex has $2(d-1)$ plaquettes attached to it. We can provide a union bound over all of these $2(d-1)$ vortices and all  vortices of size $2k$ that intersect these plaquettes. Recall that there are at most $(C(d))^{k} $ vortices of size $2k$ that contain any given plaquette. 

 We see that the probability of the event $B(E_1,E_2)$ conditioned on $G(E_1,E_2)$ is less than,
\begin{equation}
   2(d-1)\sum_{k=1} C_1^k (C(d))^k \le 2(d-1)  C(d)  D,
\end{equation}
provided $C(d) C_1 <1$.

We thus see that the expectation $\mathbb{E}[\mathbbm{1}[F_{P(e)}]|M(E_1,E_2)]$ is greater than $\frac{\Phi(P(e))}{1- 2(d-1)  C(d) D}$, provided $2(d-1)  C(d)D <1$.

Through the same logic, we can show that $\mathbb{P}(F_{P(e)}) \ge \frac{\Phi(P(e))}{1-2(d-1)  C(d)D}$. Recalling our previous upper bound of $\Phi(P(e))$ on both of these quantities, we see that,
\begin{equation}
    |\mathbb{E}[\mathbbm{1}[F_{P(e)}]| M(E_1,E_2)] - \mathbb{E}[F_{P(e)}]| \le \frac{2\Phi(P(e))(d-1)  C(d)  D }{1 - 2(d-1)  C(d)  D}.
\end{equation}

We can now integrate over our conditioning on $M(E_1,E_2)$ to show
$$\mathbb{E}\left[| \mathbb{E}[\mathbbm{1}(F_{P(e)})| \mathbbm{1}(F_{P(e')}) e' \not \in B_e] - \mathbb{P}(F_{P(e)})|\right] < \frac{2\Phi(P(e))(d-1)  C(d)  D }{1 - 2(d-1)  C(d)  D}.$$ Since this bound did not depend on the specific choice of edge $e$, we see that our bound on $b_3$ is,
\begin{equation}
    b_3 \le \frac{2|\gamma|\Phi(P(e))(d-1)  C(d)  D }{1 - 2(d-1)  C(d)  D}.
\end{equation}

$b_2$ is trivially $0$ since by definition, a support of a configuration cannot have support on both $P(e)$ and $P(e')$ when $e'$ is in $B_e$.

Now, to estimate $b_1$, we use the bound $\mathbb{P}(F_{P(e)})\le \Phi(P(e))$ for all edges. Now for any edge $e$, there are at most $8(d-1)$ other edges $e'$ in $B_e$ that lie in $\gamma$. Here, we use the assumption that $\gamma$ has no self-intersection. We can now apply \ref{thn:PoisApproxGen} to assert that,
\begin{equation}
    \text{d}_{TV}(\mathcal{L}(M_{\gamma}), \text{Poisson}(\lambda)) \le  8(d-1) |\gamma|(\Phi(P(e)))^2 + \frac{2|\gamma|\Phi(P(e))(d-1) C(d)  D }{1 -2 (d-1)  C(d)  D}.
\end{equation}

Finally, 
an immediate consequence of the fact that $$\frac{\Phi(P(e)) }{1 -2 (d-1)  C(d)  D} \le \mathbb{P}(F_{P(e)}) \le \Phi(P(e))$$ shows that the $|\lambda- |\gamma| \Phi(P(e))| \le\frac{2|\gamma|\Phi(P(e))(d-1) C(d)  D}{1 -2 (d-1)  C(d)  D}$.

Corollary 3.1 of \cite{Adell} shows that $\text{d}_{TV}(\text{Poisson}(\lambda), \text{Poisson}(|\gamma|\Phi(P(e))) \le |\lambda - |\gamma|\Phi(P(e))|$. We can combine this with our earlier bound on the total variation to show the desired equation \eqref{eq:totalvarfromPoisson}. 
\end{proof}

Combining the results of Lemma \ref{lem:comparwpoisson} and Theorem \ref{thm:reductiontominimal}
will give us our following main result.

\begin{thm} \label{thm:MainThm1}
    Assume that the conditions of Theorem \ref{thm:reductiontominimal} and \ref{lem:comparwpoisson} hold. Let $X$ be a Poisson random variable with expectation parameter $|\gamma|\Phi(P(e)) =\exp[-4(d-1) \beta\text{Re}(\rho(1) - \rho(-1))]\exp[-2\kappa\text{Re}(f(1,1)- f(-1,1))] $. (Recall that the only way to excite a minimal vortex is to excite the center of the edge that is on its center.)
    Then we have that,
    \begin{equation}
    \begin{aligned}        &|\mathbb{E}[W_{\gamma}] - \mathbb{E}[\rho(-1)^{X}]| \le  8(d-1) |\gamma|(\Phi(P(e)))^2 + \frac{4|\gamma|\Phi(P(e))(d-1)  C(d)  D }{1 - 2(d-1)  C(d)  D}\\ & \hspace{1.2 cm}+  2(d-1)|\gamma|  (2^8 C(d) )^{2(d-1)+1} \exp[-4(d-1) \beta\text{Re} (\rho(1) - \rho(-1)]D .
    \end{aligned}
    \end{equation}
\end{thm}
\begin{proof}
This is a simple triangle inequality, provided we characterize the value of the difference of $\mathbb{E}[\rho(-1)^X]$ and $\mathbb{E}[\rho(-1)^{M_{\gamma}}]$. We can choose a coupling of $X$ and $M_\gamma$ such that $\mathbb{P}(X \ne M_{\gamma}) = d_{TV}(\mathcal{L}(M_{\gamma}), \text{Poisson}(|\gamma|\phi(P(e)))$. Now, $|\rho(-1)^{M_{\gamma},X}|\le 1$ for all values of $X$ or $M_{\gamma}$. Thus, we see that, under this coupling
$$|\mathbb{E}[\rho(-1)^X] - \mathbb{E}[\rho(-1)^{M_{\gamma}}]| \le \mathbb{P}(X \ne M_{\gamma}) \le d_{TV}(\mathcal{L}(M_{\gamma}), \text{Poisson}(|\gamma|\phi(P(e))).$$ We can apply the triangle inequality with our estimate on $|\mathbb{E}[W_{\gamma}] - \mathbb{E}[\rho(-1)^{M_{\gamma}}]|$ and complete the proof.
\end{proof}
\section{The Non-Abelian Case}

\subsection{Model and Some Preliminary Discussion}


\subsubsection{Definitions}
To treat the non-abelian case, we need significantly different notation. This preliminary section will discuss many of our new conventions.
The Hamiltonian we will consider in the case of a non-abelian gauge field $G$ will be as follows:

\begin{equation}
\begin{aligned}
    H_{N,\beta,\kappa}(\sigma,\phi)&= \sum_{p \in P_N} \beta( \psi_{p}(\sigma) - \psi_{p}(1)) 
    + \sum_{e \in E^{U}_N}\kappa [g_e(\sigma,\phi) - g^{u}(1,1)]. 
\end{aligned}
\end{equation}


$\sigma_e$ is still a map from the set of oriented edges $E_N \to G$, a non-abelian group with $\sigma_{-e}= \sigma_e^{-1}$, $\rho(\cdot)$ is a $D$ dimensional unitary representation of $G$. Finally, $\phi_x$  will be represented as a field $V_N \to H$ taking values in $H$, a finite multiplicative subgroup of the unit circle. In some sense, we can interpret $\phi$ as a scalar field. 

There are a few technical differences from the presentation of the abelian gauge model. First of all, we need to be more careful when defining the `current' around a plaquette $p$. If a plaquette $p$ has boundary vertices $v,w,x$ and $y$ and has an oriented boundary consisting of edges $e_1=(v,w)$, $e_2= (w,x)$, $e_3=(x,y)$ and $e_4=(y,v)$, then the `current'  defined around the oriented plaquette $p$ would be the product $\sigma_{e_1} \sigma_{e_2}\sigma_{e_3} \sigma_{e_4}$ and we set $\psi_p(\sigma)= \text{Tr}[\rho(\sigma_{e_1} \sigma_{e_2}\sigma_{e_3} \sigma_{e_4})]$, where $\rho$ is a $D$ dimensional unitary representation of the group $G$. We remark that even if one chose a different start edge for the boundary( $e_2,e_3,e_4$ and $e_1$ in that order, for example),  the `current' might change, but the Hamiltonian will not change due to the multiplicative property $\rho$ and the cyclic property of trace. Thus, $\psi_p( \sigma)$ will be well defined regardless of how we choose the starting edge of the boundary of $p$. 

Recall also that $P_N$ is a set of oriented plaquettes, so $P_N$ would also contain $-p$ whose boundary is $-e_4,-e_3,-e_2$ and $-e_1$ in that order. We see that our `currents' satisfy $(\sigma_{e_1}\sigma_{e_2}\sigma_{e_3} \sigma_{e_4})^{-1}= \sigma_{-e_4}\sigma_{-e_3}\sigma_{-e_2}\sigma_{-e_1}$.
With respect to these `currents', it is obvious $\psi_p(\sigma)= \overline{\psi_{-p}(\sigma)}$.
 By our earlier remarks, this relation does not depend on the specific choice of start point. We thus see that this part of the Hamiltonian will ultimately take real values.

We introduce the notion of unoriented edges $E_N^{U}$.  In contrast to $E_N$, for every adjacent pair of vertices $(x,y)$ in $\Lambda_N$, $E_N^{U}$ includes only a single undirected edge between $x$ and $y$ rather than both $e=(x,y)$ and $e_- = (y,x)$. The function $g^u_e(\sigma,\phi)$ will consist of the sum $\phi_x \text{Tr} [\rho(\sigma_e)] \phi_{y}^{-1} + \phi_y \text{Tr}[\rho(\sigma_{-e})] \phi_x^{-1}$ where $e$ is the oriented edge $(x,y)$ and $-e$ is the oriented edge $(y,x)$. Namely, it is the sum of the Higgs boson action $\phi_x \text{Tr}[\rho(\sigma_e)] \phi_y^{-1}$ over the oriented edge pair $\{e,-e\}$ in $E_N$. The reason we introduce $g$ instead of using our previous sum over $E_N$ is that we will eventually have to perform an expansion of the Hamiltonian with respect to these unoriented edges. One can see that the function 
$g_e$ is manifestly real.


\begin{rmk}
Cao \cite{SC20} introduces many notions from algebraic topology to describe non-abelian gauge field configurations. Though such concepts would be useful in places, it is not essential to understanding the logic of the proof in this section. Whenever we need to refer to some of these topological notions, I would refer to the appropriate location in \cite{SC20}.
\end{rmk}
\subsubsection{Removing Trivial Gauge Invariances}
We discuss some gauge invariances that simplify our analysis. First of all, we can assume that there is no element $\sigma \in G$ such that $\rho(\sigma)=I$, the identity matrix. The set of all elements $g' \in G$ such that $\rho(g')=I$ forms a normal subgroup $G'$ of $G$. To understand our Hamiltonian, it suffices to interpret $\sigma_e$ as a map from $E_N$ to the quotient group $G/G'$ instead. From a map from $E_N \to G/G'$, one can obtain all maps from $E_N \to G$ by freely multiplying each edge with some member of the group $G'$. As $\rho(g)=I$ for all members in $G'$, this will not change the Hamiltonian, nor Wilson loop values. Essentially, the reduction from $G$ to $G/G'$ is the removal of a trivial gauge invariance.

Furthermore, after this reduction, we know that if $\Psi_p(\sigma)= \Psi_p(1)$, then we can say in a well-defined sense that the `current' around the plaquette $p$ is exactly 1. The condition $\Psi_p(\sigma)= \Psi_p(1)$ for $p$ surrounded by boundary edges $e_1$,$e_2$,$e_3$ and $e_4$ in that order asserts that the product $\sigma_{e_1}\sigma_{e_2}\sigma_{e_3}\sigma_{_4}$ is an element of $G$ whose image under $\rho$ is the identity matrix. Our reduction shows the only such element is $1$, so $\sigma_{e_1} \sigma_{e_2}\sigma_{e_3}\sigma_{e_4}=1$. By multiplying by $\sigma_{e_1}^{-1}$ on the left and 
$\sigma_{e_1}$ on the right, we can derive that $\sigma_{e_2}\sigma_{e_3}\sigma_{e_4}\sigma_{e_1}=1$ and similar for other choices of start vertex.

Now, $|\text{Tr}(\rho(\sigma_e))| < D $ unless $\rho(\sigma_e)= e^{\frac{2\pi \ti j}{n}} I$ for some integers $j$ and $n$. Observe that since we are considering a unitary representation, all eigenvalues of $\rho(g)$ have absolute value less than $1$.
Thus, we know that $|\phi_x \text{Tr}[\rho(\sigma_e)] \phi_y^{-1}+ \overline{\phi_x \text{Tr}[\rho(\sigma_e)] \phi_y^{-1}}| \le 2D$ for all possible choices of $\phi_x,\sigma_e$ and $\phi_y$. Equality can only occur if $\rho(\sigma_e)=e^{\frac{2\pi \ti j}{n}} I $ for some integer $j$ and if $e^{\frac{2\pi \ti j}{n}} \in H$. 

Let $X$ be the set of values $e^{\frac{2\pi \ti j}{n}}$ such that there exists an element $g$ of $G$ such that $\rho(g)= e^{\frac{2\pi \ti j}{n}} I$. We see that $X$ is a subgroup of the multiplicative group of the unit circle, since if $\rho(g_1) = c_1 I$ and $\rho(g_2)= c_2 I$, where $X_1$ and $X_2$ are roots of unity, then $\rho(g_1 g_2) = c_1 c_2 I$. This implies that if $c_1 \in X$ and $c_2 \in X$, then the product $c_1 c_2$ is also in $X$. Furthermore, $\rho(g) = \rho(g)^{-1}$, so if $\rho(g)= cI$, then $\rho(g^-1)= c^{-1}I$. This shows that $X$ contains inverses. Manifestly $X$ has the identity, so $X$ must be a subgroup of the unit circle.

As we have done in the abelian case, we can consider the Higgs field to take values in $H / X$ instead of $H$. Observe that
the Wilson loop expecation can be computed as the ratio,

\begin{equation}
   \langle 
   W_{\gamma} \rangle=\frac{\sum_{\substack{\sigma_e \in G\\ e \in E_N}} \sum_{\substack{\phi_x \in H\\ x \in V_N }} W_{\gamma}(\sigma) \exp[H_{N,\beta,\kappa}(\sigma,\phi)]}{\sum_{\substack{\sigma_e \in G\\ e \in E_N}} \sum_{\substack{\phi_x \in H\\ x \in V_N }} \exp[H_{N,\beta,\kappa}(\sigma,\phi)]}.
\end{equation}

Fix some representative $h_i \in H$ for each coset class $C_i$ in $H/X$.
Now, fix a particular map $\phi_x=V_N \to H$; we will denote this map by $M$. For every vertex choose an element $\zeta^M_v$ and $\eta^M_v$ (depending on $M$) such that $\phi_v (\zeta^M_v)^{-1} =h_i$ for some $i$ and $\rho(\eta^M_v) = \zeta^M_v I$ . Notice that the transformation $(\{\phi_v\},\{\sigma_e\}) \to (\{\phi_v (\zeta^M_v)^{-1}\},\{\eta^M_v \sigma_e (\eta^M_w)^{-1} \} $ for $e=(v,w)$ does not change the Wilson loop action nor the value of the Hamiltonian.
Define $\phi^M_v :=\phi_v (\zeta^M_v)^{-1} $ and $\sigma^M_e=  \eta_v^M \sigma_e (\eta^M_w)^{-1}$.

Another way to represent the Wilson loop expectation is as,
\begin{equation}
    \langle W_{\gamma} \rangle = \frac{\sum_{M=\phi_x: V_N \to H} \sum_{\sigma_e:E_N \to 
    G} W_{\gamma}(\sigma^M) \exp[H_{N,\beta,\kappa}(\sigma^M,\phi^M)]}{\sum_{M=\phi_x: V_N \to H} \sum_{\sigma_e:E_N \to 
    G}  \exp[H_{N,\beta,\kappa}(\sigma^M,\phi^M)]}.
\end{equation}

In both the numerator and denominator, when considering the internal sum over $\sigma$, the map $\sigma \to \sigma^M$ can be treated as a change of parameters for the summation. That is, for fixed $M$, if $\sigma_e$ is uniform over all maps from $E_N \to G$, then $\sigma^M_e$ is uniform
for all maps from $E_N \to G$.

We can reparameterize the sum as,
\begin{equation}
    \langle W_{\gamma} \rangle = \frac{\sum_{M= \phi_x:V_N \to H} \sum_{\sigma_e:E_N \to G} W_{\gamma}(\sigma) \exp[H_{N,\beta,\kappa}(\sigma,\phi^M)]}{\sum_{M= \phi_x:V_N \to H} \sum_{\sigma_e:E_N \to G}  \exp[H_{N,\beta,\kappa}(\sigma,\phi^M)]}.
\end{equation}
Now, the map $\phi^M$ is a set of maps from $V_N \to H/X$. The map $M \to \phi^M$ is onto and the preimage of any $\phi^M$ is of size $|X|^{|V_N|}$. This shows that we can reduce our problem to computing Wilson loop expectations with a Higgs field taking values in $H/X$. After these reductions, we see that $|\phi_x \text{Tr}[\rho(\sigma_e)] \phi_y^{-1}+ \overline{\phi_x \text{Tr}[\rho(\sigma_e)] \phi_y^{-1}}|< 2d$ unless $\phi_x= \phi_y$ and $\rho(\sigma_e)=1$. Thus, we may consider edges $e=(x,y)$ with $\sigma_e\ne 1$ or $\phi_x \ne \phi_y$ as real excitations of the Hamiltonian.

\begin{rmk} \label{rmk:Higgsnonabelian}
If instead of choosing $\phi_x$ to be a scalar, we do the following steps
\begin{enumerate}
\item We let $\phi_x$ take values in a general group $H$.
\item We assume $H$ has a unitary representation $\rho_\phi$ with the same dimension as the unitary representation on $G$, $\rho$ and the intersection of the groups $\rho_{\phi}(H) \cap \rho(G)$ is a normal subgroup of the image $\rho_{\phi}(H)$.
\item We change  the Higgs boson interaction to,
\begin{equation}
    \sum_{e=(x,y) \in E^u_N} \text{Tr}[\rho_\phi(\phi_x) \rho(\sigma_e) \rho_\phi(\phi_y)] + \overline{\text{Tr}[\rho_\phi(\phi_x) \rho(\sigma_e) \rho_\phi(\phi_y^{-1})]}.
\end{equation}
\end{enumerate}
Then, after similar transformations as we have detailed above, we can show we have an excitation with exponential suppression $\exp[-O(\kappa)]$ whenever we have an edge $e=(x,y)$ that satisfies $\sigma_e \ne 1$ or $\phi_x \ne \phi_y$.

To show this, not that since $\rho_{\phi}(\phi_x) \rho(\sigma_e) \rho_{\phi}(\phi_y^{-1})$ is a unitary matrix,  so the absolute value of its trace is at most $D$, the dimension of the representation. Furthermore, the trace plus its conjugate is less than $2D$ unless the unitary matrix considered is the identity.

If $\rho_{\phi}(\phi_x) \rho(\sigma_e) \rho_{\phi}(\phi_y^{-1})$ is the identity, then $\rho(\sigma_e) = \rho_{\phi}(\phi_y \phi_x^{-1})$, and there is a common element in the image of $\rho_{\phi}(H)$ and $\rho(G)$. We can quotient out $\rho_{\phi}(H)$ by its intersection with $\rho(G)$ through the same gauge transformation procedure we outlined previously if $\rho_{\phi}(H) \cap \rho(G)$ is a normal subgroup of $\rho_{\phi}(H)$. If we also quotient out $G$ and $H$ by the elements such that $\rho(g)=I$ and $\rho_{\phi}(h)=I$, respectively, then we will have removed all possibilities for $\rho_{\phi}(\phi_x) \rho(\sigma_e) \rho_{\phi}(\phi_y^{-1})= I$. 
\end{rmk}

\subsection{The case of Low Disorder in the Higgs Field}

Before we analyze Wilson loop expectations under this Hamiltonian, we start by discussing the general difficulty of introducing non-abelian gauge interactions.

Our key tool in probabalistic computation was based on cluster expansions; thus, our first step in any problem was to first find an appropriate definition of cluster that would satisfy nice properties. In the case of a pure gauge field, the natural choice of clusters would be to find those plaquettes such that the current around it is $0$, or, in other words, $\psi_p(\sigma) \ne \psi_p(1)$.  These are the natural excitations that suppress probability.

In the case of an abelian gauge field, we are actually able to use the set of excited plaquettes $p$ with $\psi_p(\sigma) \ne \psi_p(1)$ as the basis of a legitimate cluster expansion. This due to the fact that a configuration $\sigma$ whose support consists two compatible sets $P_1 \cup P_2$ , according to Definition \ref{def:graphg2}, can be split into two configurations $\sigma_1$ and $\sigma_2$ whose supports are $P_1$ and $P_2$ respectively.

In the case of a non-abelian gauge field, it is no longer possible in general to split a configuration $\sigma$ whose support consists of two compatible sets $P_1 \cup P_2$ into two configurations $\sigma_1$ and $\sigma_2$ with support $P_1$ and $P_2$.  As shown in the papers  \cite{SC20} and \cite{PolandMan} ,
there are topological restrictions that prevent such a splitting. To get around this difficulty, one must consider a more sophisticated condition to determine whether we can split a configuration or not. This `sophisticated condition' is related the knotting properties of the vortices in the support. The combinatorial analysis of the knotting properties of vortices was done by Cao in \cite{SC20}. 

Now, let us return to our Higgs field model. As we have seen, a basic difficulty is that our cluster expansion must consider both the excitations of the Higgs field and the gauge field. We have observed already in the abelian case with low disorder (large $\kappa$) that considering the Higgs field makes defining the splitting substantially more complicated.

However, in the large $\kappa$ case, our construction becomes more robust when considering non-abelian groups. In fact, the cluster expansion we proposed in Section \ref{sec:ToyModel} works very well when considering non-abelian groups.
We try now to explain intuitively why we are able to more easily perform a splitting in the non-abelian case with low $\kappa$ than the pure gauge field case.

In the case of a pure gauge field, there is a large gauge symmetry. Namely, if we let  $\eta:V_N \to G$ be a map from the vertices $V_N$ to the group $G$, then $\sigma_e \to \eta_x \sigma_{e} \eta_y^{-1}$ does not change the value of the Hamiltonian. Thus, the basic object in the pure gauge field case is not a configuration, but a gauge equivalence class of configurations. A gauge equivalence class containing the configuration $\sigma$ will also contain the configurations $\eta_x \sigma \eta_y^{-1}$ for any map $\eta:V_N \to G$. 

When $G$ is abelian, these gauge equivalence classes of a configuration $\sigma$ can be understood as a $2$-form $(\td \sigma)$. These $2$-forms can split on disjoint supports by restriction. However, for non-abelian groups, the gauge equivalence classes are homomorphisms to the fundamental group; for details, one can refer to Section 4 of \cite{SC20}. These homomorphisms cannot split disjoint supports in general; the support components $P_1$ and $P_2$ cannot be knotted with each other if one wants to split the support between $P_1$ and $P_2$. 

By contrast, the introduction of the Higgs field and the $\kappa$ term ensures that there is a strong breaking of the gauge equivalence symmetry created by the maps $\eta:V_N \to G$. For large $\kappa$, our basic objects are indeed configurations $(\sigma,\phi)$ rather than gauge equivalence classes of configurations. The basic object of study in the large $\kappa$ regime of the Higgs boson stays the same, whether we are consider abelian or non-abelian groups. For this reason, the analysis in Section 2 is robust to the introduction of a non-abelian group $G$.



In the analysis of Section \ref{sec:ToyModel}, we found that any edge with $\sigma_e \ne 1$ automatically suppresses the probability by a factor of $\exp[-O(\kappa)]$. This allowed us to consider plaquettes that contain an edge with $\sigma_e \ne 1$
as part of our support. Immediately, if the support of our configuration is $P_1 \cup P_2$ where $P_1$ and $P_2$ are compatible, then we can split the gauge field part $\sigma$ of our configuration by restriction to $P_1$ and $P_2$.

The assignment of Higgs bosons to vertices is similar to what is done in Lemma \ref{lem:compsplit}. Lemma \ref{lem:compsplit} treated the case when $H/(\rho(G) \cap H)$ is isomorphic to $Z_2$. The key point of the argument was to identify the phase boundary and perform appropriate flips to correct phase boundaries when we separate the supports. The only key fact that we used about the group is that if we switch $\phi_x,\phi_y$ to $-\phi_x,-\phi_y$ where vertices $x$ and $y$ bound an edge $e$, then the Hamiltonian for that edge $e$ does not change. In particular, if we apply this flip for all vertices inside a phase boundary, it does not change the Hamiltonian except on the phase boundary.
For a more general group $H$, we still see that the Hamiltonian does not change on an edge if we apply the transformation $\phi_x \to h \phi_x$, $\phi_y \to h \phi_y$ for an arbitrary group element in $H/ (H \cap \rho(G))$. This will allow the argument to go through when we consider Higgs boson groups larger than $Z_2$.
Combining these two facts together, we see that all of the arguments outlined in Section \ref{sec:ToyModel} apply verbatim when $\kappa$ is sufficiently large.
\begin{rmk}
Even if we consider the case outline in Remark \ref{rmk:Higgsnonabelian}, with the Higgs boson taking non-abelian values, we can still apply the vertex assignment procedure outline in Lemma \ref{lem:compsplit}. This is again due to the fact that the transformation $\phi_x \to h\phi_x $ and $\phi_y \to h\phi_y $ does not change the Hamiltonian.
\begin{equation}
\begin{aligned}
    &\text{Tr}[\rho_{\phi}(h\phi_x) \rho(\sigma_e) \rho_{\phi}((h \phi_x)^{-1})] = \text{Tr}[\rho_{\phi}(h) \rho_{\phi}(\phi_x) \rho(\sigma_e) \rho_{\phi}(\phi_y^{-1}) \rho_{\phi}(h^{-1})]\\
    & = \text{Tr}[\rho_{\phi}(h^{-1})\rho_{\phi}(h) \rho_{\phi}(\phi_x) \rho(\sigma_e) \rho_{\phi}(\phi_y^{-1})] =  \text{Tr}[\rho_{\phi}(\phi_x) \rho(\sigma_e) \rho_{\phi}( \phi_x^{-1})]. 
\end{aligned}
\end{equation}
The first equality used the fact that $\rho_{\phi}$ is a representation. The second inequality used that $\text{Tr}[AB]= \text{Tr}[BA]$ for general matrices $A$ and $B$. The final inequality again used the fact that $\rho_{\phi}$ is a representation and $\rho_{\phi}(h)\rho_{\phi}(h^{-1}) = \rho_{\phi}(1)=I$.

\end{rmk}


\section{The High Disorder Regime: Small $\kappa$}\label{sec:nonabelianhighdisorder}

\subsection{Expansion in Random Currents}
In contrast to the low disorder regime, we cannot simply consider edges with $\sigma_e \ne 1$ to be excitations blindly. Unfortunately, this means we would have to consider a fundamentally new definition of cluster. We will observe later that in our new definition of cluster, the knotting problem is a serious difficulty when trying to split the configuration. Interestingly, this difficulty even appears in the case of an abelian gauge group with low $\kappa$. This means we have to treat the Higgs field model with  non-abelian and abelian gauge field with small $\kappa$ in the same way. To avoid presenting some long combinatorial estimates, we will restrict to the important case $d=4$ where we can cite these combinatorial estimates from previous results and focus on the new ideas.  


Intuitively, one can imagine in the small $\kappa$ case, there will be large fluctuations in the Higgs field which will decorrelate very quickly along large distances. To quantify this intuition in a cluster expansion, we introduce the notion of a random cluster expansion.
 First, let $c$ be some positive constant such that for any $\sigma_e$ $\phi_x$ and $\phi_y$, we have that,
\begin{equation}
    g_e(\sigma,\theta)= \phi_x \text{Tr}[\rho(\sigma_e)] \phi_y^{-1} + \overline{\phi_x \text{Tr}[\rho(\sigma_e)] \phi_y^{-1}} > -c, c \ge 0.
\end{equation}

We see that the measure generated by the Hamiltonian would be the same whether we considered $H_{N,\beta,\kappa}$ or the following Hamiltonian,
\begin{equation} \label{def:mathcalHam}
\begin{aligned}
  \tilde{H}_{N,\beta,\kappa}:&=  \sum_{p \in P_N} \beta [\psi_p(\sigma) - \psi_p(1)]
    + \sum_{e \in E^u_N}\kappa [g_e(\sigma,\phi) +c] .
\end{aligned}
\end{equation}

We can relate the above model to a random current model. This has three sets of random variables, the $\sigma$'s, the $\phi$'s, and a new set of edge activations $I(e)$. The marginal distributions of $\sigma$'s and $\phi$'s are those given by the Hamiltonian $\tilde{H}_{N,\beta,\kappa}$. Provided one has a configuration $\phi_x$ and $\sigma_e$, the distribution of $I(e)$ is given as follows:
\begin{equation}
    \mathbb{P}(I(e)= k)= \frac{[\kappa(g_e(\sigma,\phi) +c)]^{k}}{k! \exp[\kappa(g_e(\sigma,\phi)+c)]},
\end{equation}
where $k$ can be any non-negative integer.  The benefit of the random current expansion is that it is equivalent to finding the measure associated with the following Hamiltonian.
\begin{equation}
\begin{aligned}
    \mathcal{H}_{N,\beta,\kappa}(\phi,\sigma,I)&:= \sum_{p \in P_N} \beta (\psi_p(\sigma) -\psi_p(1))  \\&
    +\sum_{e \in E_N}  I(e)[\kappa(g_e(\sigma,\phi) +c)] - \log I(e)!.
\end{aligned}
\end{equation}
$I(e)$ is only allowed to take non-negative integer values. Essentially, the idea is that summing over $I(e)$ returns the exponential. 

With this in hand, we can start writing up our definition of cluster.
\begin{defn}\label{def:clustnonabel}
Given an configuration of the form $\mathcal{C}:=(\phi,\sigma,I)$, we can define the set of excited vertices as
\begin{equation}
    \mathcal{V}:=\{v\in V_N: \exists e \text{ s.t. } v \in \delta(e) \text{ and } I(e) \ne 0\}.
\end{equation}
We can define the support of the configuration as,
\begin{equation}
\begin{aligned}
    \supp\text{ }{\mathcal{C}}:=\{p \in P_N: & \psi_p(\sigma) \ne \psi_p(1) \text{ or } \\ & \exists v \in \mathcal{V} \text { s.t. } v \text{ is a boundary vertex of }  p\}.
\end{aligned}    
\end{equation}
In words, a plaquette is excited if either the plaquette has a non-trivial current running through it or it has a vertex  on its boundary that is adjacent to an edge that is excited with non-zero $I(e)$ value.

We still have the same notion of compatibility of plaquette sets $P_1$ and $P_2$ from Definition \ref{def:graphg2}. We also remark that the condition $\psi_p(\sigma) \ne \psi_p(1)$ here is the same as the condition $\sigma_p \ne 1$ in Section 4 of  \cite{SC20}. In addition, the support of a configuration only depends on the values of $\sigma$ and $I$, not on $\phi$.

\end{defn}

\begin{rmk} \label{rmk:minvortexsize6}
We can show that, with regards to this definition, the excitations with smallest support are still minimal vortices $P(e)$ centered around an edge $e$. We remark that if there is an excited edge $e$ with $I(e) \ne 0$, then the support will contain all plaquettes that share a boundary vertex with this plaquette. This is certainly more than 12 oriented plaquettes ( 6 plaquette pairs $\{p,-p\}$). Furthermore, Lemma 4.3.3 of \cite{SC20} asserts that the smallest vortex of the form $\supp(\sigma,0)$( i.e. all $I(e)=0$) has 12 excited oriented plaquettes( 6 plaquette pairs $\{p,-p\}$).

Furthermore, in order to avoid dealing with tedious edge cases, we will assume a boundary condition on which we will not find excitation whose support contains a plaquette on the boundary of $\Lambda_N$. One can show that for loops sufficiently far away from the boundary, the effect of excitations on the boundary is of size $e^{-N}$, where $N$ is the size of our lattice.
\end{rmk}

We now define the analogue of the function $\Phi$ as ,
\begin{equation}
\begin{aligned}
    \Phi(P) &= \frac{1}{|G|^{|V_N|-1}|H/X|^{|
    V_N|}} \sum_{\supp(\sigma,I) =P} \prod_{p \in P_N} \exp[\beta(\psi_p( \sigma)- \psi_p(1))] \\&\sum_{\phi} \prod_{e=(x,y) \in E_N} \frac{(\kappa( g_e(\sigma,\phi)+c))^{I(e)}}{I(e)!}.
\end{aligned}
\end{equation}


The normalization is introduced in order to derive some multiplicativity relations in the future. 
Because we are dealing with non-abelian gauge groups, we have a more complicated condition that ensures that if the support can be decomposed into two compatible plaquette sets $P_1$ and $P_2$ satisfying this complicated condition, then $\Phi(P_1 \cup P_2) = \Phi(P_1) \Phi(P_2)$. A sufficient condition for this splitting is the well-separated condition of Cao \cite{SC20}[Lemma 4.1.21]. 
\begin{defn}[Well-Separated] \label{def:wellsep}
This definition is derived from the one used in \cite{SC20}. We modify some of the notation used in \cite{SC20} to follow the conventions of this manuscript.

We say that two plaquette sets $P_1$ and $P_2$ are well-separated if we can find some rectangle $R$ in $\Lambda_N$ that satisfies the following properties: $P_1$ is contained in $P(R)$, the plaquettes inside $R$, while $P_2$ is contained in $P(R)^c$, the plaquettes outside $R$. Furthermore, $P_1$ and $P_2$ do not contain a plaquette belonging to the boundary of $R$.

\end{defn}
Wit this definition in hand, we can prove the following lemma.
\begin{lem}\label{lem:wellsplit}
    Consider two plaquette sets $P_1$ and $P_2$ that are well-separated according to Definition \ref{def:wellsep}. Then, we can make the following assertion,
    \begin{equation}
        \Phi(P_1 \cup P_2) = \Phi(P_1) \Phi(P_2).
    \end{equation}
\end{lem}
\begin{proof}
\textit{Part 1: Splitting the configuration $\sigma$}

Let $P_1$ and $P_2$ be separated by the rectangle $R$. Let $T$ be a spanning tree of $\Lambda_N$ that is simultaneously a spanning tree of $R$, the complement $R^c$, and the boundary of $R$ with some basepoint $b$.

Notice by how we defined the support of our configuration $P_1 \cup P_2$, all plaquettes $p$ such that $\psi_p(\sigma) \ne \psi_p(1)$ are contained in $P_1$  or in $P_2$. 

We will say that any gauge field configuration $\sigma$ is gauged with respect to the tree $T$ if $\sigma_e=1$ for all edges $e \in T$.
We can apply Lemma 4.1.21 of \cite{SC20} in order to find a  unique gauge transformation on the vertices $\eta^{\sigma}: V_N \to G$ such that $\tilde{\sigma}_e:= \eta^\sigma_x \sigma_e (\eta^\sigma_y)^{-1}$ with $\eta_b=1$ at the root and $\tilde{\sigma}$ is gauged with respect to the tree $T$. 

Once defining $\tilde{\sigma}_e$, we see that $\tilde{\sigma}
$ can be split into a product $\tilde{\sigma}_1$ and $\tilde{\sigma}_2$ such that $\tilde{\sigma}_e= (\tilde{\sigma}_1)_e (\tilde{\sigma}_2)_e$ for all edges $e=(x,y)$ such that $\tilde{\sigma}_1$ and $\tilde{\sigma}_2$ satisfy the following properties: $(\tilde{\sigma}_1)_e=\tilde{\sigma}_e$ for edges $e$ inside $R$ and $(\tilde{\sigma}_1)_e= 1$ for edges $e$ outside of $R$. In addition, $(\tilde{\sigma}_2)_e = \tilde{\sigma}_e$ for edges $e$ outside of $R$ and $(\tilde{\sigma}_2)_e =1$ for edges $e$ inside of $R$. This treats how we would divide $\sigma$ into two disjoint supports. Now, let us see how we would deal with activated edges.

Using the activated edges of $I(e)$, one can generate a subgraph $\mathcal{I}$ of the lattice $\Lambda_N$ whose edges consist of the activated edges with $I(e) \ne 0$. One can consider the connected clusters $C_1,C_2,\ldots,C_m$. We will observe later  in the last line of \eqref{eq:gaugesplit} that our function $\Phi$ will not share variables between the different clusters; this will allow us to split $\Phi$ appropriately as a product.  We also remark that given a cluster $C_i$, all of its vertices are either boundary vertices of plaquettes of  $P_1$ or of  $P_2$ exclusively. Otherwise, there exists some vertices $v$ and $w$ adjacent in the cluster $C_i$ such that $v$ is the boundary vertex of some plaquette $p_1$ in $P_1$ and $w$ is the boundary vertex of some plaquette $p_2$ in $P_2$. However, the edge $(v,w)$ is activated; thus all plaquettes that either have $v$ or $w$ as a boundary vertex are in the support of the configuration. In the connectivity graph of plaquettes $G_2$ from Definition \ref{def:graphg2}, this would imply that $p_1$ and $p_2$ are connected. This contradicts our assumption that $P_1$ and $P_2$ are disconnected components of the support.

We can let $V_1$ be the vertices that bound excited edges and are boundary vertices of some plaquette in $P_1$ and we let $I^a_1$ be the activated edges whose boundary vertices are both in $V_1$. We define $V_2$ and $I^a_2$ similarly. From our earlier discussion, $I^a_1 \cup I^a_2$ is a disjoint union covering all activated edges and $V_1$ is disjoint from $V_2$.


For an activation $I$, we define $I_1(e)$ to be the activation restricted to edges of $I_1^a$ and $I_2(e)$ to be the activation restricted to the edges of $I_2^a$.

We now claim that the map $(\tilde{\sigma},I)\to (\tilde{\sigma}_1,I_1),(\tilde{\sigma}_2,I_2)$ is a bijection from those configurations $(\tilde{\sigma},I)$ whose support is $P_1 \cup P_2$ and $\tilde{\sigma}$ is gauged with respect to the tree $T$ and pairs of configuration $(\tilde{\sigma}_1,I_1), (\tilde{\sigma}_2,I_2)$ whose supports are $P_1$ and $P_2$ respectively and $\tilde{\sigma}_1, \tilde{\sigma}_2$ are gauged with respect to the tree $T$. 
To show that this is a map between the proposed spaces,
it suffices to show that the support of $(\tilde{\sigma}_1,I_1)$ is $P_1$ exactly. By our construction, $(\tilde{\sigma}_1,I_1)$ has support contained in $P_1$. Now, let $p$ be an arbitrary plaquette in $P_1$. Since $p$ was in the support of $(\tilde{\sigma},P)$, we know that either $\psi_p(\tilde{\sigma})\ne \psi_p(1)$ or there is a vertex $v$ on the boundary of $p$ such that $v$ is adjacent to an excited edge, $e$ with $I(e)\ne 0$.

Consider the case that $\psi_p(\tilde{\sigma}) \ne \psi_p(1)$. We know that for $p \in P_1$ that $\psi_p(\tilde{\sigma})= \psi_p(\tilde{\sigma}_1)$ by construction. Thus, in this case we have that $p \in \supp(\tilde{\sigma}_1,I_1)$.  Now consider the case that there is a vertex $v$ in the boundary of $p_1$ such that $v$ is adjacent to an activated edge $e$ in $I$. Our earlier discussion shows that $v$ must be in $V_1$ and $e$ must be in $I^a_1$. Thus, $v$ is still adjacent to an activated edge in $I_1$.

Since all of our maps have been defined by restriction, it is easy to show the proposed map is a bijection. The splitting is clearly unique and one can combine an arbitrary split by multiplying the $\sigma$'s and combining the $I$'s. This proves our claim on the properties of the map $(\tilde{\sigma},I) \to (\tilde{\sigma}_1,I_1),(\tilde{\sigma}_2,I_2) $.

\textit{Part 2: The Multiplicativity of $\Phi$}


We rewrite our sum in $\Phi$ with respect to a base representative $\tilde{\sigma}_e$ gauged with respect to the spanning tree $T$ and a separate field $\eta: V_N \to G$, which re-introduces the gauge invariance.  The sum in $\Phi(P_1 \cup P_2)$ can be written as,
\begin{equation}\label{eq:gaugesplit}
\begin{aligned}
    \Phi(P_1 \cup P_2) &= \frac{1}{|G|^{|V_N|}|H/X|^{|V_N|}} \sum_{\substack{\supp(\tilde{\sigma},I)=P_1 \cup P_2\\\tilde{\sigma} \text{ gauged with }T}} \prod_{p \in P_1 \cup P_2}\exp[\beta (\psi_p(\tilde{\sigma}) - \psi_p(1))]\\& \sum_{\eta,\phi} \prod_{e=(x,y) \in \{e:I(e) \ne 0\}} \frac{(\kappa(g(\eta(\tilde{\sigma}),\phi)+c))^{I(e)}}{I(e)!}.
\end{aligned}
\end{equation}

The new variables $\eta:V_N \to G$ act as follows: it takes a gauge field configuration $\sigma$ to the configuration $\eta(\sigma)$ that takes values $\eta(\sigma)_e = \eta_x \sigma_e \eta_y^{-1}$ for the edge $e=(x,y)$. This reintroduces the gauge invariance we removed when defining $\tilde{\sigma}$.
 We remark that we get an extra factor of  $\frac{1}{|G|}$ due to the removal of the gauge fixing $\eta_b=1$ for the root. This can be done by a global transformation multiplying each element $\eta$ by the same element $g \in G$.

 Recall our map $(\tilde{\sigma},I) \to (\tilde{\sigma}_1,I_1),(\tilde{\sigma}_2,I_2)$ from earlier. We see that it can split the product as follows. 
\begin{equation}
\begin{aligned}
    \Phi(P_1 \cup P_2) &= \frac{(|G||H/X|)^{|V_N|-|V_1|-|V_2|}}{(|G|||H/X|)^{|V_N|}} \\&\sum_{\supp(\tilde{\sigma}_1,I_1)=P_1} \prod_{p \in P_1} \exp[\beta(\psi_p(\tilde{\sigma}_1) - \psi_p(1))] \\
    &\times \sum_{\substack{\eta_v,\phi_v\\v in V_1}} \prod_{\substack{e=(x,y)\\I_1(e) \ne 0}} \frac{(\kappa(g(\eta(\tilde{\sigma}_1),\phi)+c))^{I_1(e)}}{I_1(e)!}
    \\&\sum_{\supp(\tilde{\sigma}_2,I_2)=P_2} \prod_{p\in P_2}  \exp[\beta(\psi_p(\tilde{\sigma}_2) - \psi_p(1))] \\&\times \sum_{\substack{\eta_v,\phi_v\\v in V_2}} \prod_{\substack{e=(x,y)\\I_2(e) \ne 0}} \frac{(\kappa(g(\eta(\tilde{\sigma}_2),\phi)+c))^{I_2(e)}}{I_2(e)!}.
\end{aligned}
\end{equation}

We used the following facts:
\begin{enumerate}
\item Firstly, the mapping $(\tilde{\sigma},I) \to (\tilde{\sigma}_1,I_1),(\tilde{\sigma}_2,I_2)$ is a bijection.
\item Secondly, the product of the $\exp[\beta \psi_p]$ factors split between the disjoint supports $P_1$ and $P_2$.
\item Thirdly, the products $(\kappa(g(\eta(\tilde{\sigma},\phi)+c))^{I_1(e)}$ split between the disjoint sets $E_1$ and $E_2$ where the only variable values of $\eta$ and $\phi$ affecting the values of the product over $I_1$( $I_2$ resp.) are those in $V_1$ (resp. $V_2$).
\item Fourthly, we used the fact that the restriction of $\tilde{\sigma}$ to $E_1$(resp. $E_2$) is $\tilde{\sigma}_1$(resp. $\tilde{\sigma}_2$).
\item Finally, the sum over all remaining variables $\phi_x$ and $\eta_x$ belonging to variables not in $V_1$ or $V_2$ is $(|G||H/X|)^{|V|-|V_1|-|V_2|}$.
\end{enumerate}

Now, in the second line, we can reintroduce a dummy summation over new variables $\eta^1_v,\phi^1_v$ for vertices $v \not \in V_1$ and a new dummy summation over variables $\eta^2_v, \phi^2_v$ for vertices $v \not \in V_2$ in the third line. We will also relabel the variables $\eta_v,\phi_v$ in the second line as $\eta^1_v,\phi^1_v$ to simplify notation; similarly, we relabel $\eta_v,\phi_v$ as $\eta^2_v,\phi^2_v$ in the third line.

Compensating for the extra multiplicative factor of $(|G||H/X|)^{|V_N|-|V_1|}$ in the second line due to the new variables we added, as with $(|G||H/X|)^{|V_N|-|V_2|}$ in the second line, we see that we can write the above expression as expression 
\begin{equation}
\begin{aligned}
     &\Phi(P_1 \cup P_2) = \\ &\frac{1}{(|G||H/X|)^{|V_N|}}\sum_{\supp(\tilde{\sigma}_1,I_1)=P_1} \prod_{p \in P_1} \exp[\beta(\psi_p(\tilde{\sigma}_1) - \psi_p(1))] \\
     & \times\sum_{\substack{\eta^1_v,\phi^1_v}} \prod_{\substack{e=(x,y)\\I_1(e) \ne 0}} \frac{(\kappa(g(\eta^1(\tilde{\sigma}_1),\phi^1)+c))^{I_1(e)}}{I_1(e)!}
    \\&\frac{1}{(|G||H/X|)^{|V_N|}}\sum_{\supp(\tilde{\sigma}_2,I_2)=P_2} \prod_{p\in P_2}  \exp[\beta(\psi_p(\tilde{\sigma}_2) - \psi_p(1))] \\
    & \times \sum_{\substack{\eta^2_v,\phi^2_v}} \prod_{\substack{e=(x,y)\\I_2(e) \ne 0}} \frac{(\kappa((\eta^2(\tilde{\sigma}_2),\phi^2)+c))^{I_2(e)}}{I_2(e)!}.
\end{aligned}
\end{equation}

But, this is just the product $\Phi(P_1) \Phi(P_2)$.

\end{proof}

\subsection{Knot Expansion}

The above Lemma \ref{lem:wellsplit} is most useful when paired with the concept of knot expansions. A knot expansion serves as a more careful way to split the support of a configuration into disjoint components that would be better for analysis.

\begin{defn}
Let $V_1 \cup V_2 \ldots \cup V_m$ be a vortex decomposition. A partition of these vortices into knots $K_1 \cup K_2 \cup \ldots \cup K_n$ will be called a knot decomposition if for each $j$, there is a box $R_j$ that well-separates $K_j$ from $K_{j+1} \cup K_{j+2}\cup \ldots \cup K_n$. As a consequence of Lemma \ref{lem:wellsplit}, we would have that $\Phi(K_1 \cup K_2 \cup \ldots \cup K_n) = \prod_{i=1}^n \Phi(K_i)$.

A knot decomposition will be maximal if there is no further way to partition any $K_j= K_j^1 \cup K_j^2$ into non-empty components such that $K_j^1$ and $K_j^2$ are well-separated by a cube in $\Lambda_N$.
\end{defn}





\begin{rmk} \label{rmk:minvortexisknot} Due to Lemma 4.1.10 of \cite{SC20}, if one of the vortices in the support of a configuration is a minimal vortex $P(e)$ centered around some edge $e$, then one would always be able to separate out this minimal vortex in the knot decomposition; though this slightly abuses the notion of separating box.  We will usually choose these minimal vortices to be the first knots in the knot decomposition. This is exactly the convention of \cite{SC20}.  \end{rmk} 

With the notion of knot decomposition and Lemma \ref{lem:wellsplit} in hand, we can rather easily prove the following statement. This is essentially Lemma 4.3.10 of \cite{SC20}.
\begin{lem} \label{lem:boundprob}
    The probability that, under the Hamiltonian $ \mathcal{H}_{N,\beta,\kappa}$, we will observe a configuration  $(\sigma,\phi,I)$  whose support under the knot decomposition contains $K$ would be less than 
   $\Phi(K)$.
\end{lem}
\begin{proof}

Consider a knot decomposition of the form $K_1 \cup \ldots \cup K_j \cup \ldots K_m$ that includes $K=K_j$. $|G|^{|V_N|-1}|H/X|^{|V_N|}\Phi(K_1 \cup \ldots \cup K_j \cup \ldots \cup K_m)$ will be the sum of $\exp[\mathcal{H}_{N,\beta,\kappa}]$ among all configurations that have knot decomposition $K_1 \cup \ldots \cup K_j \cup \ldots \cup K_m$.

We see that,
\begin{equation}\label{eq:somesum}
    \sum_{K \in \supp(\sigma,\phi,I)} \exp[\mathcal{H}_{N,\beta,\kappa}(\sigma,\phi,I)] =|G|^{|V_N|-1}|H/X|^{|V_N|} \sum_{K \in K_1 \cup K_2 \ldots K_m} \prod_{i=1}^m\Phi(K_i).
\end{equation}

Where we abuse notation to say that $K \in K_1 \cup \ldots\cup K_m$ means that $K$ is one of the knots in the decomposition.

We see that if we remove $K=K_j$ from $K_1\cup \ldots \cup K_m$, then , though $K_1 \cup \ldots \cup K_{j-1} \cup K_{j+1}\ldots \cup K_m$ may split further in a knot decomposition, we would still satisfy a well-separatedness condition to assert that
$\Phi(K_1 \cup \ldots\cup K_{j-1} \cup K_{j+1} \cup K_m)= \prod_{\substack{i=1\\i \ne j}}^n \Phi(K_i) $.

By considering those configurations whose support would be contained in $K_1 \cup K_{j-1} \cup K_{j+1} \cup K_m$ for all knot decompositions $K_1 \cup\ldots \cup K_m$ that contain $K=K_j$, we see the partition function can be bounded below by,

\begin{equation}
     Z_{\mathcal{H}_{N,\beta,\kappa}} \ge |G|^{|V_N|-1}|H/X|^{|V_N|} \sum_{K \in K_1 \cup K_2 \ldots \cup K_m} \prod_{\substack{i=1\\i \ne j}}^m\Phi(K_i).
\end{equation}

Taking the ratio of the term in \eqref{eq:somesum} with our lower bound on the partition function gives us that the probability of seeing  a configuration whose support contains $K$ is bounded by $\Phi(K)$.


\end{proof}

We also have the following quantitative bound on $\Phi(K)$.
\begin{lem} \label{lem:TrivBound}
Let $K$ be some union of vortices $V_1 \cup V_2 \cup \ldots \cup V_N$ with $2k$ oriented plaquettes ( $k$ pairs of plaquettes $\{p,-p\}$) in the support.
Define the constant,
$$\alpha_{\beta,\kappa}:= 2^4 |G|\max(\exp[2\beta(\max_{a\ne 1 \in G} \text{Re}[\psi_p(a) -\psi_p(1)])], (\kappa \mathfrak{d}\exp[\kappa \mathfrak{d}])^{\frac{1}{2{8 \choose 2}}}),  $$
where $\mathfrak{d}$ is the maximum value of $g_e(\sigma,\phi) + c$ for all possible values $\sigma$ and $\phi$. Note that this maximum does not depend on the edge $e$.

Then, we have the following bound on $\Phi(K)$.

\begin{equation}
   \Phi(K)\le \alpha_{\beta,\kappa}^k.
\end{equation}
\end{lem}
\begin{proof}

Choose some spanning tree $\mathcal{T}$. We see that we can express $\Phi(K)$ with respect to configurations gauged with respect to $\mathcal{T}$ as,
\begin{equation}
\begin{aligned}
    \Phi(K) = &\frac{1}{|G|^{|V_N|}|H/X|^{|V_N|}} \sum_{\supp(\tilde{\sigma},I)=K} \prod_{p \in P_N} \exp[\beta(\psi_p(\tilde{\sigma})-\psi_p(1)] \\&\sum_{\eta_v,\phi_v}\prod_{e=(x,y):I(e) \ne 0} \frac{(\kappa(g(\eta(\tilde{\sigma}),\phi)+c))^{I(e)}}{I(e)!}.
\end{aligned}
\end{equation}

Notice that the support of $\tilde{\sigma},I$ only depends on the set $A_{I}$ of activated edges of $I$, rather than the particular values of $I$. Given a set of edges $A$ we will let $\supp(\tilde{\sigma},A)$ to be the support of a configuration $(\tilde{\sigma},I)$ whose set of activated edges from $I$ is $A$. We remark that the definition would not depend on the choice $I$.

What we do now is replace the first sum over sets $\tilde{\sigma},A$ where $A$ is now just a set of activated edges. Later, just before we take the product over all edges in the second line, we sum values of $I(a)$ from $1$ to $\infty$ where $a$ varies over all activated edges in $A$.

Namely, we write,
\begin{equation} \label{eq:phinewrep}
\begin{aligned}
    \Phi(K) = &\frac{1}{|G|^{V_N}|H/X|^{|V_N|}} \sum_{\supp(\tilde{\sigma},A)=K} \prod_{p \in K} \exp[\beta(\psi_p(\tilde{\sigma})-\psi_p(1)] \\&\sum_{\eta_v,\phi_v}\prod_{e=(x,y):e \in A} \sum_{I(e)=1}^{\infty}\frac{(\kappa(g_e(\eta(\sigma),\phi)+c))^{I(e)}}{I(e)!}.
\end{aligned}
\end{equation}

The sum $\sum_{I(e)=1}^{\infty}\frac{(\kappa(g(\eta^1(\sigma),\phi^1)+c))^{I(e)}}{I(e)!}$ can be bounded by $\kappa \mathfrak{d}\exp[\kappa d\mathfrak{d}]$, where $\mathfrak{d}$ is the maximum of $g_e(\eta(\sigma),\phi)+c$ over all choices of $\eta$,$\sigma$ and $\phi$. Recall that since $g_{e}$ is a local function, there is clearly a maximum value.

The last line of \eqref{eq:phinewrep} can be bounded by $|G|^{|V_N|}|H/X|^{|V_N|} (\kappa \mathfrak{d} \exp[\kappa \mathfrak{d}])^{|A|}$, by performing a trivial sum over all $\eta_v$ and $\phi_v$ variables. Fortunately, this prefactor cancels out.

Now, we bound the product $\prod_{p\in P_N} \exp[\beta(\psi_p(\tilde{\sigma})-\psi_p(1))] (\kappa \mathfrak{d}\exp[\kappa \mathfrak{d}])^{|A|}$ for any configuration with $\supp(\tilde{\sigma},A)=P$.
 Let $C_p$ be the set of plaquettes in $K$ such that $\psi_p(\tilde{\sigma}) \ne \psi_p(1)$. For plaquettes in $K \setminus C_p$, we would know that, instead, they must be activated due to an edge $e \in A$ that shares a vertex with a plaquette. A single vertex can be part of at most $2{8 \choose 2}$ oriented plaquettes, so a single edge can excite at most $4{8 \choose 2}$ oriented plaquettes. Thus, we see that $|A| \ge \frac{|K|-|C_p|}{4{8 \choose 2}}$.
 
 For any configuration such that $\supp(\tilde{\sigma},A)=K$, we see that we can bound $\prod_{p\in P_N} \exp[\beta(\psi_p(\tilde{\sigma})-\psi_p(1))] (\kappa \mathfrak{d}\exp[\kappa \mathfrak{d}])^{|A|}$ by $$\max(\exp[2\beta(\max_{a\ne 1 \in G} \text{Re}[\psi_p(a) -\psi_p(1)])], (\kappa \mathfrak{d}\exp[\kappa \mathfrak{d}])^{\frac{1}{2{8 \choose 2}}})^{k}.$$
 
 Now, all that we have left is to count the number of configurations $\tilde{\sigma},A$ with support $K$. For each edge $e$ that is a boundary edge of a plaquette $p \in K$, we have at most $2$ choices: to activate the edge or not. Thus, the number of way to choose $A$ is at most $2^{4k}$ with $4k$ being the maximum number of edges that could possibly be boundary edges of plaquettes in $K$. Now, we finally have to count the number of $\tilde{\sigma}$ with support $P$. In \cite{SC20}[Lemma 4.3.7], this can be reduced to finding the number of homomorphisms from the fundamental group of the $2$-skeleton of the complement of $K$. From \cite{SC20}[Lemma B.2], this is bounded by $|G|^{k}$. Multiplying all of these constant together will give us the desired bound on $\Phi(K)$.


\end{proof}

\begin{rmk}\label{rmk:reduction}
Provided we know that $\beta$ is sufficiently large and $\kappa$ is sufficiently small, from this point, we could follow the proof of Corollary 1.2.1 in \cite{SC20} nearly word for word by using a sufficiently small value of the constant $\alpha_{\beta,\kappa}$. The proof can follow using the same steps; namely, finding a good event $E$, bounding the probability of $E^c$, and computing the value of the Wilson loop conditioned on $E^c$. If one chooses $\kappa^{\frac{1}{2{8 \choose 2}}}$ to be of the order $\exp[-\beta]$, then one could expect this to be essentially optimal. 
\end{rmk}

For the remainder of this section, we will give simpler arguments deriving the leading order behavior of the Wilson loop expectation for the convenience of the reader. One can express some corrections to the leading order in $\kappa$ through a more delicate argument. For interested readers, this procedure is performed in Section \ref{sec:decorrelation} through a careful analysis of decorrelation properties of an associated Hamiltonian.

\subsection{Identifying the Main Order Excitations}

As we have mentioned in Remark \ref{rmk:reduction}, the properties of $\Phi$ we have established in the previous subsections are sufficient to prove that the main order terms of the Wilson loop expectation are given by minimal vortices.

When $\exp[-\beta] \ll \kappa$, we can more precisely describe the value of the error term with a little bit more effort and a more careful identification of the main order terms.

\begin{defn}
We call a configuration $(\sigma,\phi,I)$ Wilson loop non-trivial if the following holds: there exists some plaquette $p$ such that $\psi_p(\sigma) \ne \psi_p(1)$. We can define the function $\Phi_{NT}$ by summing over $\exp[ \mathcal{H}_{N,\beta,\kappa}(\sigma,I,\phi)]$ over all configurations that are Wilson loop nontrivial.
Namely, we have that, when $P$ is some union of vortices,
\begin{equation}
    \Phi_{NT}(P) = \sum_{\substack{\supp(\sigma,\phi,I)= P\\ (\sigma,\phi,I) \text{ W.L. Nontrivial}}} \exp[\mathcal{H}_{N,\beta,\kappa}(\sigma,\phi,I)].
\end{equation}

A configuration $(\sigma,\phi,I)$ that is not Wilson loop-nontrivial is Wilson loop trivial. A consequence of being Wilson loop trivial is that there is a way to choose a field $\eta: V_N \to G$ such that the modified field $\tilde{\sigma}_e = \eta_x \sigma_{e}\eta_y^{-1}=1$.

If a configuration $(\sigma,\phi,I)$ has support $P_1 \cup P_2$ where $P_1$ and $P_2$ are compatible, we can say that $(\sigma,\phi,I)$ is Wilson loop nontrivial on $P_1$ if there is a plaquette $p \in P_1$ such that $\psi_p(\sigma) \ne \psi_p(1)$.

\end{defn}

We now describe some properties of $\Phi_{NT}$ that are reminiscent of properties of $\Phi$.

If we have a configuration supported on a union of well-separated plaquette sets $P_1 \cup P_2$, then we can say that a configuration $(\sigma,\phi,I)$ is Wilson loop non-trivial when restricted to $P_1$ if there is a plaquette $p \in P_1$ such that  $\psi_p(\sigma) \ne \psi_p(1)$.

We can follow the construction in the proof of Lemma \ref{lem:wellsplit} to make the following assertion:
we have the identity
$$
|G|(|G||H/X|)^{-|V_N|} \sum_{\substack{\supp(\sigma,\phi,I)=P_1 \cup P_2\\(\sigma,\phi,I) \text{ W.L. N.T. on } P_1}} \exp[\mathcal{H}_{N,\beta,\kappa}(\sigma,\phi,I)]= \Phi_{NT}(P_1) \Phi(P_2).
$$
This is due to the fact that under the restriction map on the gauged version of $\sigma$, e.g. the map $\tilde{\sigma} \to (\tilde{\sigma}_1,\tilde{\sigma}_2)$, we see we must have that $\tilde{\sigma}_1$ must be Wilson loop non-trivial on $P_1$ if $P$ was Wilson loop non-trivial on $P$.

Due this multiplicative property, we can follow the proof of Lemma \ref{lem:boundprob} to assert that the probability that we have a configuration whose support has a knot decomposition that contains $K$ and is Wilson loop non-trivial on $K$ is bounded from above by $\Phi_{NT}(K)$. 
Finally, since we know that since a minimal vortex has size at least $12$ oriented plaquettes from Remark \ref{rmk:minvortexsize6}, the support of all Wilson loop nontrivial configurations has at least 14 oriented plaquettes (7 plaquette pairs $\{p,-p\}$). Following of the proof of Lemma \ref{lem:TrivBound}, we can assert that that $\Phi_{NT}(K) \le (2^4 |G|)^{6} \exp[-12 \beta (\max_{g \ne 1} \text{Re}[\psi_p(g) - \psi_p(1)])] \alpha_{\beta,\kappa}^{k-6}$.

We are now at the stage where we can define our good event, $E$, where we have a characterization of the Wilson loop expectation in terms of minimal vortices.

\begin{defn} \label{def:goodevent}
We will define a set $E$ of `good' configurations by applying conditions for a configuration to be in the complement.
We say that a configuration $(\sigma,\phi,I)$ is in the complement, $E^c$, if there is a knot $K$ in the knot expansion corresponding to $(\sigma,\phi,I)$ such that $K$ satisfies the following properties:
\begin{enumerate}
    \item $K$ is not a minimal vortex.
    \item $K$ is Wilson loop nontrivial.
    \item There does not exist a cube $B_K$ that separates $K$ from the plaquettes $P_{\gamma}:\{ p \in P_n: \exists e \in \gamma \cap \delta p\}$, e.g. the set of plaquettes that bound an edge of $\gamma$.
\end{enumerate}

\end{defn}

We first show that restricted to the event $E$, we have a rather easy computation of the Wilson loop expectation, much like in Lemma \ref{lem:minimalCont}
\begin{lem} \label{lem:onrareevent}
Define the matrix $A_{\beta,\kappa}$ as follows.
\begin{equation}\label{def:Abeka}
A_{\beta,\kappa}:=\frac{\sum_{g \ne 1} \rho(g) \exp[12\beta\text{Re}(\text{Tr}[\rho(g)] - \text{Tr}[\rho(1)])]}{\sum_{g \ne 1} \exp[12 \beta\text{Re}(\text{Tr}[\rho(g)]- \text{Tr}[\rho(1)])]}.
\end{equation}
Conditioned on the event $E$ from Definition \ref{def:goodevent}, we see that we have,
\begin{equation}
    \mathbb{E}[W_{\gamma}|E] = \mathbb{E}[A_{\beta,kappa}^{M_\gamma}|E],
\end{equation}
where $M_{\gamma}$ is the number of minimal vortices in the support of the configuration that are centered on an edge of the loop $\gamma$.
\end{lem}
\begin{proof}
Let $\mathcal{V}_1 \cup \mathcal{V}_2 \cup \mathcal{V}_3$ be a valid knot decomposition, i.e. corresponding to some configuration $(\sigma,\phi,I)$, where the following properties hold.
\begin{enumerate}
    \item $\mathcal{V}_1$ consists of a union of minimal vortices $P(e_1)\cup \ldots \cup P(e_k)$ centered around edges $e_i\in \gamma$.
    \item $\mathcal{V}_2$ consists of knots that are well separated from the plaquettes that bound an edge of $\gamma$. Namely, for each knot $K$ in $\mathcal{V}_2$, there exists a cube $B_K$ that separates $K$ from the plaquettes $P_{\gamma}$.
    \item $\mathcal{V}_3$ consists of knots $K$ such that $(\sigma)$ restricted to $K$ is Wilson loop trivial.
\end{enumerate}  
Let $E(\mathcal{V}_1,\mathcal{V}_2,\mathcal{V}_3)$ be the event that the configuration $(\sigma,\phi,I)$ has exactly this knot decomposition. 

Our goal is to provide an expression of \begin{equation}
\begin{aligned}
    \Phi_{\gamma}(\mathcal{V}_1 \cup \mathcal{V}_2 \cup \mathcal{V}_3):= \frac{|G|}{(|G||H/X|)^{|V_N|}}\sum_{\substack{\supp(\sigma,\phi,I)=\\ \mathcal{V}_1 \cup \mathcal{V}_2 \cup \mathcal{V}_3}} W_{\gamma}(\sigma) \exp[\mathcal{H}_{N,\beta,\kappa}(\sigma,\phi,I)],
    \end{aligned}
\end{equation}based on this knot decomposition.  Say our knot decomposition is $K_1 \cup \ldots \cup K_m$. By our convention, $K_m$ can belong to $\mathcal{V}_3$ or $\mathcal{V}_2$.

Consider the case that $K_m$ belongs to $\mathcal{V}_3$. Let $B(K_m)$ be a cube that separates $K_m$ from the other knots in the knot decomposition and $T(K_m)$ a joint spanning tree of $B(K_m)$, the boundary and the complement.

Recall that since our knot $K_m$ is in $\mathcal{V}_3$, we must necessarily have that $\psi_p(\sigma)=\psi_p(1)$ for all plaquettes $p \in K_m$. By gauging our configuration$(\sigma,\phi,I)$ with respect to $T(K_m)$, we see the gauged version $(\tilde{\sigma},\phi,I)$ satisfies $\tilde{\sigma}_e=1$ for all edges inside $B(K_m)$.

By the calculations similar to those performed in the proof of Lemma \ref{lem:wellsplit}, we see that we have the following expression for $\Phi_{\gamma}(\mathcal{V}_1 \cup \mathcal{V}_2 \cup \mathcal{V}_3)$.

\begin{equation} \label{eq:splitcompli}
    \begin{aligned}
    &\Phi_{\gamma}(\mathcal{V}_1 \cup \mathcal{V}_2 \cup \mathcal{V}_3)= \frac{1}{(|G||H/X|)^{|B(K_m)^c|}}
      \sum_{\substack{\supp(\tilde{\sigma},I_1) = \bigcup_{i=1}^{m-1}K_i\\\tilde{\sigma} \text{ gauged with } T(K_m)}}W_{\gamma}(\tilde{\sigma}) \\
      & \times \prod_{p \in \mathcal{V}_2} \exp[\beta (\psi_p(\tilde{\sigma}) - \psi_p(1)] \sum_{\substack{\eta_v,\phi_v\\ v \in B(K_m)^c}} \prod_{\substack{e=(x,y)\\I_1(e) \ne 0}}  \frac{(\kappa(g(\eta(\tilde{\sigma}),\phi) +c))^{I_1(e)}}{I_1(e)!}\\
    & \times \frac{1}{(|G||H/X|)^{|B(K_m)|}} \sum_{\supp(1,I_2)= K_m} \sum_{\substack{\eta_v,\phi_v\\ v \in B(K_m)}} \prod_{\substack{e=(x,y)\\I_2(e) \ne 0}}  \frac{(\kappa(g(\eta(1),\phi))+c)^{I_2(e)}}{I_2(e)!}.
    \end{aligned}
\end{equation}

Again, the main fact we used is that when we gauged $\sigma$ with respect to $T(K_m)$, the fact that the configuration was Wilson loop trivial on $K_m$ implied that $\tilde{\sigma}=1$ on all edges of $B(K_m)$ after the gauging. Thus, the Wilson loop action only depends on the component $\tilde{\sigma}$ on the edges of $B(K_m)^c$. In addition, the splitting divides activated edges inside $B(K_m)^c$ in the first component and activated edges inside $B(K_m)$ in the second component. These activated edges do not share any vertices on the boundary. 

By adding appropriate auxiliary variables $\sum \eta_v^1,\phi_v^1$ in the first line and $\sum \eta_v^2,\phi_v^2$ in the second line and normalizing, we see that we have $\Phi_{\gamma}(\mathcal{V}_1 \cup \mathcal{V}_2 \cup \mathcal{V}_3)= \Phi(\mathcal{V}_1 \cup \mathcal{V}_2 \cup \mathcal{V}_3 \setminus K_m) [\Phi(K_m)- \Phi_{NT}(K_m)] $.

If instead, the knot $K_m$ belonged to $\mathcal{V}_2$, we can follow a similar argument, with only one observation we have to make. When we choose the cube $B(K_m)$ that well-separates $K_m$ from the other knots, our condition on $\mathcal{V}_2$ also ensures that this cube well-separates $K_m$ from the plaquettes $P_{\gamma}$ along $\gamma$. This ensures when we gauge our transformation according to $T(K_m)$ and get a splitting $\tilde{\sigma}_1$ inside $B(K_M)^c$ and $\tilde{\sigma}_2$ in $B(K_m)$, we have that $(\tilde{\sigma}_2)_{e}=1$ for all edges in $\gamma$. Thus, we see again that $\Phi_{\gamma}(\mathcal{V}_1 \cup \mathcal{V}_2 \cup \mathcal{V}_3)= \Phi(\mathcal{V}_1 \cup \mathcal{V}_2 \cup \mathcal{V}_3 \setminus K_m) \Phi(K_m) $ holds.

By fully iterating this procedure, we see that ultimately we can remove all knots that are not minimal vortices centered around edges of $\gamma$ and see that,
\begin{equation}
    \Phi_{\gamma}(\mathcal{V}_1 \cup \mathcal{V}_2 \cup \mathcal{V}_3)= \Phi_{\gamma}(\mathcal{V}_1) \Phi(\mathcal{V}_2) \prod_{K_j \in \mathcal{V}_3}[\Phi(K_j) - \Phi_{NT}(K_j)].
\end{equation}

Now, we finally evaluate $\Phi_{\gamma}(\mathcal{V}_1)$. We remark that there are no activated edges with $I(e) \ne 0$. If there were such an edge, our convention would suggest that all plaquettes that share a boundary vertex with this edge would be excited. This would be larger than a minimal vortex. To evaluate $\Phi_{\gamma}(\mathcal{V}_1)$, we apply Corollary 4.1.16 of \cite{SC20} and choose a spanning tree of $\Lambda_N$ that avoids using the edges $e_i$ that are centers of the minimal vortices in $P$. When we gauge configurations with respect to this spanning tree, the only non-trivial edges with $\sigma_e \ne 1$ are those edges that for the centers of the minimal vortices.

If we define the matrix $$\tilde{A}_{\beta,\kappa}:= \sum_{g\ne 1} \rho(g) \exp[12\beta \text{Re}(\text{Tr}[\rho(g)] - \text{Tr}[\rho(1)])],$$ we see that 
\begin{equation}\label{eq:numeratorhere}
\Phi_{\gamma}(\mathcal{V}_1)= \prod_{e_i \in \mathcal{V}_1} \sum_{g_i \ne 1} \exp[12 \beta \text{Re} [\text{Tr}[\rho(g_i)]- \text{Tr}[\rho(1)]]] \text{Tr}[\rho(\prod_{e_i \in \mathcal{V}_1} \rho(e_i))]. 
\end{equation}

We abuse notation slightly when we write $e_i \in \mathcal{V}_1$. We use this to mean that $e_i$ is the clockwise ordering of the minimal vortices that compose $\mathcal{V}_1$ around the loop $\gamma$.



A similar decomposition can evaluate $\Phi(\mathcal{V}_1 \cup \mathcal{V}_2 \cup \mathcal{V}_3)$ (we abuse notation slighlty here; we only restrict to configurations that are trivial on $\mathcal{V}_3$ rather than consider all configurations supported on $\mathcal{V}_1 \cup \mathcal{V}_2 \cup \mathcal{V}_3$) as,
\begin{equation}\label{eq:denominatorhere}
\begin{aligned}
    &\Phi(\mathcal{V}_1 \cup \mathcal{V}_2 \cup \mathcal{V}_3) =\\
    &\prod_{e_i \in \mathcal{V}_1} \sum_{g_i \ne 1} \exp[12 \beta \text{Re}(\text{Tr}[\rho(g_i)] - \text{Tr}[\rho(1)])] \Phi(V_2)\prod_{K \in \mathcal{V}_3}[\Phi(K) - \Phi_{NT}(K)].
\end{aligned}
\end{equation}


We can take the ratios of the quantities in equations \eqref{eq:numeratorhere} and \eqref{eq:denominatorhere} to show that conditioned on the event $E(\mathcal{V}_1,\mathcal{V}_2,\mathcal{V}_3)$, we see that $\mathbb{E}[W_{\gamma}|E(\mathcal{V}_1, \mathcal{V}_2,\mathcal{V}_3)]= \text{Tr}[A_{\beta,\kappa}^{k}]$. Since the event $E$ is the disjoint union of $E(\mathcal{V}_1,\mathcal{V}_2,\mathcal{V}_3)$ over all possible knot decompositions that could appear for configurations in $E$, we see that
$\mathbb{E}[W_\gamma|E]= \mathbb{E}[\text{Tr}[A_{\beta,\kappa}^{M_\gamma}]|E]$ where $M_{\gamma}$ is the number of minimal vortices that intersect the loop $\gamma$.


\end{proof}

As we have done previously, our goal now is to bound the probability of the event $E$. Once this is done, we can combine this statement with the result of the previous Lemma  \ref{lem:onrareevent} to estimate the leading order contribution of $\mathbb{E}[W_{\gamma}]$ in terms of the number of minimal vortices excited along $\gamma$.

As we have done previously, we will bound from above the probability of the event $E$. It will combine the results of Lemmas \ref{lem:boundprob} and \ref{lem:TrivBound} combined with a combinatorial count of the number of ways to knot.

\begin{thm} \label{thm:mainthmnonabelian}

Provided $10^{24} \alpha_{\beta,\kappa}<1$, with $\alpha_{\beta,\kappa}$ from Lemma \ref{lem:TrivBound}, there is some universal constant not depending on $\beta,\kappa, G$, $H$ or $\gamma$ such that, 
the probability of the complement of the event $E^c$ from Definition \ref{def:goodevent} is bounded by 
\begin{equation} \label{eq:boundgoodevent}
    \mathbb{P}(E^c) \le  O\left(|G||\gamma| \exp[12 \beta(\max_{a \ne 1 \in G} \text{Re}[\psi_p(a) - \psi_p(1)])] \left(\frac{1}{1- 10^{24} \alpha_{\beta,\kappa}}\right)^5\right).
\end{equation}

As an immediate consequence of this estimate and the previous Lemma \ref{lem:onrareevent}, we have that, for the same universal constant as in equation \eqref{eq:boundgoodevent} above,
\begin{equation} \label{eq:finalstatmentmainthm}\begin{aligned}
   & |\mathbb{E}[W_{\gamma}] - \mathbb{E}[A_{\beta,\kappa}^{M_{\gamma}}]\\
   & \le O\left(D|G||\gamma| \exp[12 \beta(\max_{a \ne 1 \in G} \text{Re} [\psi_p(a) - \psi_p(1)])] \left(\frac{1}{1- 10^{24} \alpha_{\beta,\kappa}}\right)^5\right),
\end{aligned}
\end{equation}
where $D$ was the dimension of the unitary representation $\rho$ of $G$.
\end{thm}
\begin{proof}

If we have some configuration $(\sigma,\phi,I)$ that is in $E^c$, then this means that the support of the configuration $(\sigma,\phi,I)$ contains some knot $K$ of size $m$ that cannot be separated from the plaquettes in $\mathcal{P}_{\gamma}$ by some cube.

Lemma 4.3.5  of \cite{SC20} shows that any knot of size $m$ can be contained in a cube of size $3m$. We will give a brief sketch of this fact here. The covering cube will be constructed inductively based on the following principle: a cube $C_1$ of size $3m$ and a cube $C_2$ of size $3n$ that intersect can be covered by a cube $C_3$ of size $3(m+n)$. This large cube $C_3$ can be constructed manually. If we let $x^i_{\text{min},j}$ and $x^i_{\text{max},j}$ be the smallest and larges coordinates, respectively, in the $i$th dimension of cube $C_j$, then we may set $x^{i}_{\text{min},3}= \min(x^{i}_{\text{min,1}}, x^{i}_{\text{min},2})$ and $x^{i}_{\text{max},3}=\max( x^{i}_{\text{max},1},x^{i}_{\text{max},2})$. Since the cubes $C_1$ and $C_2$ intersect, we must have that $|x^{i}_{\text{max},3}-x^{i}_{\text{min},3} | \le m+n$. This is a manual construction of $C_3$.  Now a single connected vortex of size $m$ can be contained in a cube of size $3m$. Since a knot is union of vortices whose containing cubes intersect each other, we can iteratively apply our covering construction on cubes and cover our knot of size $m$ by a cube of size $3m$.

If we let $S_{3m}^{\gamma}$ be the set of plaquettes $p$ such that the cube of size $3m$ centered around $p$ intersects $\gamma$, we see that a knot of size $m$ that cannot be separated by $\mathcal{P}_{\gamma}$ by some cube must intersect a plaquette $p$ in $S_{3m}^{\gamma}$. Notice that the size of $S_{3m}^{\gamma}= O(|\gamma| m^4)$.

Now, Lemma 4.3.4 of \cite{SC20} asserts that the number of knots $K$ of size $m$ that can contains any given plaquette $p \in \Lambda_N$ is less than $(10^{24})^m$.

By applying Lemma \ref{lem:TrivBound}, we can bound the contribution that there is a knot that is Wilson loop nontrivial of some size $m$ that intersects a plaquette in $S_{3m}^{\gamma}$.
This will be an upper bound of $\mathbb{P}(E^c)$.

We have,
\begin{equation}
\begin{aligned}
    \mathbb{P}(E^c)& \le (2^4 |G|)^6 \exp[12 \beta(\max_{a \ne 1 \in G} \text{Re}[\psi_p(a)- \psi_p(1)])\sum_{m=7}^{\infty} O(m^4|\gamma|) (10^{24})^m \alpha_{\beta,\kappa}^{m-6}\\
    &= O\left(|G||\gamma| \exp[12 \beta(\max_{a \ne 1 \in G} \text{Re}[\psi_p(a) - \psi_p(1)])] \left(\frac{1}{1- 10^{24} \alpha_{\beta,\kappa}}\right)^5\right).
\end{aligned}
\end{equation}


Finally, to derive the final consequence \eqref{eq:finalstatmentmainthm}, we observe the following. For any $g \in G$, we have that $\rho(g)$ is a $D$ by $D$ unitary matrix and has trace in absolute value less than $D$. Thus, the Wilson loop functional $W_{\gamma}(\sigma,\phi,I) = \text{Tr}[\rho(\prod_{e\in \gamma} \sigma_e)]$ necessarily has absolute value less than $D$. 

Also, for any $k$ we see that $A_{\beta,\kappa}^k$ can be represented by $\sum_{g} \rho(g) p_k(g)$ where $p_k$ is some probability distribution over the group $G$. This is the probability distribution of a random walk of $k$ steps on the group $G$ starting from the identity and with movement probability $$P(h \to g h)= \frac{\exp[12\beta\text{Re}(\text{Tr}[\rho(g)]- \text{Tr}[\rho(1)])]}{\sum_{g \ne 1}\exp[12\beta \text{Re}(\text{Tr}[\rho(g)]- \text{Tr}[\rho(1)])]}.$$
We see that $\text{Tr}[A_{\beta,\kappa}^{N_{\gamma}}] \le D$. This allows us to bound $|\mathbb{E}[W_\gamma]|$ and $|\mathbb{E}[\text{Tr}[A_{\beta,\kappa}^{N_{\gamma}}]]$ by $D$ on $E^c$. This derives our final result.

\end{proof}

\subsection{Approximation by a Poisson Random Variable}

AS we have done previously,  our final goal in this section is to approximate the variable $M_{\gamma}:= \sum_{e\in \gamma} \mathbbm{1}[F_{P(e)}]$ as a Poisson random variable. In fact, the proof we have used in Section \ref{sec:PoissonRV} can hold almost word for word in this section. This is due to the fact that the proofs in Section \ref{sec:PoissonRV} did not use the full power of the polymer expansion, but only applied these properties of the polymer expansion to a minimal vortex. These properties of the polymer expansion restricted to minimal vortices still hold here due to Lemma \ref{lem:wellsplit}. Our Poisson approximation Theorem is the same as Theorem \ref{thn:PoisApproxGen}, so we will not reproduce this here. 

We start with the following Corollary, which is our analogue of Corollary \ref{col:boundcondition}. As before, we see that this will be a consequence of Lemma \ref{lem:wellsplit}.

\begin{col}[of Lemma \ref{lem:wellsplit}] \label{col:oflemwellsplit}
    Let $E_1$ and $E_2$ be two sets of edges. Let $E(E_1,E_2)$ be the event that the support of the  configuration $(\sigma,\phi,I)$ has a minimal vortex centered at every edge of $E_1$ and there is no minimal vortex at any edge of $E_2$. Assume that $E(E_1,E_2)$ is non-empty. Then, the probability that the knot decomposition of $(\sigma,\phi,I)$ contains the knot $K$ conditional on the event $E(E_1,E_2)$ is less than $\Phi(K)$. 
\end{col}
\begin{proof}
Let $(\sigma,\phi,I)$ be a configuration in $E(E_1,E_2)$ with knot decomposition $K_1 \cup K_2\cup \ldots \cup K_m$ with $K_j = K$ for some $j$. From Remark \ref{rmk:minvortexisknot}, we see that each minimal vortex,  $P(e_i)$ for $e_i \in E_1$, is some knot $K_i$ in the knot expansion.

We know that $$\frac{1}{|G|^{|V_N|-1}|H/X|^{|V_N|}}\sum_{\supp(\sigma,\phi,I)=K_1 \cup K_2 \ldots \cup K_m} \exp[\mathcal{H}_{N,\beta,\kappa}(\sigma,\phi,I)] = \prod_{i=1}^m \Phi(K_i).$$

Thus, we see that the sum of $\frac{1}{|G|^{|V_N|-1}|H/X|^{V_N}}\exp[\mathcal{H}_{N,\beta,\kappa}(\sigma,\phi,I)]$ over all configurations that contain $K$ in the knot decomposition is

\begin{equation} \label{eq:sometop}
\begin{aligned}
    &\frac{1}{|G|^{|V_N|-1}|H/X|^{V_N}} \sum_{\substack{(\sigma,\phi,I) \in E(E_1,E_2)\\K \in \supp(\sigma,\phi,I))}}\exp[\mathcal{H}_{N,\beta,\kappa}(\sigma,\phi,I)] =\\
    &\hspace{2 cm}\sum_{\substack{\bigcup_{e_i \in \gamma}P(e_i) \cup K \in \mathcal{K}\\
    P(e) \not \in \mathcal{K}, e \in E_2}} \prod_{K_i \in \mathcal{K} } \Phi(K_i).
\end{aligned}    
\end{equation}
Here, the sum is over all valid knot decompositions $\mathcal{K}$( corresponding to some configuration) that contain $P(e_i)$ for each $e_i$ in $E_1$  as well as $K$, while not containing any minimal vortex from $E_2$.

If $\mathcal{K}$ is a valid knot decomposition for some element in $E(E_1,E_2)$ containing $K$, then we see that $\mathcal{K} \setminus K$ is a valid knot decomposition for some element in $E(E_1,E_2)$. In addition, $\Phi(\mathcal{K} \setminus K)= \Phi(\mathcal{K}) \Phi(K)^{-1}$.

Thus, we can bound the partition function $Z(E(E_1,E_2))$ from below as,
\begin{equation} \label{eq:somebot}
    \frac{1}{|G|^{|V_N|-1}|H/Z|^{|V_N|}}Z(E(E_1,E_2)) \ge \sum_{\substack{\bigcup_{e_i \in \gamma}P(e_i) \cup K \in \mathcal{K}\\
    P(e) \not \in \mathcal{K}, e \in E_2}} \prod_{K_i \in \mathcal{K} } \Phi(K_i) \Phi(K)^{-1}.
\end{equation}

Taking the ratio of \eqref{eq:sometop} and \eqref{eq:somebot} shows that the probability when conditioned on $E(E_1,E_2)$ of seeing $K$ in the knot expansion is less than $\Phi(K)$.

\end{proof}

Using the above Corollary, we can derive the following estimates on the quantities $b_1,b_2$ and $b_3$ from Theorem \ref{thn:PoisApproxGen}.
\begin{lem} \label{lem:PoisNonAbelian}
Assume we satisfy the conditions of Theorem \ref{thm:mainthmnonabelian}.

Define the constant $\mathfrak{c}$ as,
\begin{equation} \label{def:mathfrakc}
    \mathfrak{c}:=6 \sum_{m=6}^{\infty}(10^{24})^m \alpha_{\beta,\kappa}^m = 6 (10^{24} \alpha_{\beta,\kappa})^6 \frac{1}{1- 10^{24} \alpha_{\beta,\kappa}}.
\end{equation}
For $\mathfrak{c}<1$, e have the following estimates on $b_1,b_2,b_3$ and $\lambda= \mathbb{E}[M_{\gamma}]$.

\begin{equation}
\begin{aligned}
    &b_1 \le |\gamma| \Phi(P(e))^2,\\
    &b_2 =0,\\
    &b_3 \le\mathfrak{c} |\gamma| \Phi(P(e)) ,\\
    & |\lambda - |\gamma|\Phi(P(e))| \le \mathfrak{c}|\gamma|\Phi(P(e)).
\end{aligned}    
\end{equation}

As a consequence of these estimates, we have the following bound on the total variation distance between $M_{\gamma}$ and a Poisson random variable $X$ with parameter $|\gamma|\Phi(P(e))$.
\begin{equation}
 \begin{aligned}
     d_{TV}(\mathcal{L}[M_{\gamma}], \text{Poisson}(|\gamma|\Phi(P(e))) \le|\gamma|(\Phi(P(e)))^2 +2\mathfrak{c}|\gamma|\Phi(P(e)) .
 \end{aligned}
 \end{equation}
 
 Again, we remark that $\Phi(P(e))$ does not depend on the edge $e$. It is a shorthand for $\sum_{g \ne 1} \exp[12 \beta \text{Re}(\text{Tr}[\rho(g)] - \text{Tr}[\rho(1)])]$.

\end{lem}
\begin{proof}
By definition $b_2=0$. We can also bound $b_1$ by $|\gamma|\Phi(P(e))^2$ recalling that $\Phi(P(e))$ does not depend on the value of the specific minimal vortices.

As always, the main difficulty is to estimate the value of $b_3$. Fix some edge $e$. Let $E_1$ be some choice of edges in $\gamma\setminus B_e$ and $E_2$ be the remaining edges in $\gamma \setminus B_e$ not including the edges of $E_1$. We condition on the event $E(E_1,E_2)$ from Corollary \ref{col:oflemwellsplit}. 

Now, $E(E_1,E_2)$ can be divided into two types of events,
\begin{enumerate}
    \item 
    $G(E_1,E_2)$: These are configuration $(\sigma,\phi,I)$ in $E(E_1,E_2)$ whose support has a knot decomposition $K_1 \cup K_2 \ldots \cup K_m$ such that no knot $K_i$ intersects a plaquette in the minimal vortex $P(e)$.
    \item $B(E_1,E_2)$: These are events in $E(E_1,E_2)$ that are not in $G(E_1,E_2)$. If a configuration is in $B(E_1,E_2)$, then its support contains a knot that intersects some plaquette in $P(e)$.
\end{enumerate} 

If we let $Z(G(E_1,E_2))$ and $Z(B(E_1,E_2))$ be the partition functions corresponding to the events $G(E_1,E_2)$ and $B(E_1,E_2)$, we see that $\frac{Z(E(E_1 \cup \{e\},E_2))}{Z(E(E_1,E_2)}$ is the probability that $P(e)$ is a minimal vortex in the support of the configuration conditional on the event $E(E_1,E_2)$. In addition, $Z(E(E_1,E_2))= Z(G(E_1,E_2))+ Z(B(E_1,E_2))$ and $Z(G(E_1,E_2))= \Phi(P(e)) Z(E(E_1 \cup \{e\},E_2)$.

Once we show that $Z(B(E_1,E_2))$ is small relative to $Z(E(E_1,E_2))$, we will be done. This involves bounding the probability that there is a knot in the knot decomposition that intersects the minimal vortex $P(e)$. We have done something similar in the course of the proof of Theorem \ref{thm:mainthmnonabelian}. 

There are at most $(10^{24})^m$ knots of size $m$ that could contain any given plaquette, and from Lemma \ref{lem:TrivBound} and Corollary \ref{col:oflemwellsplit}, the probability of a knot excitation of size $m$ is at most $\alpha_{\beta,\kappa}^m$. Summing this over all knots of size $m$, we see that the probability that there is a knot that intersects one of the plaquettes of $P(e)$ is less than $\mathfrak{c}$ from \eqref{def:mathfrakc}.

Thus, we see that $Z(B(E_1,E_2)) \le \mathfrak{c} Z(E(E_1,E_2))$. Provided that $\mathfrak{c}<1$, we see that,
\begin{equation}
\begin{aligned}
    &\frac{Z(E(E_1 \cup \{e\},E_2))}{Z(E(E_1,E_2)) } =   \frac{Z(E(E_1 \cup \{e\},E_2))}{Z(G(E_1,E_2))+ Z(B(E_1,E_2)) }  \\ & \ge 
    (1-\mathfrak{c}) \frac{Z(E(E_1 \cup \{e\},E_2))}{Z(G(E_1,E_2)) } \ge (1- \mathfrak{c})\Phi(P(e)).
\end{aligned}
\end{equation}

This proves the lower bound of $(1- \mathfrak{c})\Phi(P(e))$ of  $\mathbb{E}[\mathbbm{1}[F_{P(e)}|E(E_1,E_2)]$. Recalling the previous lower bound of $\Phi(P(e))$ for $\mathbb{E}[\mathbbm{1}[F_{P(e)}|E(E_1,E_2)]$
 we see that $\mathbb{E}[\mathbbm{1}[F_{P(e)}|E(E_1,E_2)]- \mathbb{P}(F_{P(e)})| \le \mathfrak{c} \Phi(P(e))$.

Since this is true for all sets $E(E_1,E_2)$, we can remove the conditioning and see that our estimate on $b_3$ is,
\begin{equation}
    b_3 \le |\gamma| 
    \mathfrak{c}\Phi(P(e)).
\end{equation}

Observe that $\mathbb{P}(F_{P(e)})$
also satisfies the same upper and lower bounds of $\Phi(P(e))$ and $(1- \mathfrak{c})\Phi(P(e))$ we have shown earlier. Thus, we also have that,

Similarly, this shows that,
\begin{equation}
    |\lambda - |\gamma| \Phi(P(e))| = |\gamma|| \mathbb{P}(F_{P(e)}) - \Phi(P(e))|\le  \mathfrak{c}|\gamma| \Phi(P(e)).
\end{equation}

We can now apply Theorem \ref{thn:PoisApproxGen}
 to assert that,
 \begin{equation}
     d_{TV}(\mathcal{L}[M_{\gamma}], \text{Poisson}(\lambda)) \le |\gamma|(\Phi(P(e)))^2 +\mathfrak{c}|\gamma|\Phi(P(e)).
 \end{equation}
 
 Since the total variation distance of two Poisson random variables is bounded by the difference of their expectations, we also have,
 \begin{equation}
 \begin{aligned}
    & d_{TV}(\mathcal{L}[M_{\gamma}], \text{Poisson}(|\gamma|\Phi(P(e)))\\& \le |\gamma|(\Phi(P(e)))^2 +\mathfrak{c}|\gamma|\Phi(P(e))+|\lambda- |\gamma|\Phi(P(e))| \\& \le|\gamma|(\Phi(P(e)))^2 +2\mathfrak{c}|\gamma|\Phi(P(e)) ,
 \end{aligned}
 \end{equation}
 as desired.
 
\end{proof}

Since we can choose a coupling such that $\mathbb{P}(X \ne M_{\gamma}) = d_{TV}(X, M_{\gamma}),$ where $X$ is a Poisson random variable with parameter $|\gamma|\Phi(P(e))$, we automatically have our main result by combining Theorem \ref{thm:mainthmnonabelian} and Lemma \ref{lem:PoisNonAbelian}.

\begin{thm} \label{thm:mainthm3}
Assume that the conditions of Theorem \ref{thm:mainthmnonabelian} and Lemma \ref{lem:PoisNonAbelian} hold,
then we have that,
\begin{equation}
\begin{aligned}
    &|\mathbb{E}[W_{\gamma}]- \mathbb{E}[\text{Tr}[A_{\beta,\kappa}^{X}]]| \\& \le  O\left(|G||\gamma| \exp[12 \beta(\max_{a \ne 1 \in G} \text{Re}[\psi_p(a) - \psi_p(1)])] \left(\frac{1}{1- 10^{24} \alpha_{\beta,\kappa}}\right)^5\right)\\
    & + |\gamma|(\Phi(P(e)))^2 +2\mathfrak{c}|\gamma|\Phi(P(e)).
\end{aligned}
\end{equation}
Here, $X$ is a Poisson random variable with parameter $\mathbb{E}[X] = |\gamma|\Phi(P(e))$ for a minimal vortex $P(e)$.
\end{thm}

\section{Decorrelation Estimates and a More Precise Expansion}\label{sec:decorrelation}

As promised, this section will give a more precise expansion of the main order terms in the Wilson loop. In order to reduce technicalities and focus on the innovations, we will assume that our group $G$ is abelian for now.

One technical issue with the argument in section \ref{sec:nonabelianhighdisorder} is that the presence of the Higgs boson could cause long range correlations. In order to bound these contributions, we had to introduce the random current expansion and treat current excitations like knotted  plaquette expansions. Due to the knotting property, it is hard to split a configuration $\sigma$ into disjoint parts in general. However, we observe that there is more power in controlling this splitting when we are dealing with minimal vortices.

In this section, we will reperform the analysis of section \ref{sec:nonabelianhighdisorder} specifically taking advantage of the presence of the minimal vortices along the loop $\gamma$.  We can express our Hamiltonian as
\begin{equation} \label{eq:newToyHam}
H_{N,\beta,\kappa}(\sigma,\phi):= \sum_{p \in P_N} \beta(\rho((\td \sigma)_p) - \rho(1)) + \sum_{e=(v,w) \in E_N} \kappa (f(\sigma_e,\phi_v \phi_w^{-1}) - f(1,1)),
\end{equation}
where $f$ is our shorthand for the action of the Higgs boson and $\rho$ is a 1-dimensional representation.

To proceed along our analysis, we first need to redefine our notion of support of configurations without immediately appealing to the random current expansion.

\begin{defn} \label{def:SupportLowDisorder}
Given a configuration $\{\sigma\}$ on the set of edges $E_N$, the support of our plaquette is
\begin{equation}
    \supp(\{\sigma_e\})= \{p \in P_N: (\td \sigma)_p \ne 1 \}.
\end{equation}


We now use the notation $\mathcal{S}$ to denote configurations of gauge fields $\{\sigma_e\}$. We treat configurations of the Higgs boson in a matter differently from the gauge fields, so we do not need notation to include the Higgs boson.

The support of any configuration $\mathcal{S}$ can be split into a disjoint union $V_1 \cup V_2 \cup \ldots \cup V_N$ where each $V_i$ is a maximally connected set of plaquettes(called a vortex).

We see also that the smallest possible size of a vortex is $12$ oriented plaquettes (6 plaquette pairs $\{p,-p\}$). One of the simplest ways to generate this vortex consists of a single excited edge with $\tilde{\sigma}_e \ne 1$, and all other edges are set to $\tilde{\sigma}_e=1$.

\end{defn}

This  definition of the support is the same as one corresponding to a pure gauge field. Since only exciting the plaquettes, e.g. setting $(\td \sigma)_p \ne 1$ has an effect on the Wilson loop action, we would expect this definition of support to be a more natural description of the changes in the Wilson loop action.


We now define our good set in the following definition.

\begin{defn} \label{def:rareevent}
Fix a Wilson loop $\gamma$ and some value $K$, whose value will be specified later. We let $B_K$ be the set of all plaquettes that are at distance at most $K$ from some edge $e$ in $\gamma$. (More formally, for each edge $e$ in $\gamma$, let $B_K^e$ be the cube of side length $2K$ that is centered around the edge $e$. Then $B_K = \cup_{e \in \gamma} B_K^e$). We see that $|B_K| \le |\gamma| (2K)^d$ .

$\tilde{E}$ will be an event whose complement we will show to be rare.
$\tilde{E}$ is defined as the complement of the union of two events $\tilde{E}^c= \tilde{E}_1 \cup \tilde{E}_2$, which we will describe as follows.

$\tilde{E}_1$ is the event that the support of the configuration contains a knot $\mathcal{K}$ in the knot decomposition of size $\ge 7$ that cannot be separated from the plaquettes of  $B_K$ by some cube $C(\mathcal{K})$.

$\tilde{E}_2$ is the event that the vortex decomposition of $\mathcal{S}$ has at least two minimal vortices $M_1$ (centered around $e_1$) and $M_2$ (centered around $e_2$) such that $e_1$ is an edge in $\gamma$ and  $e_2$ is a boundary edge of some plaquette in $B_K^{e}$. If two such minimal vortices are of distance greater than $K$ from each other, then we will call them $K$-separated.  

\end{defn}

In this section, $K$ will not refer to a knot, but instead our decorrelation distance. We will instead use $\mathcal{K}$ for the few times we refer to a knot.
We see that in the complement of the event $E$, either there are no vortices that intersect $\gamma$ or, if there are intersections of vortices with $\gamma$, these must be minimal vortices that are well-spaced (at least distance $K$) from each other. 

On the event $E$, we will assert that the main contribution to the Wilson loops will be from the minimal vortices that are centered around edges of $\gamma$, much as in the cases we have considered previously. The only technicality is that we have to ensure that the excitations have sufficient distance from each other in order to get some independence behavior.

Unfortunately, we do not see any direct way to bound the probability of the event $\tilde{E}^c$ without appealing to the random current expansion of section \ref{sec:nonabelianhighdisorder}. We see that that if the configuration $(\sigma,\phi)$ is in the event $\tilde{E}_1$ from Definition \ref{def:rareevent}, we see that $(\sigma,\phi,I)$
 must contain a knot $\mathcal{K}$ that intersects the region $B_K(\gamma)$ for any choice of $I$. Thus, we can bound the event $\tilde{E}_1$ by an analog of the event $E^c$ from Definition \ref{def:goodevent}. Since we now deal with a neighborhood of size $K$ around $\gamma$, a simple union bound shows that $\mathbb{P}_{H_N}(\tilde{E}_1)$ can be no more than $O(K^4\mathbb{P}_{\mathcal{H}_N}(E^c))$.
 
 The event $E_2$ is bounded in a different way. This specifically takes advantage of the fact that we are dealing with minimal vortices. 
 Before we do this, we define the notion of $\Phi_{UB}(V)$ for a general vortex $V$.

 \begin{defn} \label{def:phiub}
 Let $V$ be a vortex contained in a cube $C(V)$. We define $\Phi_{UB}(V)$ as follows,
 \begin{equation}
 \begin{aligned}
     \Phi_{UB}(V) = & \sum_{\supp(\sigma,\phi)=V}\prod_{p \in V} \exp[\beta(\rho((\td \sigma)_p) - \rho(1))] \\ &\prod_{e \in C(V)} \exp[\kappa (\max_{(a,b)} \text{Re}[ f(a,b)] - \min_{(c,d)} \text{Re}[f(c,d)])].
 \end{aligned}
 \end{equation}
 The product $e \in C(V)$ is over all oriented edges found in the cube $C(V)$.
 \end{defn}
 
 We see that $\Phi_{UB}(P)$ forms an upper bound in a very weak sense.

 \begin{lem}\label{lem:LowDisorderDecomp}
Under the Hamiltonian \eqref{eq:newToyHam} we are considering,
the probability that we will see a configuration whose support contains a given vortex $V$
and this $V$ is separated by a cube $C(V)$ containing $V$ from all other vortices in the support is less than $\Phi_{UB}(V)$. 
\end{lem} 
\begin{proof}

Let $\mathcal{S}$ be a configuration whose support $S$ contains $V$ and $V$ is separated from all other vortices in the support by a cube $C(V)$ containing $V$. The support of $\mathcal{S}$ its support can be split as $\supp(\mathcal{S})= V \cup X$, where $X $ is a plaquette set in $C(V)^c$.

Fix a simultaneous spanning tree $\mathcal{T}$ of $C(V)$, $C(V)^c$ and the boundary of $C(V)$.

The configuration $\sigma$ can be uniquely gauged with respect to $\mathcal{T}$ such that the gauged configuration $\tilde{\sigma}$ can be split as $\tilde{\sigma}= \tilde{\sigma}_1 \tilde{\sigma}_2$ where $\tilde{\sigma}_1 =1$ for all edges in $C(V)^c$ and $\tilde{\sigma}_2 =1$ for all edges in $C(V)$.


With this division, 
we can express, $\sum_{\substack{\supp(\sigma,\phi)=V \cup X}}\exp[H_{N,\beta,\kappa}(\sigma,\phi)]$ as follows:

\begin{equation} \label{eq:somepartfunc}
\begin{aligned}
   \mathcal{Z}(X \cup V):=&\sum_{\supp(\sigma,\phi)=V \cup X}\exp[H_{N,\beta,\kappa}(\sigma,\phi)]\\  =& \frac{1}{|G|}\sum_{\substack{\supp(\tilde{\sigma})= V \cup X\\ \tilde{\sigma} \text{ gauged w.r.t. }\mathcal{T}}} \prod_{p \in X \cup V} \exp[\beta(\rho(d \tilde{\sigma}) - \rho(1))]\\& \sum_{\eta,\phi}
   \prod_{e \in E_N} \exp[\kappa f(\eta_v \tilde{\sigma}_e \eta_w^{-1},\phi_v\phi_w^{-1})].
   \end{aligned}
\end{equation}

What we have essentially done is, for every configuration, $\sigma$, find the version that is gauged fixed with respect to $\mathcal{T}$ and reintroduce the removed gauge fixing as a sum with respect to a new field $\eta:V_N \to G$ which operates on $\tilde{\sigma}$ as $\tilde{\sigma}_e \to \eta_x \tilde{\sigma}_e \eta_y^{-1}$. However, instead of fixing some value of $\eta_b$ to  be $1$ at a basepoint, we introduce a global gauge transformation that multiplies each $\eta$ by some group element  in $G$.

Now, we can use the fact that $\tilde{\sigma}_e$ splits into $(\tilde{\sigma}_1)_e$ and $(\tilde{\sigma}_{2})_e$ bijectively with respect to the spanning tree $\mathcal{T}$.

We see that we have the expression,
\begin{equation}
\begin{aligned}
    &\mathcal{Z}(X \cup V) =  \frac{1}{|G|} \sum_{\substack{\supp(\tilde{\sigma}_1) =V\\\tilde{\sigma}_1 \text{ gauged to }\mathcal{T}}} \sum_{\substack{ \supp(\tilde{\sigma}_2)= X\\ \tilde{\sigma}_2 \text{ gauged to }\mathcal{T}}} \sum_{\eta_v ,\phi_v} \prod_{e \in E_N} \exp[\kappa f(\eta_v \eta_w^{-1},\phi_v\phi_w^{-1})]\\  & \hspace{0.8 cm} \times \prod_{p \in X} \exp[\beta (\rho((\td \tilde{\sigma}_2)_p) - \rho(1) )] \\ & \hspace{0.8 cm} \times  \prod_{e =(v,w) \in C(V)^c} \exp[\kappa (f(\eta_v (\tilde{\sigma}_2)_e \eta_w^{-1},\phi_v \phi_w^{-1}) - f(\eta_v \eta_w^{-1},\phi_v\phi_w^{-1}))]
    \\ & \hspace{0.8 cm} \times \prod_{p \in V } \exp[\beta (\rho((\td \tilde{\sigma}_1)_p) -\rho(1))] \\ & \hspace{0.8 cm} \times\prod_{e=(v,w) \in C(V) \setminus \delta C(V)} \exp[\kappa (f(\eta_v (\tilde{\sigma}_1)_{e} \eta_w^{-1},\phi_v \phi_w^{-1})- f(\eta_v \eta_w^{-1},\phi_v\phi_w^{-1}))]\\
    & \hspace{0.8 cm} \times \prod_{e =(v,w) \in \delta C(V)} \exp[\kappa (f(\eta_v(\tilde{\sigma}_1 \tilde{\sigma}_2)_e \eta_w^{-1},\phi_v\phi_w^{-1}) - f(\eta_v(\tilde{\sigma}_2)_e \eta_w^{-1},\phi_v\phi_w^{-1}))].
\end{aligned}
\end{equation}

In the last two lines, we upper bound the exponential involving $\kappa$ and $f$ by $\exp[\kappa(\max_{a,b} \text{Re}f(a,b) -\min_{c,d} \text{Re} f(c,d))]$ (observing that that product in $C(V)$ will include the product over both $e$ and $-e$ as oriented edges if any edge $e$ were in $C(V)$). After this replacement, the sum over $\tilde{\sigma}_1$ splits from all other sums. This uses the fact that the resulting sum over $\tilde{\sigma_1}$ does not depend on the choice of the spanning tree $\mathcal{T}$. We can replace the last two lines above by $\Phi_{UB}(V)$.



We see that
\begin{equation}
\begin{aligned}
    &\mathcal{Z}(X \cup V) \\ & \le    \frac{1}{|G|} \Phi_{UB}(V)\sum_{X \text{ disjoint from} V} \sum_{ \supp(\tilde{\sigma}_2)= X} \sum_{\eta_v ,\phi_v} \prod_{e \in E_N} \exp[\kappa f(\eta_v\eta_w^{-1},\phi_v\phi_w^{-1})]\\ & \times \prod_{p \in X} \exp[\beta (\rho((\td \tilde{\sigma}_2)_p) - \rho(1) )] \\ & \prod_{e =(v,w) \in E(X)} \exp[\kappa (f(\eta_v (\tilde{\sigma}_2)_e \eta_w^{-1},\phi_v \phi_w^{-1}) - f(\eta_v \eta_w^{-1},\phi_v\phi_w^{-1}))]\\
    & \le \Phi_{UB
}(V) \mathcal{Z}(X ).    
\end{aligned}
\end{equation}


We see that the probability of finding $V$ in the support, while $V$ is separated by $C(V)$ from the rest of the support, is bounded by,
\begin{equation}
\frac{\sum_{X \text{ is separated from } C(V)} \mathcal{Z}(X \cup V)}{\sum_{X \text{ is separated from  }C(V)} \mathcal{Z}(X)}
\le \Phi_{UB}(V),
\end{equation}
as desired. 

\end{proof}

Just as we can always split away a minimal vortex in the knot decomposition, we can always apply the above lemma to a minimal vortex. In fact, since we can control the number of edges in 
$C(P(e))$ by the number of plaquettes in our minimal vortex, $\Phi_{UB}(P(e))$ is not a very bad bound when only considering minimal vortices.

A consequence of the above lemma is that the probability that our configuration has $P(e)$ and $P(e')$, where $P(e)$ and $P(e')$ are minimal vortices that are compatible with each other is $\Phi_{UB}(P(e))^2$.

We can bound the probability of $\mathcal{E}_2$ by performing a union bound over pairs of minimal vortices $P(e)$ and $P(e')$ where $e$ is an edge of $\gamma$
 and $e'$ is an edge in the box of size $K$ centered around the edge $e$.
 We see that $\mathbb{P}_{H_N}(\mathcal{E}_2) \le O( K^4 |\gamma| \Phi_{UB}(P(e))^2)$, where $O$ is a universal constant not depending on parameters $\beta,\kappa, \gamma$ or $G$. 
 
 We can combine our discussion into the following lemma,
 \begin{lem}\label{lem:rarelowdisorder}
 Assume the conditions of Theorem \ref{thm:mainthmnonabelian}. We have that,
 \begin{equation}
 \begin{aligned}
    \mathbb{P}(\tilde{E}^c)&\le O\left(K^4|G||\gamma| \exp[12 \beta(\max_{a \ne 1 \in G} \text{Re}[\rho(a) - \rho(1)])] \left(\frac{1}{1- 10^{24} \alpha_{\beta,\kappa}}\right)^5\right)\\
    &+ O(K^4 |\gamma| \Phi_{UB}(P(e))^2).
\end{aligned}
 \end{equation}
 As mentioned earlier, the implied constants in $O$ do not depend on $|G|,|H|,\beta,\kappa$ or $\gamma$.
 
 Furthermore, an easy bound on $\Phi_{UB}(P(e))$ is $\mathcal{B}^6$ where $\mathcal{B} $ is defined as,
 \begin{equation}
    \mathcal{B}:=|G|^4\exp[ 2\beta  ( \max_{a \ne 1} \text{Re}[\rho(a) - \rho(1)]) ] \exp[8 \kappa ( \max_
    {a,b} \text{Re}[f(a,b)] - \min_{c,d} \text{Re}[f(c,d)])].
\end{equation}

This is due to the fact that we can restrict our product over $e \in C(V)$ in the definition of $\Phi_{UB}$ from Definition \ref{def:phiub} to those edges that bound a plaquette in the minimal vortex. This will be at most $4$ unoriented edges for every unoriented plaquette in the minimal vortex. Finally, recall that our minimal vortex contains 6 unoriented plaquettes.
 \end{lem}
 




At this point, we can now discuss the computation of Wilson expectations when conditioned on the set $\tilde{E}$, but we must first make an aside to the decorrelation estimates that we will use for our calculation.

\subsection{Minimal Vortices and Decorrelation} \label{subsec:decorrelation}

As we have seen in the previous sections, to see the effect of disorder in the gauge field, it was important to introduce to auxiliary field $\eta:V_N \to G$. The main difficulty is that the presence of the two fields $\eta$ and $\phi$ can create correlations over large distances. However, for small $\kappa$, we can expect we can establish decorrelation estimates that would allows us to express Wilson loop espectations as a product over nearly independent minimal vortices. This section quantifies these desired decorrelation estimates.


To this end, we start by giving some definitions.

\begin{defn} \label{def:Decor}
Consider the Hamiltonian 
\begin{equation}
    K_N:= \sum_{e=(v,w) \in E_N} \kappa f(\eta_v \eta_w^{-1}, \phi_v\phi_w^{-1}).
\end{equation}

We say that the Hamiltonian $K_N$ satisfies decorrelation estimates if the following is true.  Let $\mathcal{V}$ be some set of vertices in $V_N$ Let $B_K(\mathcal{V})$ be set consisting of all vertices of distance $K$ from $\mathcal{V}$. Let $S_1$ and $S_2$ be two configurations of $\eta_v,\phi_w$ that differ on a single boundary point on $B_K(\mathcal{V})$. Then, for any configuration $\hat{\eta}_v,\hat{\phi}_v$ for $v$ the vertices on $\mathcal{V}$,there exists some constants $c$, $K_0$ such that for $K \ge K_0$, we have the inequality,
\begin{equation} \label{eq:decorrelation}
    |\mathbb{P}_{K_N}(\eta_v= \hat{\eta_v}, \phi_v= \hat{\phi}_v| S_1) - \mathbb{P}_{K_N}(\eta_v= \hat{\eta}_v,\phi_v= \hat{\phi}_v|S_2)| \le |\mathcal{V}| e^{-cK}.
\end{equation}

\end{defn}

Let us first discuss some consequences of the decorrelation estimate.
\begin{lem} \label{lem:decoresti1}
Assume that the Hamiltonian $K_N$ satisfies the decorrelation estimate.

Let $V_1$ be some set of vertices and $V_2$ be some set of vertices outside $B_K(V_1)$. Let $\hat{\eta}_v^1, \hat{\phi}_v^1$ be some choice of configurations for  the vertices in $V_1$ and $\hat{\eta}_v^2$, $\hat{\phi}_v^2$ be some choice of configurations for vertices in $V_2$. Then, we have the following estimate,
\begin{equation} \label{eq:decor1}
\begin{aligned}
   | &\mathbb{P}_{K_N}(\eta_v= \hat{\eta}^1_v,\phi_v = \hat{\phi}^1_v, v \in V_1, \eta_w= \hat{\eta}^2_w, \phi_w= \hat{\phi}_w^2, w \in V_2)\\
   &- \mathbb{P}_{K_N}(\eta_v= \hat{\eta}^1_v,\phi_v = \hat{\phi}_v^1, v \in V_1) \mathbb{P}_{K_N}(\eta_w= \hat{\eta}^2_w, \phi_w= \hat{\phi}_w^2, w \in V_2)| \le\\
   &2 |\delta B_K(V_1)| |V_1| e^{-cK} \mathbb{P}_{K_N}(\eta_w= \hat{\eta}^2_w, \phi_w= \hat{\phi}_w^2, w \in V_2) .
   \end{aligned}
\end{equation}
Here, $|\delta B_K(V_1)|$ is the size of the boundary of the set $B_K(V_1)$
\end{lem}
\begin{proof}
Let $V_c$ be the set of vertices outside $B_K(V_1)$ and not including $V_2$.

We see that summation, we have that,
\begin{equation}
\begin{aligned}
    &\mathbb{P}_{K_N}(\eta_v= \hat{\eta}^1_v,\phi_v = \hat{\phi}_v^1, v \in V_1, \eta_w= \hat{\eta}^2_w, \phi_w= \hat{\phi}_w^2, w \in V_2)\\
   &= \sum_{\tilde{\eta}_a, \tilde{\phi}_a, a \in V_c} \mathbb{P}_{K_N}(\eta_v= \hat{\eta}^i_v,\phi_v = \hat{\phi}_v^i, i=1,2, \eta_a= \tilde{\eta}_a, \phi_a= \tilde{\phi}_a, a \in V_c).
\end{aligned}
\end{equation}

The  quantity on the last line can be computed as a conditional expectation.

\begin{equation} \label{eq:conditioned}
\begin{aligned}
  & \mathbb{P}_{K_N}(\eta_v,\phi_v= \hat{\eta}^1_v, \hat{\phi}_v^1, v \in V_1, \eta_w,\phi_w= \hat{\eta}^2_w, \hat{\phi}_w^2, w \in V_2, \eta_a,\phi_a= \tilde{\eta}_a, \tilde{\phi}_a, a \in V_c)\\& = \mathbb{P}_{K_N}(\eta_v,\phi_v = \hat{\eta}^1_v, \hat{\phi}_v^1, v \in V_1| \eta_w,\phi_w= \hat{\eta}^2_w, \hat{\phi}_w^2, w \in V_2, \eta_a,\phi_a= \tilde{\eta}_a, \tilde{\phi}_a, a \in V_c)\\
  &\times \mathbb{P}_{K_N}( \eta_w,\phi_w= \hat{\eta}^2_w, \hat{\phi}_w^2, w \in V_2, \eta_a,\phi_a= \tilde{\eta}_a,\tilde{\phi}_a, a \in V_c).
 \end{aligned}  
\end{equation}

Because our Hamiltonian $K_N$ acts on nearest neighbors, the conditional expectation on the second time only depends on the values at the boundary $|\delta B_K(V_1)|$.

From the inequality \eqref{eq:decorrelation}, one can show that there is a constant $E$ such that for any any boundary condition $B$ of the form,
\begin{equation}
    B=\{\eta_w,\phi_v:\eta_w,\phi_w= \hat{\eta}_w^2, \hat{\phi}_w^2, w \in V_2, \eta_a,\phi_a= \tilde{\eta}_a, \tilde{\phi}_a, a \in V_c\},
\end{equation} we have
$$
\begin{aligned}
&|\mathbb{P}_{K_N}(\eta_v,\phi_v = \hat{\eta}^1_v, \hat{\phi}_v^1, v \in V_1|B ) - E| \le|V_1| |\delta B_K(V_1)| e^{-cK}.
\end{aligned}
$$

To see the derivation, one can first fix some arbitrary boundary condition $B$. From any other boundary condition, $\tilde{B}$, one needs to change at most $|\delta B_K(V_1)|$ terms. Thus, by applying the triangle inequality at most $|\delta B_K(V_1)|$ times, we can derive the last line.

Since we have the relation,
\begin{equation}
    \sum_{B}\mathbb{P_{K_N}}(\eta_v= \hat{\eta}_v^1, \phi_v= \hat{\phi}_v^1, v \in V_1|B) \mathbb{P}_{K_N}(B) = \mathbb{P}_{K_N}(\eta_v= \hat{\eta}_v^1, \phi_v= \hat{\phi}_v^1),
\end{equation}
where $B$ is a sum over all possible boundary conditions, we must necessarily have that $|E - \mathbb{P}_{K_N}(\eta_v= \hat{\eta}_v^1, \phi_v= \hat{\phi}_v^1)| \le 2 |V_1||\delta B_K(V_1)| e^{-cK}$.

Substituting back our relation on the conditional probability in equation \eqref{eq:conditioned}, we see that we can upper and lower bound the conditional probability with $\mathbb{P}_{K_N}(\eta_v= \hat{\eta}_v^1,\phi_v= \hat{\phi}_v^1) \pm 2 |\delta B_K(V_1)| e^{-cK}$. Then, we can resum the expression over $V_c$. We see we derive the desired inequality,
\begin{equation}
\begin{aligned}
    | &\mathbb{P}_{K_N}(\eta_v= \hat{\eta}^1_v,\phi_v = \hat{\phi_v}^1, v \in V_1, \eta_w= \hat{\eta}^2_w, \phi_w= \hat{\phi}_w^2, w \in V_2)\\
   &- \mathbb{P}_{K_N}(\eta_v= \hat{\eta}^1_v,\phi_v = \hat{\phi_v}^1, v \in V_1) \mathbb{P}_{K_N}(\eta_w= \hat{\eta}^2_w, \phi_w= \hat{\phi}_w^2, w \in V_2)| \le\\
   &2|V_1| |\delta B_K(V_1)| e^{-cK} \mathbb{P}_{K_N}(\eta_w= \hat{\eta}^2_w, \phi_w= \hat{\phi}_w^2, w \in V_2) .
 \end{aligned}
\end{equation}

\end{proof}

As a corollary of this lemma, we can get decorrelation estimates over a product of different sites, provided the sites are sufficiently distant from each other.
\begin{col} \label{col:DecayEst2}
Suppose $S_1,\ldots,S_m$ are events supported on vertices $V_1,\ldots, V_m$ such that each $V_j$ for $j \ne i$ lies outside the block $B_K(V_i)$ for each $i$. Fix some constant $C$ and assume that that each boundary satisfies $|\delta B_K(V_i)|\le C K^{d-1}$ and $|V_i|\le C$.  

Assume that $K_N$ satisfies the decorrelation estimates.
For some constant $c'<c$  and $K$ sufficiently large depending on $c'$ and $C$ , we have the estimate,
\begin{equation}
    |\mathbb{P}_{K_N}(S_1,S_2,\ldots,S_m) - \prod_{i=1}^m\mathbb{P}_{K_N}(S_i)| \le m \frac{e^{-c'K}}{\min_i \mathbb{P}_{K_N}(S_i)} \prod_{i=1}^m (\mathbb{P}_{K_N}(S_i) + e^{-c'K}).
\end{equation}
\end{col}
\begin{rmk}
This corollary will essentially only be applied to the case where $V_i$'s are minimal vortices. Thus, we do not need to really worry about the value of the constant $C$ and the dependence of $K$ on $C$. 
\end{rmk}

\begin{proof}
The proof of this lemma involves performing some induction.

We let $F_k := \mathbb{P}_{K_N}(S_1,\ldots,S_k) $ and $E_k = |\mathbb{P}_{K_N}(S_1,\ldots,S_k) - \prod_{i=1}^k \mathbb{P}_{K_N}(S_i)|$. By splitting the events $\{S_1,\ldots,S_{k-1}\}$ from $S_k$, we see that applying equation \eqref{eq:decor1}  will show that
\begin{equation}
    F_k \le F_{k-1} (\mathbb{P}_{K_N}(S_k) + e^{-c'K}),
\end{equation}
for sufficiently large $K$. (This allows the exponential factor to decay faster than any polynomial factor of $K$ that appears from $|\delta B_K(V_i)|$ and the constants that appear in $|V_i|$.)
Thus, we can derive,
\begin{equation} \label{eq:Fest}
    F_k \le \prod_{i=1}^k (\mathbb{P}_{K_N}(S_i) + e^{-c'K}),
\end{equation}
for all $k$.

To determine $E_k$, we also split the events $\{S_1,\ldots,S_{k-1}\}$ from $S_k$. We first write $E_k$ as $$|\mathbb{P}(S_1,\ldots,S_k) - \mathbb{P}(S_1,\ldots,S_{k-1})\mathbb{P}(S_k) + \mathbb{P}(S_1,\ldots,S_{k-1}) \mathbb{P}(S_k) - \prod_{i=1}^k \mathbb{P}(S_i)|$$ we see that by applying the triangle inequality, that 
\begin{equation}
    E_k \le |F_k - F_{k-1} \mathbb{P}_{K_N}(S_k)| + \mathbb{P}_{K_N}(S_k) E_{k-1}.
\end{equation}
The decorrelation estimate applied to $|F_k - F_{k-1} \mathbb{P}_{K_N}(S_k)|$ gives that this quantity is less than $F_{k-1}  e^{-c'K}$. We then apply \eqref{eq:Fest} to estimate $F_{k-1}$ and apply the induction hypothesis to $E_{k-1}$. We see that,
\begin{equation}
\begin{aligned}
    E_k \le &  e^{-c'K}\prod_{i=1}^{k-1} (\mathbb{P}_{K_N}(S_i) +  e^{-c'K})\\ &  + (k-1) \frac{e^{-c'K}}{\min_{i \in [1,k-1]} \mathbb{P}_{K_N}(S_i)} \prod_{i=1}^{k-1}(\mathbb{P}_{K_N}(S_i) +  e^{-c'K}) \mathbb{P}_{K_N}(S_k).
\end{aligned}
\end{equation}

$ e^{-c'K}\prod_{i=1}^{k-1} (\mathbb{P}_{K_N}(S_i) + e^{-c'K})$ can be bounded by $\frac{e^{-c'K}}{\mathbb{P}_{K_N}(S_k)} \prod_{i=1}^k (\mathbb{P}_{K_N}(S_i) + e^{-c'K}) $. We can replace the probability in the denominator by the minimum of the probability over all $i$, $\min_{i \in [1,k]} \mathbb{P}_{K_N}(S_i)$.

In addition,
$\frac{e^{-c'K}}{\min_{i\in[1,k-1]} \mathbb{P}_{K_N}(S_i)} \prod_{i=1}^{k-1}(\mathbb{P}_{K_N}(S_i) +  e^{-c'K}) \mathbb{P}_{K_N}(S_k)$ is bounded by $\frac{e^{-c'K}}{\min_{i\in [1,k]} \mathbb{P}_{K_N}(S_i)} \prod_{i=1}^{k}(\mathbb{P}_{K_N}(S_i) +  e^{-c'K})$. Adding up these terms completes the proof.
\end{proof}

Now we can show through a percolation type argument that for sufficiently small $\kappa$, we can prove decorrelation estimates for the Hamiltonian $K_N$.

\begin{thm}
For $\kappa$ sufficiently small, we have the decorrelation estimates as detailed in Definition \ref{def:Decor} for the Hamiltonian $K_N$ 
\end{thm}
\begin{proof}
This is essentially a percolation argument based on the random current representation for the Ising model.

Consider the Hamiltonian $K_N:= \sum_{e=(v,w)} \kappa f(\eta_v \eta_w^{-1},\phi_v\phi_w^{-1})$. Let $c$ be a constant such that  $2 \text{Re}f(\eta_v \eta_w^{-1},\phi_v\phi_w^{-1}) +c >0$ for any choice of  $\eta_v$,$\eta_w$,$\phi_v$ and $\phi_w$ at any given edge $e=(v,w)$. Considering the new Hamiltonian,
\begin{equation}
    \tilde{K}_N:= \sum_{e=(v,w) \in E_N^{U}} \kappa [f(\eta_v \eta_w^{-1},\phi_v\phi_w^{-1}) + \overline{f(\eta_v\eta_w^{-1},\phi_v\phi_w^{-1})} +c],
\end{equation}
where we recall the notion of unoriented edges.  The Hamiltonian $\tilde{K}_N$ pairs up the edges $e$ and $-e$ in order to get a real number. Thus, the measure generated by $K_N$ is the same as the measure generated by $\tilde{K}_N$.

With this in hand, we now define a new probability model based on the random current representation of the Ising model. On this probability space, we have two sets of random variables. The first are the same $(\eta_v,\phi_v)$ configuration variables on the vertices. The marginal distribution of these variables are given by the Hamiltonian $\tilde{K}_N$.

The second is a set of activations on each edge $I_{e}$ for $e \in E_N$. $I(e)$ is a variable that takes non-negative integer values with the following probability distribution. Given a configuration $S =(\eta_v,\phi_v)$ for all vertices $v$, we have the following distribution
\begin{equation}
    P(I(e) = k|S) =\frac{(\kappa (2 \text{Re }f(\eta_v\eta_w^{-1},\phi_v\phi_w^{-1})+c)^k}{k! \exp[\kappa(2 \text{Re}f(\eta_v\eta_w^{-1},\phi_v\phi_w^{-1})+c)]} .
\end{equation}

What is interesting about this representation is that the joint distribution can be represented as a new Hamiltonian as follows,
\begin{equation}
    \mathcal{K}_N(\eta,\phi,I):= \sum_{e=(v,w) \in E_N}   I(e) \kappa (2 \text{Re}f(\eta_v\eta_w^{-1},\phi_v\phi_w^{-1}) +c) - \log I(e)!. 
\end{equation}

The point is, one can marginalize by summing over $I(e)$. The summation over $I(e)$ returns $\exp[\kappa(2 \text{Re}f(\eta_v\eta_w^{-1},\phi_v\phi_w^{-1})+c)]$ as desired.

Now, let $\mathcal{V}$ be a set of vertices with $K$-boundary $B_K(\mathcal{V})$ surrounding it. To understand the effect of boundary conditions applied to $B_K(\mathcal{V})$ on vertices in $\mathcal{V}$, we see it will be better to understand the effect upon conditioning on the value $I(e)$.

For a fixed configuration of values $I(e)$, we let $A$ be the set of activated edges defined as $A:=\{ e\in E_N: I(e) \ne 0\}$. With the set of activated edges defined above, for every vertex $v$, we can define a cluster $Cl(v)$ to be the set of vertices connected to $v$ using only edges of $A$.

Now, we claim that, upon conditioning on the values of $I(e)$. Configurations of $(\eta_v,\phi_v)$ that lie on different clusters are independent of each other. This can be seen by explicitly writing out this conditional probability as a summation,
\begin{equation}
    P(\eta,\phi| I) = \frac{\prod_{e=(v,w)\in E_N} \frac{(\kappa(2 \text{Re}f(\eta_v\eta_w^{-1},\phi_v\phi_w^{-1}) +c))^{I(e)}}{I(e)!}}{\sum_{\eta,\phi}\prod_{e=(v,w)\in E_N} \frac{(\kappa(2 \text{Re}f(\eta_v\eta_w^{-1},\phi_v\phi_w^{-1}) +c))^{I(e)}}{I(e)!}}.
\end{equation}
The partition function in the denominator can split as a product over different clusters. Thus, different clusters once conditioned on $I(e)$ are independent of each other.

Thus, if $v$ is a vertex in $V$, then $v$ can only be affected by the boundary if there is a path using activated edges in $A$ connecting $v$ to the boundary.

Namely, we see that,
\begin{equation}
\begin{aligned}
   & |\mathbb{P}_{K_N}(\eta_v= \hat{\eta_v}, \phi_v= \hat{\phi}_v| S_1) - \mathbb{P}_{K_N}(\eta_v= \hat{\eta}_v,\phi_v= \hat{\phi}_v|S_2)|\\ &  \le \sum_{v \in \mathcal{V}} \mathbb{P}_{\mathcal{K}_N}(v \sim_{A} B_K(v)),
\end{aligned}
\end{equation}
where we use the notation $\sim_{A}$ to denote connection using edges in $A$. Notice that the left hand side is $0$ if $v$ were not connected to the boundary $B_k(v)$ using edges of $A$ and is bounded by $1$ otherwise.

Now, our goal is to derive exponential decay bounds for the percolation type estimate $\mathbb{P}_{\mathcal{K}_N}(v \sim_{A} B_K(v))$. 
To compute this probability, we condition on the configuration $(\eta_v,\phi_v)$ and then compute the resulting percolation process.

For any given configuration $(\eta_v,\phi_w)$, we see that the probability of any given edge $e$ being activated is at most,
$
    1- e^{-\kappa(2 \max_{a,b}\text{Re} f(a,b)+c)}
$. By a monotone coupling on percolation processes, we see that the probability that $v$ is connected to the boundary $B_K(v)$ through edges in $A$ is bounded by the probability that $v$ is connected to the boundary $B_K(v)$ throught the edge-bond percolation process whose probability of edge activation is $1- e^{-\kappa(2 \max_{a,b}\text{Re}f(a,b)+c)}$. At this point, there are standard arguments to show that the probability that $v$ is connected to its boundary is bounded from percolation theory. For example, we can apply Theorem 1.1 Part 2 of \cite{Copin-Tassion}.


\end{proof}

\begin{thm} \label{thm:WilsonLoopReplowdisorder}


Recall the Hamiltonian $K_N$  from Definition \ref{def:Decor} and assume it satisfies said decorrelation estimates.
Let $m(\hat{\phi}_v, \hat{\phi}_w,\hat{\eta}_v,\hat{\eta}_w)$ be the probability under $\lim_{N \to \infty}\langle \cdot \rangle_{K_N}$(i.e., the limiting probability distribution on the infinite lattice) that on on some edge $e=(v,w)$, we will see the configuration $\phi_v = \hat{\phi}_v$ and so on.

We first define the quantity,
\begin{equation}
\begin{aligned}
    &X(g):=\\
    &\sum_{\hat{\phi}_1,\hat{\phi}_2, \hat{\eta}_1, \hat{\eta}_2} m(\hat{\phi}_1, \hat{\phi}_2, \hat{\eta}_1,\hat{\eta}_2) \exp[-2\kappa\text{Re}[f(\hat{\eta}_1 g \hat{\eta}_2^{-1}, \hat{\phi}_1 \hat{\phi}_2^{-1}) -f(\hat{\eta}_1 \hat{\eta}_2^{-1},\hat{\phi}_1 \hat{\phi}_2^{-1}) ]],
\end{aligned}
\end{equation}
which represents the relative change of the Higgs boson action due to setting $\sigma_e=g$.
We then define the quantity,
\begin{equation}
    \mathcal{D}_{\beta,\gamma}:= \frac{\sum_{g \ne 1} \rho(g) \exp[-12 \beta\text{Re}[\rho(g) - \rho(1)]] X(g)}{ \sum_{g \ne 1}  \exp[-12 \beta\text{Re}[\rho(g) - \rho(1)]] X(g)}.
\end{equation}

In addition, define $L$ to be the following constant.
  \begin{equation} \label{eq:defL}
      \begin{aligned}
      &L:=|G|^{24} |H|^{24}\exp[96\kappa(\max_{a,b}\text{Re}f(a,b) - \min_{c,d }\text{Re} f(c,d ))] .
      \end{aligned}
  \end{equation} \eqref{eq:defL}

Provided $N$ is large, the loop $\gamma$ is sufficiently far away from the boundary, and $K$ is  sufficiently large, we have, 
    \begin{equation} \label{eq:conditioning2}   
  \begin{aligned}      &|\mathbb{E}[W_{\gamma} -\mathcal{D}^{M_{\gamma}}_{\beta,\kappa}| \tilde{E}]| \le \mathbb{P}(\tilde{E}^c)^{-1} |\gamma|^2 \mathcal{B}^{6} L K^{3} e^{-cK} \exp[\mathcal{B}^{6}|\gamma|L]\\ & \hspace{0.8 cm}
     + \mathbb{P}(\tilde{E}^c)^{-1} \frac{|\gamma| \mathcal{B}^{6} [(1+e^{-c'K})L] e^{-c'K}}{\min_{\substack{\eta_v,\phi_v\\ v \in \text{ minimal vortex }}} \mathbb{P}_{K_N}(\eta_v,\phi_v)} \exp[|\gamma| \mathcal{B}^{6}[(1+e^{-c'}K) L]].
  \end{aligned}
  \end{equation}

    Here $M_{\gamma}$ is the number of edges on $\gamma$ that form the center of minimal vortex excitations. We also recall the constant $\mathcal{B}$ from Lemma \ref{lem:rarelowdisorder} and the constant $c$ from the decorrlation estimates Definition \ref{def:Decor}.
    
    Similarly, we have,
    \begin{equation}\label{eq:conditioning3}
    \begin{aligned}
    |\mathbb{E}[W_{\gamma} -\mathcal{D}^{M_{\gamma}}_{\beta,\kappa}]| &\le \mathbb{P}(\tilde{E}^c)+ |\gamma|^2 \mathcal{B}^{6} L K^{3} e^{-cK} \exp[\mathcal{B}^{6}|\gamma|L]\\
     &+  \frac{|\gamma| \mathcal{B}^{6} [(1+e^{-c'K})L] e^{-c'K}}{\min_{\substack{\eta_v,\phi_v\\ v \in \text{ minimal vortex }}} \mathbb{P}_{K_N}(\eta_v,\phi_v)} \exp[|\gamma| \mathcal{B}^{6}[(1+e^{-c'}K) L]].
    \end{aligned}
    \end{equation}

\end{thm}
\begin{rmk}

For finite $N$, there are marginal differences between the magnetizations $m(\phi,\eta)$ calculated with respect to the Hamiltonian $K_N$ and its infinite limit. However, provided we consider loops far away from the boundary (say on order $\sqrt{N}$). We would expect such differences to decay on the order $e^{-\sqrt{N}}$. In the course of the proofs that follow, we ignore such marginal differences between the magnetizations computed with respect to $K_N$ and those computed with respect to the infinite limit to simplify the presentation of the core ideas.

\end{rmk}

\begin{rmk} \label{rmk:onK}
The main benefit of the estimate \eqref{eq:conditioning2} is that when $K$ is large, we can supress the factor of $|\gamma|$ in the third and fourth terms of the upper bound, even when $|\gamma|= O (\exp[d\beta])$ for example. Even in this case, we see that we could set $K= O (\beta)$, so we can suppress the factor of $K^d$ with $\mathcal{B}$ when $\beta$ is sufficiently large.
This ensures that with an appropriate choice of $K$, the right hand side of equation \eqref{eq:conditioning2} is sufficiently small.

\end{rmk}

\begin{proof}[Proof of Theorem \ref{thm:WilsonLoopReplowdisorder}]

On the event $\tilde{E}^c$, both $W(\gamma)$ and $\mathcal{D}_{\beta,\gamma}$ are bounded by $1$. We now consider what happens on the high probability event $\tilde{E}$.

  Consider a set $V_1 \cup V_2 \ldots \cup V_N$ that would form the support of an excitation that would belong in $E$. Since this event belongs to $E$, there are two disjoint possibilities:
  \begin{enumerate} \label{enum:dich}
     \item $G^1$:The only excitations in $B_K$ are minimal vortices that intersect the loop $\gamma$ and these minimal vortices are spaced at least distance $K$ from each other.

      \item $G^2:$The support of the excitation does not intersect the loop $\gamma$ and we are not in $G^1$.
     
  \end{enumerate}
  
  When conditioned on $G^2$, we see that $\langle W_{\gamma} \rangle$ is 1 and $\rho(-1)^{M_{\gamma}}$ is also $1$ and, so, the difference is $0$.
  
  Now, let us comment on what we should do when we consider $G^1$.
  We can decompose $G^1$ as follows. First fix some set $K_1:= V_1 \cup V_2\cup\ldots \cup V_m$ of $K$-separated minimal vortices $V_i$ centered  on edges of $\gamma$. We let the event $G^1_{V_1,\ldots,V_m}$ be the event that the support of the configuration consists of minimal vortices $V_1,\ldots,V_m$ and possibly some set $K_2$ that lies on the exterior  of $B_K(\gamma)$.
  
  \begin{lem} \label{lem:WeirdCondition}

  On the event $G^1_{V_1,\ldots,V_m}$, we have that for sufficiently large $K$ and values of $\kappa$  such that $K_N$ satisfies the decorrelation estimates from Definition \ref{def:Decor} that 
  \begin{equation}
  \begin{aligned}
     | \mathbb{E}[W_{\gamma} - \mathcal{D}_{\beta,\kappa}^{M_{\gamma}}| G^1_{V_1,\ldots,V_m}]|& \le |\gamma|K^{3} e^{-cK} L^m \\
     &+ \frac{m e^{-cK}}{\min_{\substack{\eta_v,\phi_v\\ v \in \text{ minimal vortex }}} \mathbb{P}_{K_N}(\eta_v,\phi_v)}  [(1+ e^{-cK}) L]^m.
     \end{aligned}
  \end{equation}

  \end{lem}
  \begin{proof}

  We can choose a spanning tree $\mathcal{T}(V_1,V_2,\ldots,V_m)$ that is a simultaneous spanning tree of the boxes $B_K(V_i)$ of size $k$ centered around each minimal vortex $V_i$ and  the complement of these boxes.
  
  By gauging our configuration with respect to this spanning tree, we see that any configuration $\sigma$ with support $V_1\cup V_2\ldots \cup V_m \cup P$ can be split into $m+1$ parts $\tilde{\sigma}^1,\ldots \tilde{\sigma}^m$ and $\tilde{\sigma}^{\text{rest}}$ such that $\tilde{\sigma}^i$ has its only nontrivial edges, with $\tilde{\sigma}^i \ne 1$, at the center $e_i$ of the minimal vortex $P(e_i)$ forming $\mathcal{V}_i$. Finally, $\tilde{\sigma}^{\text{rest}}$ has its only non-trivial edges in the complement of the boxes $B_K(V_i)$ surrounding the minimal vortex. Furthermore, since we are dealing with an abelian group, the contribution Wilson loop action of $\tilde{\sigma}^{\text{rest}}$ can separate from the contribution from the $\tilde{\sigma}^i$'s. 
  
  \textit{Part 1: Splitting $K_1$ from $K_2$}
  
 
 
 For a general group $G$, this splitting gives us a good way to write down the expression of $W_{\gamma}$ conditional on the event $E(V_1,\ldots,V_m)$ is as a ratio $\frac{\mathcal{N}}{\mathcal{D}}$, where we now write out the definitions of the numerator and the denominator.
At this point in the proof, we are only concerned with splitting the support of a configuration into two parts: the part supported on $V_1 \cup V_2\ldots \cup V_m$ and part supported on the outside of $B_K(\gamma)$. Throughout this part  We will use the notation $\tilde{\sigma}^1$ to represent the part supported on the boxes $B_K(K_1):=B_K(V_1) \cup B_K(V_2) \ldots \cup B_K(V_m)$ and $\tilde{\sigma}^2$ to represent the part supported outside of the union of the aforementioned boxes.  As mentioned before, in the following expressions, each configuration is represented by its unique gauge fixed configuration

  We can write out the denominator as,
  \begin{equation}
  \begin{aligned}
     \mathcal{D}= & \sum_{K_2} \sum_{\substack{\eta_v^1,\phi_v^1\\
     v \in V(K_1)}} \sum_{\substack{ \eta_v^2, \phi_v^2\\  v \in B_K(K_1)^c}}  \sum_{\substack{\eta_v, \phi_v\\ v \not \in V(K_1) \cup B_K(K_1)^c}} \\ & \prod_{e =(v,w)} \exp[\kappa(f(\eta_v \eta_w^{-1}, \phi_v \phi_w^{-1}) - f(1,1))]\\
     & \sum_{ \supp(\tilde{\sigma}^1)=K_1}  \prod_{p \in K_1} \exp[ \beta(\rho(\td(\tilde{\sigma}^1)_p) - \rho(1))]\\
     & \times \prod_{e \in E(K_1)} \exp[\kappa(f(\eta^1_v \tilde{\sigma}^1_e (\eta^1_w)^{-1}, \phi^1_v (\phi^1_w)^{-1})  - f( \eta^1_v (\eta^1_w)^{-1}, \phi^1_v (\phi^1_w)^{-1}))]\\&\sum_{ \supp(\tilde{\sigma}^2)=K_2} \prod_{p \in K_2} \exp[ \beta(\rho(\td(\tilde{\sigma}^2)_p) - \rho(1))]\\
     & \times \prod_{e \in B_K(K_1)^c} \exp[\kappa(f(\eta^2_v \tilde{\sigma}^2_e (\eta^2_w)^{-1}, \phi^2_v (\phi^2_w)^{-1})  - f( \eta^2_v (\eta^2_w)^{-1}, \phi^2_v (\phi^2_w)^{-1}))].
    \end{aligned}
  \end{equation}
  
  We slightly abuse notation here, in the first line, the sum $v \in B_K(K_1)^c$ is a sum over the vertices in $B_K(K_1)^c$ while the sum $e \in B_K(K_1)^c$ is a sum over the edges in $B_K(K_1)^c$. We hope this distinction is always clear in context by the use of $e$ or $v$ and the variables we associate to it.
  Note that this is the sum of $\exp[ H_{N,\beta,\kappa}(\sigma,\phi)]$ for configurations found in $E_1$, i.e. the partition function for $E_1$.
  In the above decomposition, we used the fact that, after the gauging, $\tilde{\sigma}^1$ only takes non-trivial values on the edges that are on the boundary of some minimal vortex. In addition, we used the fact that $\tilde{\sigma}^2$ can only take non-trivial values on the edges that are in the complement of $B_K(K_1)$.

  In the product on the first line, any appearance of $\eta_v$(resp. $\phi_v$) for $v$ a vertex in $K_1$ would be replaced with $\eta_v^1$ (resp. $\phi_v^1$). Similar things happen with $\eta_v$ for $v$ a vertex in $B_K(K_1)^c$. We also remark that in the last line, we know that $\tilde{\sigma}_2$ takes trivial values on $B_K(K_1)$, not necessarily on all of $E(K_2)^c$. This is also why we have to sum up $\eta_v^2,\phi_v^2$ over all all vertices in $B_K(K_1)^c$
  
  We have a similar expansion for the numerator.
  \begin{equation}
  \begin{aligned}
     \mathcal{N}= & \sum_{K_2} \sum_{\substack{\eta_v^1,\phi_v^1\\
     v \in V(K_1)}} \sum_{\substack{ \eta_v^2, \phi_v^2\\  v \in B_K(K_1)^c}}  \sum_{\substack{\eta_v, \phi_v\\ v \not \in V(K_1) \cup B_K(K_1)^c}}\\& \prod_{e =(v,w)} \exp[\kappa(f(\eta_v \eta_w^{-1}, \phi_v \phi_w^{-1}) - f(1,1))]\\
     & \sum_{ \supp(\tilde{\sigma}^1)=K_1}  \langle \tilde{\sigma}^1, \gamma \rangle\prod_{p \in K_1} \exp[ \beta(\rho(\td(\tilde{\sigma}^1)_p) - \rho(1))]\\
     & \times \prod_{e \in E(K_1)} \exp[\kappa(f(\eta^1_v \tilde{\sigma}^1_e (\eta^1_w)^{-1}, \phi^1_v (\phi^1_w)^{-1})  - f( \eta^1_v (\eta^1_w)^{-1}, \phi^1_v (\phi^1_w)^{-1}))]\\&\sum_{ \supp(\tilde{\sigma}^2)=K_2} \langle \gamma,\tilde{\sigma}^2 \rangle  \prod_{p \in K_2} \exp[ \beta(\rho(\td(\tilde{\sigma}^2)_p) - \rho(1))]\\
     & \times \prod_{e \in B_K(K_1)^c} \exp[\kappa(f(\eta^2_v \tilde{\sigma}^2_e (\eta^2_w)^{-1}, \phi^2_v (\phi^2_w)^{-1})  - f( \eta^2_v (\eta^2_w)^{-1}, \phi^2_v (\phi^2_w)^{-1}))].
    \end{aligned}
  \end{equation}
  
  We remark here that, in the end, we can assert that the Wilson loop expectation only depends on the values of $\tilde{\sigma}^1$ rather than on $\tilde{\sigma}^2$.
   This uses the fact that $K_2$ has support that is separated by a rectangle from $\gamma$, this $\tilde{\sigma}^2$ could not possibly contribute to the Wilson action.

  If we define $Z_{K_N}$ to be the partition function associated to the Hamiltonian from Definition \ref{def:Decor} ,
  $
  K_N(\phi,\eta):= \sum_{e \in E_N} \kappa( f(\eta_v \eta_w^{-1}, \phi_v \phi_w^{-1}) - f(1,1)),
  $
  and $\mathbb{P}_{K_N}$ to be the probability distribution corresponding the Hamiltonian $K_N$, then we see that $\frac{\mathcal{D}}{Z_{K_N}}$ has the following expression.
 
  \begin{equation}
  \begin{aligned}
      &\frac{\mathcal{D}}{Z_{K_N}}= \sum_{K_2} \sum_{\substack{\eta_v^1, \phi_v^1\\ v^1 \in V(K_1)}} \sum_{\substack{\eta_v^2, \phi_v^2\\ v^2 \in B_K(K_1)^c}} \mathbb{P}_{K_N}(\eta_v^1,\phi_v^1, \eta_v^2, \phi_v^2)\\
      & \hspace{0.3 cm} \times  \sum_{ \supp(\tilde{\sigma}^1)=K_1} \prod_{p \in K_1} \exp[ \beta(\rho(\td(\tilde{\sigma}^1)_p) - \rho(1))] \\
      & \hspace{0.3 cm}\times \prod_{e \in E(K_1)} \exp[\kappa(f(\eta^1_v \tilde{\sigma}^1_e (\eta^1_w)^{-1}, \phi^1_v (\phi^1_w)^{-1})  - f( \eta^1_v (\eta^1_w)^{-1}, \phi^1_v (\phi^1_w)^{-1}))]\\& \hspace{0.3 cm}\times\sum_{ \supp(\tilde{\sigma}^2)=K_2} \prod_{p \in K_2} \exp[ \beta(\rho(\td(\tilde{\sigma}^2)_p) - \rho(1))] \\
      & \hspace{0.3 cm} \times \prod_{e \in B_K(K_1)^c} \exp[\kappa(f(\eta^2_v \tilde{\sigma}^2_e (\eta^2_w)^{-1}, \phi^2_v (\phi^2_w)^{-1})  - f( \eta^2_v (\eta^2_w)^{-1}, \phi^2_v (\phi^2_w)^{-1}))],
\end{aligned}
  \end{equation}
  
  Similarly, we have this following expression for the numerator,
  \begin{equation}
  \begin{aligned}
     & \frac{\mathcal{N}}{Z_{K_N}}:=\sum_{K_1} \sum_{K_2} \sum_{\substack{\eta_v^1, \phi_v^1\\ v^1 \in V(K_1)}} \sum_{\substack{\eta_v^2, \phi_v^2\\ v^2 \in B_K(K_1)^c}} \mathbb{P}_{K_N}(\eta_v^1,\phi_v^1, \eta_v^2, \phi_v^2)\\
      & \hspace{0.3 cm} \times  \sum_{ \supp(\tilde{\sigma}^1)=K_1} \langle \gamma, \tilde{\sigma}^1 \rangle \prod_{p \in K_1} \exp[ \beta(\rho(\td(\tilde{\sigma}^1)_p) - \rho(1))]\\
      & \hspace{0.3 cm} \times \prod_{e \in E(K_1)} \exp[\kappa(f(\eta^1_v \tilde{\sigma}^1_e (\eta^1_w)^{-1}, \phi^1_v (\phi^1_w)^{-1})  - f( \eta^1_v (\eta^1_w)^{-1}, \phi^1_v (\phi^1_w)^{-1}))]\\& \hspace{0.3 cm} \times \sum_{ \supp(\tilde{\sigma}^2)=K_2} \langle \gamma, \tilde{\sigma}^2\rangle \prod_{p \in K_2} \exp[ \beta(\rho(\td(\tilde{\sigma}^2)_p) - \rho(1))] \\
      &\hspace{0.3 cm} \times \prod_{e \in B_K(K_1)^c} \exp[\kappa(f(\eta^2_v \tilde{\sigma}^2_e (\eta^2_w)^{-1}, \phi^2_v (\phi^2_w)^{-1})  - f( \eta^2_v (\eta^2_w)^{-1}, \phi^2_v (\phi^2_w)^{-1}))].
  \end{aligned}
  \end{equation}
 
 By introducing the probability terms $\mathbb{P}_{K_N}$, we can use decorrelation estimates from Lemma \ref{lem:decoresti1} for small $\kappa$.
  
  \begin{equation}
  |\mathbb{P}_{K_N}(\eta_v^1, \phi_v^1, \eta_v^2,\phi_v^2) - \mathbb{P}_{K_N}(\eta_v^1, \phi_v^1) \mathbb{P}_{K_N}(\eta_v^2,\phi_v^2)| \le |\gamma| K^{3} e^{-c K} \mathbb{P}_{K_N}(\eta_v^2, \phi_v^2),
  \end{equation}
  for some constant $c$.
  
 Let us now define the quantities
 \begin{equation}
 \begin{aligned}
  &\tilde{\mathcal{N}}:=\\
  & \sum_{\substack{\eta^1_v,\phi^1_v\\v^1 \in V(K_1)}} \mathbb{P}_{K_N}(\eta_v^1,\phi_v^1) \sum_{ \supp(\tilde{\sigma}^1)=K_1} \langle \gamma, \tilde{\sigma}^1 \rangle \prod_{p \in K_1} \exp[ \beta(\rho(\td(\tilde{\sigma}^1)_p) - \rho(1))]\\
  & \times \prod_{e \in E(K_1)} \exp[\kappa(f(\eta^1_v \tilde{\sigma}^1_e (\eta^1_w)^{-1}, \phi^1_v (\phi^1_w)^{-1})  - f( \eta^1_v (\eta^1_w)^{-1}, \phi^1_v (\phi^1_w)^{-1}))]\\
  &\sum_{K_2} \sum_{\substack{\eta_v^2,\phi_v^2}} \mathbb{P}_{K_N}(\eta_v^2,\phi_v^2) \sum_{\supp(\tilde{\sigma}^2) = K_2} \langle \gamma, \tilde{\sigma}^2 \rangle \prod_{p \in K_2} \exp[\beta(\rho(\td(\tilde{\sigma}^2)_p) - \rho(1))]\\
  & \times \prod_{e \in B_K(K_1)^c}  \exp[\kappa(f(\eta^2_v \tilde{\sigma}^2_e (\eta^2_w)^{-1}, \phi^2_v (\phi^2_w)^{-1})  - f( \eta^2_v (\eta^2_w)^{-1}, \phi^2_v (\phi^2_w)^{-1}))],
 \end{aligned}
 \end{equation}
 
 and 
 \begin{equation}
     \begin{aligned}
       &\tilde{\mathcal{D}}:=\\
  & \sum_{\substack{\eta^1_v,\phi^1_v\\v^1 \in V(K_1)}} \mathbb{P}_{K_N}(\eta_v^1,\phi_v^1) \sum_{ \supp(\tilde{\sigma}^1)=K_1}  \prod_{p \in K_1} \exp[ \beta(\rho(\td(\tilde{\sigma}^1)_p) - \rho(1))]\\
  & \times \prod_{e \in E(K_1)} \exp[\kappa(f(\eta^1_v \tilde{\sigma}^1_e (\eta^1_w)^{-1}, \phi^1_v (\phi^1_w)^{-1})  - f( \eta^1_v (\eta^1_w)^{-1}, \phi^1_v (\phi^1_w)^{-1}))]\\
  &\sum_{K_2} \sum_{\substack{\eta_v^2,\phi_v^2}} \mathbb{P}_{K_N}(\eta_v^2,\phi_v^2) \sum_{\supp(\tilde{\sigma}^2) = K_2} \prod_{p \in K_2} \exp[\beta(\rho(\td(\tilde{\sigma}^2)_p) - \rho(1))]\\
  & \times \prod_{e \in B_K(K_1)^c}  \exp[\kappa(f(\eta^2_v \tilde{\sigma}^2_e (\eta^2_w)^{-1}, \phi^2_v (\phi^2_w)^{-1})  - f( \eta^2_v(\eta^2_w)^{-1}, \phi^2_v (\phi^2_w)^{-1}))].
     \end{aligned}
 \end{equation}

 We see that we can write the difference between $\frac{\mathcal{N}}{\mathcal{D}}$ and $\frac{\tilde{\mathcal{N}}}{\tilde{\mathcal{D}}}$.
 We have,
 \begin{equation}
     \begin{aligned}
     \left | \frac{\mathcal{N}}{\mathcal{D}} - \frac{\tilde{\mathcal{N}}}{\tilde{\mathcal{D}}} \right | = \left |\frac{\mathcal{N}/Z_{K_N}}{\mathcal{D}/Z_{K_N}} - \frac{\tilde{\mathcal{N}}}{\tilde{\mathcal{D}}}\right| \le & \left|\frac{(\mathcal{N}/Z_{K_N} - \mathcal{\tilde{N}}) \mathcal{D}/Z_{K_N}}{\tilde{\mathcal{D}} \mathcal{D}/Z_{K_N} } \right|\\ 
     & + \left| \frac{\mathcal{N}/Z_{K_N} (\mathcal{D}/Z_{K_N} - \tilde{\mathcal{D}})}{ \tilde{\mathcal{D}} \mathcal{D}/Z_{K_N}} \right|.
     \end{aligned}
 \end{equation}
 
 Since we have the absolute value bound that $|\langle \gamma, \tilde{\sigma}^i \rangle| \le 1$ and all other remaining quantities in the expressions of $\mathcal{N}$ and $\tilde{\mathcal{N}}$ are positive, we see that we have the bound $\frac{\mathcal{N}/Z_{K_N}}{\mathcal{D}/Z_{K_N}} \le 1$.
 
 If we define,
 \begin{equation}
 \begin{aligned}
     &\hat{\mathcal{D}}:=\\
  & \sum_{\substack{\eta^1_v,\phi^1_v\\v^1 \in V(K_1)}} \sum_{ \supp(\tilde{\sigma}^1)=K_1}  \prod_{p \in K_1} \exp[ \beta(\rho(\td(\tilde{\sigma}^1)_p) - \rho(1))]\\
  & \times \prod_{e \in E(K_1)} \exp[\kappa(f(\eta^1_v\tilde{\sigma}^1_e (\eta^1_w)^{-1}, \phi^1_v (\phi^1_w)^{-1})  - f( \eta^1_v (\eta^1_w)^{-1}, \phi^1_v (\phi^1_w)^{-1}))]\\
  &\sum_{K_2} \sum_{\substack{\eta_v^2,\phi_v^2}} \mathbb{P}_{K_N}(\eta_v^2,\phi_v^2) \sum_{\supp(\tilde{\sigma}^2) = K_2} \prod_{p \in K_2} \exp[\beta(\rho(\td(\tilde{\sigma}^2)_p) - \rho(1))]\\
  & \times \prod_{e \in B_K(K_1)^c}  \exp[\kappa(f(\eta^2_v \tilde{\sigma}^2_e (\eta^2_w)^{-1}, \phi^2_v (\phi^2_w)^{-1})  - f( \eta^2_v (\eta^2_w)^{-1}, \phi^2_v (\phi^2_w)^{-1}))],
  \end{aligned}
 \end{equation}
  then we see that we can bound $\frac{\mathcal{N}/Z_N - \tilde{\mathcal{N}}}{ \tilde{\mathcal{D}}}$ by $|\gamma|K^3 e^{-cK} \frac{\hat{\mathcal{D}}}{\mathcal{D}}$.

 When computing the ratio between $\hat{\mathcal{D}}$ and $\tilde{\mathcal{D}}$, we can cancel out the ratio of the terms involving $K_2$. We see that we only need to consider the ratio
  
  \begin{equation}
      \frac{\hat{\mathcal{D}}}{\tilde{\mathcal{D}}}= \frac{\substack{ \sum_{\substack{\eta^1_v,\phi^1_v\\v^1 \in V(K_1)}} \sum_{ \supp(\tilde{\sigma}^1)=K_1}  \prod_{p \in K_1} \exp[ \beta(\rho(\td(\tilde{\sigma}^1)_p) - \rho(1))]\\
   \times \prod_{e \in E(K_1)} \exp[\kappa(f(\eta^1_v \tilde{\sigma}^1_e (\eta^1_w)^{-1}, \phi^1_v (\phi^1_w)^{-1})  - f( \eta^1_v (\eta^1_w)^{-1}, \phi^1_v (\phi^1_w)^{-1}))] }}{\substack{ \sum_{\substack{\eta^1_v,\phi^1_v\\v^1 \in V(K_1)}} \mathbb{P}_{K_N}(\eta_v^1,\phi_v^1) \sum_{ \supp(\tilde{\sigma}^1)=K_1}  \prod_{p \in K_1} \exp[ \beta(\rho(\td(\tilde{\sigma}^1)_p) - \rho(1))]\\ \times \prod_{e \in E(K_1)} \exp[\kappa(f(\eta^1_v \tilde{\sigma}^1_e (\eta^1_w)^{-1}, \phi^1_v (\phi^1_w)^{-1})  - f( \eta^1_v (\eta^1_w)^{-1}, \phi^1_v (\phi^1_w)^{-1}))]}}.
  \end{equation}
  
 We can bound $\exp[\kappa(f(\eta^1_v \tilde{\sigma}^1_e (\eta^1_w)^{-1}, \phi^1_v (\phi^1_w)^{-1})  - f( \eta^1_v (\eta^1_w)^{-1}, \phi^1_v (\phi^1_w)^{-1}))]$ from above by $\exp [2\kappa (\max_{a,b}\text{Re}f(a,b) - \min_{c,d} \text{Re} f(c,d))]$ in the numerator, and lower bound it by $\exp[-2 \kappa (\max_{a,b}\text{Re}f(a,b) - \min_{c,d} \text{Re} f(c,d))]$ in the denominator.
 
 This allows us to bound the ratio by
 
 \begin{equation}
 \begin{aligned}
     \frac{\hat{\mathcal{D}}}{\tilde{\mathcal{D}}}& \le \prod_{e \in E(K_1)}\exp[4\kappa (\max_{a,b} \text{Re}f(a,b) - \min_{c,d} \text{Re}f(c,d))]\\
     &\frac{\sum_{\substack{\eta_v^1,\phi_v^1\\v^1 \in V(K_1)} } \sum_{\supp(\tilde{\sigma}^1)=K_1}\prod_{p \in K_1} \exp[\beta(\rho((\td \tilde{\sigma}^1)_p) - \rho(1))]}{\sum_{\substack{\eta_v^1,\phi_v^1\\v^1 \in V(K_1)} } \mathbb{P}_{K_N}(\eta^1_v,\phi^1_v){\sum_{\supp(\tilde{\sigma}^1)=K_1}\prod_{p \in K_1} \exp[\beta(\rho((\td \tilde{\sigma}^1)_p) - \rho(1))]}}
 \end{aligned}
 \end{equation}
 In the fraction above, in both the numerators and denominators,
 the sum over $\eta_v^1$ and $\phi_v^1$ splits from the sum over $\tilde{\sigma}^1$. We can cancel out the sum over $\tilde{\sigma}^1$. We can cancel out this sum over $\tilde{\sigma}^1$ in both the numerator and denominator. The remaining sum over $\eta_v^1,\phi_v^1$ is equal to $1$ in the denominator and $|G|^{|V(K_1)}|H |^{|V(K_1)|}$ in the numerator. Considering that these are all minimal vortices that are well separated, we can explicitly compute $|V(K_1)|$ and $|E(K_1)|$.

  
  Recall the definition of  $L$ from equation \ref{eq:defL}.
  Our ultimate bound on $\frac{\hat{\mathcal{D}}}{\tilde{\mathcal{D}}}$ is $L^m$.

  
  
  Therefore, we can bound $|\frac{\mathcal{N}}{\mathcal{D}} - \frac{\tilde{\mathcal{N}}}{\tilde{\mathcal{D}}}| \le |\gamma|K^{3} e^{-cK} L^m$.
  
  \textit{Part 2: Decorrelating the Minimal Vortices}
  The importance of the expression $\frac{\tilde{\mathcal{N}}}{\tilde{\mathcal{D}}}$ is that one is able to completely cancel out the effect of excited plaquettes in $K_2$.
  
  Recall the spanning tree $\mathcal{T}(V_1,V_2,\ldots,V_m)$ and the splitting of the excitation $\sigma$ into $\tilde{\sigma}^1,\ldots,\tilde{\sigma}^m$ between them. The spanning tree can be chosen so that the only excited edge with $\tilde{\sigma}^i_e \ne 1$ on $V_i$ is the edge that $V_i$ shares with $\gamma$.
  
  Consider the quantity,
  \begin{equation}
  \begin{aligned}
      &\mathfrak{D}:= \sum_{\eta_v^1,\phi_v^1,\ldots, \eta_v^m,\phi_v^m} \mathbb{P}_{K_N}(\eta_v^1,\phi_v^1,\ldots,\eta_v^m, \phi_v^m)\\
      & \times \sum_{\substack{ \supp({\tilde{\sigma}^1})= V_1}} \exp[-12\beta \text{Re}(\rho(1) - \rho((\td \tilde{\sigma}^1))] \\
      &\prod_{e=(v,w) \in E(V_1)} \exp[\kappa(f(\eta_v^1 \tilde{\sigma}^1_e (\eta_w^1)^{-1},\phi_v^{1}(\phi_w^1)^{-1} ) - f(\eta_v^1 (\eta_w^1)^{-1},\phi_v^{1} (\phi_w^1)^{-1}))]\\
      & \times \ldots\\
      & \times  \sum_{\substack{ \supp({\tilde{\sigma}^m})= V_m}} \exp[-12\beta \text{Re}(\rho(1) - \rho((\td \tilde{\sigma}^m))]\\
      & \prod_{e=(v,w) \in E(V_m)} \exp[\kappa(f(\eta_v^m \tilde{\sigma}^m_e (\eta_w^m)^{-1},\phi_v^{m}(\phi_w^m)^{-1} ) - f( \eta_v^m (\eta_w^m)^{-1}, \phi_v^m (\phi_w^m)^{-1}) )]
  \end{aligned}
  \end{equation}
  and
  \begin{equation}
       \begin{aligned}
      &\mathfrak{N}:= \sum_{\eta_v^1,\phi_v^1,\ldots, \eta_v^m,\phi_v^m} \mathbb{P}_{K_N}(\eta_v^1,\phi_v^1,\ldots,\eta_v^m, \phi_v^m)\\
      & \times \sum_{\substack{ \supp({\tilde{\sigma}^1})= V_1}} \langle \delta \gamma, \td \tilde{\sigma}^1 \rangle \exp[-12\beta\text{Re}(\rho(1) - \rho((\td \tilde{\sigma}^1))]\\
      & \prod_{e=(v,w) \in E(V_1)} \exp[\kappa(f(\eta_v^1 \tilde{\sigma}^1_e (\eta_w^1)^{-1},\phi_v^{1}(\phi_w^1)^{-1} ) - f(\eta_v^1(\eta_w^1)^{-1},\phi_v^1 (\phi_w^1)^{-1}))]\\
      & \times \ldots\\
      & \times  \sum_{\substack{ \supp({\tilde{\sigma}^m})= V_1}} \langle \delta \gamma, \td \tilde{\sigma}^m \rangle \exp[-12\beta \text{Re}(\rho(1) - \rho((\td \tilde{\sigma}^m))] \\
      &\prod_{e=(v,w) \in E(V_m)} \exp[\kappa(f(\eta_v^m \tilde{\sigma}^m_e (\eta_w^m)^{-1},\phi_v^{m}(\phi_w^m)^{-1} ) - f(\eta_v^m (\eta_w^m)^{-1}, \phi_v^m(\phi_w^m)^{-1}))].
  \end{aligned}
  \end{equation}
  
  We see that $\frac{\tilde{\mathcal{N}}}{\tilde{\mathcal{D}}} = \frac{\mathfrak{N}}{\mathfrak{D}}$. In this equality, we implicitly used the fact that the $V_i$'s are minimal vortices so $(\td \tilde{\sigma}^i)$ is constant on its support of $6$ unoriented plaquettes. Since our vortices are separated by distance $K$, we have the following inequality from Corollary \ref{col:DecayEst2},
  \begin{equation}\label{eq:SplitMinimal}
  \begin{aligned}
     & |\mathbb{P}_{K_N}(\eta_v^1,\phi_v^1,\ldots, \eta_v^m,\phi_v^m) - \prod_{i=1}^m \mathbb{P}_{K_N}(\eta_v^i,\phi_v^i)|\\ & \le m \frac{e^{-c'K}}{\min_j \mathbb{P}_{K_N}(\eta_v^j,\phi_v^j)}\prod_{i=1}^m [\mathbb{P}_{K_N}(\eta_v^i, \phi_v^i) + e^{-cK}].
  \end{aligned}
  \end{equation}
  
  Provided $\gamma$ is sufficiently far away from the boundary, one can quite readily show that there is constant $B$ such that $\mathbb{P}_{K_N}(\eta_v^j,\phi_v^j) \ge B$ for any configuration of $\eta_v^j,\phi_v^j$ associated to a minimal vortex.
  \begin{rmk}
    We will give here a rough sketch of this fact. This is essentially due to the fact that we are considering a finite range interaction on a fixed finite subsystem with no assignment explicitly prohbited by the Hamiltonian. Consider two configurations  $(\tilde{\eta}_v^1,\tilde{\phi}_v^1)$ and $(\hat{\eta}_v^1,\hat{\phi}_v^1)$ on some minimal vortex $V_1$.  Consider a gauge field assignment on the lattice $\mathcal{I}=\{\eta_v,\phi_v\}$ such that restricted to the vertices on $K$ we have that $\eta_v= \tilde{\eta}_v^1$. 
    Now consider the configuration $\hat{I}$ which is equal to $\mathcal{I}$ on $\{\eta_v,\phi_v\}$ outside of $V_1$ and equal to $\{\hat{\eta}_v^1,\hat{\phi}_v^1\}$ on $V_1$. There is some finite bound $B$ on the shift of energy between $\mathcal{I}$ and $\hat{I}$. Thus, the configuration $\{\hat{\eta}_v^1,\hat{\phi}_v^1\}$ on $K_1$ is no more than $B$ times more likely than the configuration $\{ \tilde{\eta}_v^1,\tilde{\phi}_v^1 \}$. We can apply this logic for all the finitely many configurations which shows that $\mathbb{P}_{K_N}(\cdot)$ is bounded below for all configurations on the vertices of $K_1$.
  \end{rmk}
  
  We now define the quantities $\tilde{\mathfrak{N}}$ and $\tilde{\mathfrak{D}}$ as,
  \begin{equation}
  \begin{aligned}
      &\tilde{\mathfrak{N}}:= \sum_{\substack{\eta_v^1,\phi_v^1\\ v \in V(V_1)}}  \mathbb{P}_{K_N}(\eta_v^1,\phi_v^1) \sum_{\substack{ \supp({\tilde{\sigma}^1})= V_1}} \langle \delta \gamma, \td \tilde{\sigma}_1 \rangle \exp[-12\beta \text{Re}(\rho(1) - \rho((\td \tilde{\sigma}^1))] \\
      & \prod_{e=(v,w) \in E(V_1)} \exp[\kappa(f(\eta_v^1 \tilde{\sigma}^1_e (\eta_w^1)^{-1},\phi_v^{1}(\phi_w^1)^{-1} ) - f(\eta_v^1 (\eta_w^1)^{-1},\phi_v^{1}(\phi_w^1)^{-1} ))]\\
      & \times \ldots\\
      & \times \sum_{\substack{\eta_v^m,\phi_v^m \\ v \in V(V_m)}} \mathbb{P}_{K_N}(\eta_v^m,\phi_v^m) \sum_{\substack{ \supp({\tilde{\sigma}^m})= V_m}} \langle \delta \gamma, \td \tilde{\sigma}^m \rangle \exp[-12\beta\text{Re}(\rho(1) - \rho((\td \tilde{\sigma}^m))]\\
      & \prod_{e=(v,w) \in E(V_m)} \exp[\kappa(f(\eta_v^m \tilde{\sigma}^m_e (\eta_w^m)^{-1},\phi_v^{m}(\phi_w^m)^{-1} ) -f(\eta_v^m  (\eta_w^m)^{-1},\phi_v^{m}(\phi_w^m )^{-1}))],
  \end{aligned}
  \end{equation}
  and
  \begin{equation}
  \begin{aligned}
      &\tilde{\mathfrak{D}}:= \sum_{\substack{\eta_v^1,\phi_v^1\\ v \in V(V_1)}} \mathbb{P}_{K_N}(\eta_v^1,\phi_v^1) \sum_{\substack{ \supp({\tilde{\sigma}^1})= V_1}} \exp[-12\beta\text{Re}(\rho(1) - \rho((\td \tilde{\sigma}^1))] \\
      & \prod_{e=(v,w) \in E(V_1)} \exp[\kappa(f(\eta_v^1 \tilde{\sigma}^1_e (\eta_w^1)^{-1},\phi_v^{1}(\phi_w^1)^{-1} ) - f(\eta_v^1 (\eta_w^1)^{-1},\phi_v^{1}(\phi_w^1 )^{-1}))]\\
      & \times \ldots\\
      & \times \sum_{\substack{\eta_v^m,\phi_v^m \\ v \in V(V_m)}} \mathbb{P}_{K_N}(\eta_v^m,\phi_v^m) \sum_{\substack{ \supp({\tilde{\sigma}^m})= V_m}} \exp[-12\beta\text{Re}(\rho(1) - \rho((\td \tilde{\sigma}^m))]\\
      & \prod_{e=(v,w) \in E(V_m)} \exp[\kappa(f(\eta_v^m\tilde{\sigma}^m_e (\eta_w^m)^{-1},\phi_v^{m}(\phi_w^m)^{-1} )- f(\eta_v^m  (\eta_w^m)^{-1},\phi_v^{m}(\phi_w^m)^{-1} ))],
  \end{aligned}
  \end{equation}

  As before, we see that,
  \begin{equation}
      \left| \frac{\mathfrak{N}}{\mathfrak{D}} - \frac{\tilde{\mathfrak{N}}}{\tilde{\mathfrak{D}}} \right| \le \left| \frac{\mathfrak{N}(\mathfrak{D} - \tilde{\mathfrak{D}})}{\mathfrak{D} \tilde{\mathfrak{D}}} \right| + \left| \frac{\mathfrak{D}(\mathfrak{N} - \tilde{\mathfrak{N}})}{ \mathfrak{D} \tilde{\mathfrak{D}}} \right|.
  \end{equation}
  As before, we can bound $\left|\frac{\mathfrak{N}}{\mathfrak{D}} \right| \le 1$, and we can relate both $|\mathfrak{D} - \tilde{\mathfrak{D}}|$ and $|\mathfrak{N} - \tilde{\mathfrak{N}}|$ to a third quantity, which we bound. We define,
  \begin{equation}
  \begin{aligned}
      &\hat{\mathfrak{D}}:= \sum_{\substack{\eta_v^1,\phi_v^1\\ v \in V(V_1)}} [\mathbb{P}_{K_N}(\eta_v^1,\phi_v^1) + e^{-c'K}] \sum_{\substack{ \supp({\tilde{\sigma}^1})= V_1}} \exp[-12\beta\text{Re}(\rho(1) - \rho((\td \tilde{\sigma}^1))] \\
      & \prod_{e=(v,w) \in E(V_1)} \exp[\kappa(f(\eta_v^1 \tilde{\sigma}^1_e (\eta_w^1)^{-1},\phi_v^{1}(\phi_w^1)^{-1} ) - f(\eta_v^1 (\eta_w^1)^{-1},\phi_v^{1}(\phi_w^1)^{-1} ))]\\
      & \times \ldots\\
      & \times \sum_{\substack{\eta_v^m,\phi_v^m \\ v in V(V_m)}} [\mathbb{P}_{K_N}(\eta_v^m,\phi_v^m) + e^{-c'K}] \sum_{\substack{ \supp({\tilde{\sigma}^m})= V_m}} \exp[-12\beta \text{Re}(\rho(1) - \rho((\td \tilde{\sigma}^m))]\\
      & \prod_{e=(v,w) \in E(V_m)} \exp[\kappa(f(\eta_v^m \tilde{\sigma}^m_e (\eta_w^m)^{-1},\phi_v^{m}(\phi_w^m)^{-1} ) -f(\eta_v^m  (\eta_w^m)^{-1},\phi_v^{m}(\phi_w^m)^{-1} ))].
  \end{aligned}
  \end{equation}
  That $\frac{m e^{-c'K}}{\min_{\eta_v^1,\phi_v^1} \mathbb{P}_{K_N}(\eta_v^1,\phi_v^1)} \hat{\mathfrak{D}}$ is a bound on $|\mathfrak{D}- \tilde{\mathfrak{D}}|$
and $|\mathfrak{N}- \tilde{\mathfrak{N}}|$ is a consequence of the inequality in  \eqref{eq:SplitMinimal}, where the minimum in the denominator is taken over all configuration of $\eta$ and $\phi$ on the vertices of any minimal vortex.

Computing the ratio of $\frac{\hat{\mathfrak{D}}}{\tilde{\mathfrak{D}}}$ is very similar to computing the ratio of $\frac{\hat{\mathcal{D}}}{\tilde{\mathcal{D}}}$. We can always bound

  $$\prod_{e=(v,w) \in E(V_m)} \exp[\kappa(f(\eta_v^m \tilde{\sigma}^m_e (\eta_w^m)^{-1},\phi_v^{m}(\phi_w^m)^{-1} ) -f(\eta_v^m  (\eta_w^m)^{-1},\phi_v^{m}(\phi_w^m)^{-1} ))]$$ by $\exp[48 \kappa (\max_{a,b}\text{Re}f(a,b) - \min_{c,d} \text{Re} f(c,d))]$ from above in the numerator and $\exp[-48 \kappa (\max_{a,b}\text{Re}f(a,b) - \min_{c,d} \text{Re} f(c,d))]$ from below in the denominator. Once this is done, the sum over $\tilde{\sigma}^i$ and $\eta^i,\phi^i$ split and we can divide the numerator and denominator by this sum.

  We bound each $\mathbb{P}_{K_N}(\eta_v^1,\phi_v^1) + e^{-c'K}$ by $1+ e^{-c'K}$. Thus, in the numerator, the sum over $\eta^1$ and $\phi^1$ is bounded by $(1+ e ^{-c'K})(|G||H|)^{24}$ since there are $|G|$ choices for every $\eta$ and $|H |$ choices for every $\phi$.  The sum over $\phi^1$ and $\eta^1$ in the denominator is exactly 1.
  
  Recalling $L$ from \eqref{eq:defL}, we see that our bound we can bound $\hat{\mathfrak{D}}$ as,
  \begin{equation}
  \begin{aligned}
      \frac{\mathfrak{\hat{D}}}{\tilde{\mathfrak{D}}} &\le  [(1+ e^{-c'K}) L]^m,
 \end{aligned}
 \end{equation}
 and we get the bound
\begin{equation}
    \left| \frac{\mathfrak{N}}{\mathfrak{D}}- \frac{\tilde{\mathfrak{N}}}{\tilde{\mathfrak{D}}} \right| \le |\gamma|K^{3} e^{-cK} L^m + \frac{m e^{-c'K}}{\min_{\substack{\eta_v,\phi_v\\ v \in \text{ minimal vortex }}} \mathbb{P}_{K_N}(\eta_v,\phi_v)}  [(1+ e^{-c'K}) L]^m
\end{equation}

  The ratio $\frac{\tilde{\mathfrak{N}}}{\tilde{\mathfrak{D}}}$ is the desired expression $\mathcal{D}_{\beta,\kappa}^m$. Combining all error terms, this completes the proof of the Lemma.
  

\end{proof}

  \textit{Return to Proof of Theorem \ref{thm:WilsonLoopReplowdisorder}}.
  
  With the previous Lemma \ref{lem:WeirdCondition} in hand, we can compute the error bounds by removing the conditioning.
  
  We see that,
  \begin{equation} \label{eq:almostfinalequation}
  \begin{aligned}
     & |\mathbb{E}[W_{\gamma} -\mathcal{D}^{M_{\gamma}}_{\beta,\kappa}| \tilde{E}]| \le \mathbb{P}(G^2|\tilde{E})[\mathbb{E}[W_{\gamma} -\mathcal{D}^{M_{\gamma}}_{\beta,\kappa}| G^2 ] \\
     & + \sum_{m} \sum_{\substack{V_1,\ldots,V_m\\K-Separated\\Minimal}} \mathbb{P}(G^1(V_1,\ldots,V_m)|\tilde{E})|[\mathbb{E}[W_{\gamma} -\mathcal{D}^{M_{\gamma}}_{\beta,\kappa}| G^1(V_1,\ldots,V_m) ] |.
 \end{aligned}
  \end{equation}
  
  One the event $G^2$, we know that $W_{\gamma}=1= \mathcal{D}_{\beta,\kappa}^{M_{\gamma}}$. We can also bound $\mathbb{P}(G^1(V_1,\ldots,V_m)|\tilde{E}) \le P(\tilde{E})^{-1} \mathcal{B}^{6m}$. Notice that $\mathbb{P}(\tilde{E}) \mathbb{P}(G^1(V_1,\ldots,V_m)|\tilde{E}) $ is bounded by the probability that our configuration has a support containing the $K$-separated minimal vortices $V_1,\ldots V_m$. The probability that we see a configuration whose support contains the minimal vortices $V_1 \ldots V_m$ is less than $\prod_{i=1}^m\Phi_{UB}(V_i)\le \mathcal{B}^{6m}$ by Lemma \ref{lem:LowDisorderDecomp}. Now, there are $\frac{|\gamma|^m}{m!}$ ways to choose $V_1,\ldots,V_m$ such that they are centered on edges of $\gamma$.
  
  Combining these estimates, we see that we can bound the last line of \eqref{eq:almostfinalequation} by performing the summation over the last term.
  \begin{equation}
  \begin{aligned}
 & \sum_{m} \sum_{\substack{V_1,\ldots,V_m\\K-Separated\\Minimal}} \mathbb{P}(G^1(V_1,\ldots,V_m) |\tilde{E})|[\mathbb{E}[W_{\gamma} -\mathcal{D}^{M_{\gamma}}_{\beta,\kappa}| G^1(V_1,\ldots,V_m) ] |\\
     & \le \mathbb{P}(\tilde{E})^{-1} \sum_{m=1}^{\infty}\mathcal{B}^{6m} \frac{|\gamma|^m}{m!}|\gamma|K^{3} e^{-cK} L^m \\
     &+ \mathbb{P}(\tilde{E})^{-1} \sum_{m=1}^{\infty} \mathcal{B}^{6m}\frac{|\gamma|^m}{m!} \frac{m e^{-c'K}}{\min_{\substack{\eta_v,\phi_v\\ v \in \text{ minimal vortex }}} \mathbb{P}_{K_N}(\eta_v,\phi_v)}  [(1+ e^{-c'K}) L]^m\\
     &\le \mathbb{P}(\tilde{E})^{-1} |\gamma|^2 \mathcal{B}^{6} L K^{3} e^{-cK} \exp[\mathcal{B}^{6}|\gamma|L]\\
     &+ \mathbb{P}(\tilde{E})^{-1} \frac{|\gamma| \mathcal{B}^{6} [(1+e^{-c'K})L] e^{-c'K}}{\min_{\substack{\eta_v,\phi_v\\ v \in \text{ minimal vortex }}} \mathbb{P}_{K_N}(\eta_v,\phi_v)} \exp[|\gamma| \mathcal{B}^{6}[(1+e^{-c'}K) L]].
  \end{aligned}
  \end{equation}
  
  Notice that we only accrue error from $m=1$ onwards, not at $m=0$. 
  
  \begin{rmk}
  As a remark, a slight modification of the proof shows that we can bound $\mathbb{P}(G^1(V_1,\ldots,V_m)|\tilde{E})$ by $\mathcal{B}^{6m}$ and we do not actually need to include the prefactor $\mathbb{P}(\tilde{E})^{-1}$. This is actually established later in Corollary \ref{col:condbound2}. To simplify notation later, we will actually apply the above estimate without the prefactor $\mathbb{P}(\tilde{E})^{-1}$.
  \end{rmk}
  
  Similarly, we can also bound $|\mathbb{E}[W_{\gamma} -\mathcal{D}^{M_{\gamma}}_{\beta,\kappa}]|$ without the restriction to $\tilde{E}$.
  We see that,
  \begin{equation} \label{eq:almostfinalequation2}
  \begin{aligned}
     & |\mathbb{E}[W_{\gamma} -\mathcal{D}^{M_{\gamma}}_{\beta,\kappa}]| = |\mathbb{P}(\tilde{E}^c)[\mathbb{E}[W_{\gamma}- \mathcal{D}^{M_{\gamma}}_{\beta,\kappa}|\tilde{E}^c] + \mathbb{P}(G^2)[\mathbb{E}[W_{\gamma} -\mathcal{D}^{M_{\gamma}}_{\beta,\kappa}| G^2 ] \\
     & + \sum_{m} \sum_{\substack{V_1,\ldots,V_m\\K-Separated\\Minimal}} \mathbb{P}(G^1(V_1,\ldots,V_m))[\mathbb{E}[W_{\gamma} -\mathcal{D}^{M_{\gamma}}_{\beta,\kappa}| G^1(V_1,\ldots,V_m) ] |\\
     &\le\mathbb{P}(\tilde{E}^c)|[\mathbb{E}[W_{\gamma}- \mathcal{D}^{M_{\gamma}}_{\beta,\kappa}|\tilde{E}^c]| + \mathbb{P}(G^2)|[\mathbb{E}[W_{\gamma} -\mathcal{D}^{M_{\gamma}}_{\beta,\kappa}| G^2 ]|\\
     & + \sum_{m} \sum_{\substack{V_1,\ldots,V_m\\K-Separated\\Minimal}} \mathbb{P}(G^1(V_1,\ldots,V_m))|[\mathbb{E}[W_{\gamma} -\mathcal{D}^{M_{\gamma}}_{\beta,\kappa}| G^1(V_1,\ldots,V_m) ] |.
 \end{aligned}
  \end{equation}
  
  The only difference we need to mention from our previous analysis is that we can bound the first term in the last inequality by $\mathbb{P}(\tilde{E})^c$.
  

  
\end{proof}

\subsection{Approximation by Poisson Random Variables}

We would again expect that with good decorrelation estimates, that the excitation of two minimal vortices would roughly be independent of each other. Again, our goal is to use a version of Theorem \ref{thn:PoisApproxGen} in order to show that $M_{\gamma}$ is roughly distributed according to a Poisson random variable. We will do this when conditioning on the high probability event $\tilde{E}$.

Again, we write,
\begin{equation}
    M_{\gamma}:= \sum_{e \in \gamma} \mathbbm{1}[F_{P(e)}],
\end{equation}
where $F_{P(e)}$ is the event that there is a minimal vortex centered at an edge of $\gamma$. Let $B_K(e)$ be the set of edges $e'$ such that the minimal vortex $P(e')$ would be found in a box $B_K(e)$ of size $2K$ centered around the edge $e$.

We use $\mathbb{P}_{\tilde{E}}$ to denote the probability distribution when restricted to the event $\tilde{E}$. 
In the discussion that follows in this subsection, we will restrict our analysis to events that belong in $\tilde{E}$.

The version of Theorem \ref{thn:PoisApproxGen} we would use in this case is as follows,
\begin{thm} \label{thn:PoisApproxGen2}
    Consider the following constants.
    \begin{equation}
    \begin{aligned}
       & b_1:= \sum_{e \in \gamma} \sum_{e' \in B_K(e)} \mathbb{P}_{\tilde{E}}(F_{P(e)}) \mathbb{P}_{\tilde{E}}(F_{P(e')}),\\
    &    b_2:= \sum_{e \in \gamma} \sum_{e' \in B_K(e) \setminus e} \mathbb{P}_{\tilde{E}}(\mathbbm{1}(F_{P(e)}) \mathbbm{1}(F_{P(e')})),\\
    & b_3:= \sum_{\gamma} \mathbb{E}_{\tilde{E}}\left[| \mathbb{E}_{\tilde{E}}[\mathbbm{1}(F_{P(e)})| \mathbbm{1}(F_{P(e')}) e' \not \in B_K(e)] -\mathbb{P}_{\tilde{E}}(F_{P(e)})| \right].
    \end{aligned}
    \end{equation}
    Let $\mathcal{L}(M_{\gamma})$ denote the law of $M_{\gamma}$ and let $\lambda= \mathbb{E}_{\tilde{E}}[M_{\gamma}]$. Then 
    \begin{equation}
    d_{
    TV}(\mathcal{L}_{\tilde{E}}(M_{\gamma}), \text{Poisson}(\lambda))\le min(1, \lambda^{-1})(b_1 + b_2) +\min(1, 1.4 \lambda^{-1/2})b_3.
    \end{equation}
    $\mathcal{L}_{\tilde{E}}(M_{\gamma})$ is law of $M_{\gamma}$ conditioned on the event $\tilde{E}$.
\end{thm}

Again, a small adaptation of the proof of Lemma \ref{lem:LowDisorderDecomp} shows the following Corollary.
\begin{col}\label{col:condbound2}[of Lemma \ref{lem:LowDisorderDecomp}]
Let $E_1$ and $E_2$ be two sets of edges. Let $M(E_1,E_2)$ (inside $\tilde{E}$) be the event that there is a minimal vortex centered around each edge of $E_1$ and there is no minimal vortex centered around any edge of $E_2$. Assume that $M(E_1,E_2)$ has positive probability. Let $
V$ be some minimal vortex that is not centered around an edge in $E_1$ or in $E_2$. The probability that we have a configuration $(\sigma,\phi)$ whose support contains $V$ conditioned on the event $M(E_1,E_2)$ is less than $\Phi_{UB}(V)$.

\end{col}
\begin{proof}
The same adaptations that allowed Corollary \ref{col:boundcondition} to follow from Lemma \ref{lm:vortex} allow this Corollary to follow from Lemma  \ref{lem:LowDisorderDecomp}.

We see that if a configuration $(\sigma,\phi)$ in $M(E_1,E_2)$ has support that contains $V$, then the support is of the form $V \bigcup_{e\in E_1} P(e) \cup R$, where $R$ is a union of vortices that is disjoint from $V \bigcup_{e \in E_1} P(e)
$ and does not contain any minimal vortex centered around an edge of $E_2$.

We have already computed the sum of $\exp[H_{N,\beta,\kappa}(\sigma,\phi)]$ for all configurations whose support is $V \bigcup_{e \in E_1} P(e) \cup R$. This is the expression $\mathcal{Z}( V \bigcup_{e\in E_1}P(e) \cup R)$ from equation \eqref{eq:somepartfunc}.

We can follow the steps in the proof of Lemma \ref{lem:LowDisorderDecomp} by choosing a spanning tree $T(V)$ that is a spanning tree of $C(V)$, a cube surrounding $V$, and its complement. If we gauge the configuration $\sigma$ with respect to the spanning tree $T(V)$ and split accordingly, we can ultimately derive that $\mathcal{Z}(V \bigcup_{e \in E_1} P(e) \cup R)$ is less than $\Phi_{UB}(V) \mathcal{Z}(\bigcup_{e \in E_1} P(e) \cup R)$. This exactly uses the fact that $V$ was a minimal vortex and can be separated from the other excitations.
We can now assert that the sum of $\exp[H_{N,\beta,\kappa}(\sigma,\phi)]$ for all configurations in $M(E_1,E_2)$ whose support contains $V$ is bounded as follows,
\begin{equation} \label{eq:somebound2}
\begin{aligned}
    \sum_{\substack{V \in \supp((\sigma,\phi))\\(\sigma,\phi) \in M(E_1,E_2)}}\exp[H_{N,\beta,\kappa}(\sigma,\phi)]& = \sum_{R} \mathcal{Z}(V \bigcup_{e \in E_1} P(e) \cup R) \\
    &\le \Phi_{UB}(V) \sum_{R} \mathcal{Z}(\bigcup_{e \in E_1} P(e) \cup R).
\end{aligned}
\end{equation}
Here, $R$ runs over all all vortices that are found in the vortex expansion of some configuration in $M(E_1,E_2)$ that contains $V$ in its support.
$$
Z(M(E_1,E_2)) \ge \sum_{R} \mathcal{Z}(\bigcup_{e \in E_1} P(e) \cup R).
$$

Taking the ratio of equation \eqref{eq:somebound2} with our partition function $Z(M(E_1,E_2))$ shows that the probability of observing the a configuration $(\sigma,\phi)$ whose support contains $U
$ when conditioned on the event $Z(M(E_1,E_2))$ is less than $\Phi_{UB}(V)$.


\end{proof}

As before, with this lemma in hand, we can start to bound the quantities $b_1,b_2,b_3$
\begin{lem} \label{lem:comparwpoisson2}
Assume that the conditions of Theorem \ref{thm:WilsonLoopReplowdisorder}.
Recall the decorrelation estimates from Lemma \ref{lem:decoresti1}, the constant $L$ from equation \eqref{eq:defL}, and the constant $\mathcal{B}$ from Lemma \ref{lem:rarelowdisorder}.  Define the constant
$$
\tilde{\mathfrak{c}}:=12 K^4 \mathcal{B}^6.
$$
Assume that $\tilde{\mathfrak{c}}<1$.

We have the following bounds on $b_1,b_2,b_3$ and $\lambda= \mathbb{E}[M_{\gamma}]$.
\begin{equation}
    \begin{aligned}
        &b_1 \le |\gamma|12 K^4 \mathcal{B}^{12},\\
        &b_2 =0 \\
        & b_3 \le |\gamma| \left|(1-\tilde{\mathfrak{c}})[\mathcal{D}_{\beta,\kappa}- K^{3} e^{-cK} L] -   [\mathcal{D}_{\beta,\kappa} + K^{3}e^{-cK} L]\right|,\\
        &|\lambda- |\gamma| \mathcal{D}_{\beta,\kappa}| \le |\gamma| \left|(1-\tilde{\mathfrak{c}})[\mathcal{D}_{\beta,\kappa}- K^{3} e^{-cK} L] -   [\mathcal{D}_{\beta,\kappa} + K^{3}e^{-cK} L]\right|.
    \end{aligned}
\end{equation}
As a consequence of these estimates, we have from Theorem \ref{thn:PoisApproxGen2} that,
\begin{equation} \label{eq:fincompar}
\begin{aligned}
    \text{d}_{TV}(\mathcal{L}_{\tilde{E}}(M_{\gamma}), &\text{Poisson}(|\gamma|\mathcal{D}_{\beta,\kappa})) \le 12|\gamma|K^4 \mathcal{B}^{12}\\
    &+ 2|\gamma|\left|(1-\tilde{\mathfrak{c}})[\mathcal{D}_{\beta,\kappa}- K^{3} e^{-cK} L] -   [\mathcal{D}_{\beta,\kappa} + K^{3}e^{-cK} L]\right|
\end{aligned}
\end{equation}

\end{lem}
\begin{proof}
As before, the most difficult part is to bound $b_3$. The proof is largely similar to  the proof of Lemma \ref{lem:comparwpoisson}, except for the fact that we must now deal with some technicalities of decorrelation estimates.

As before, let $E_1$ be some subset of edges in $\gamma \setminus B_K(e)$ and $E_2$ be all remaining edges in $\gamma \setminus B_K(e)$. Recall the notation $M(E_1,E_2)$ and $Z(M(E_1,E_2))$ from the proof of Corollary \ref{col:condbound2}. We also define sets $G(E_1,E_2)$ and $B(E_1,E_2)$ similarly to the proof of Lemma \ref{lem:comparwpoisson}, but there are minor differences related to the decorrelation.
\begin{enumerate}
    \item $G(E_1,E_2)$: This is the set of all configurations $(\sigma,\phi)$ in $M(E_1,E_2)$ such that the support of the configuration contains no vortices in the box $B_K(e)$.
    \item $B(E_1,E_2)$: These are all configurations in $M(E_1,E_2)$ that do not belong to $G(E_1,E_2)$. They are characterized by having a vortex in the support that intersects the box $B_K(e)$. In fact, this vortex is minimal since in we assume $M(E_1,E_2) \in \tilde{E}$ and $\tilde{E}$ has no non-minimal vortices that are within a $K$ neighborhood of $\gamma$.
    
    \item $TG(E_1,E_2)$: These are all configurations $(\sigma,\phi)$ in $M(E_1 \cup \{e\},E_2)$ whose support does not contain any element in $B_K(e)$ aside from the minimal vortex at $e$. This is in fact all of $M(E_1 \cup \{e\},E_2)$ since we are considering events in the set $\tilde{E}$.
\end{enumerate}

We remark that $M(E_1 \cup \{e\},E_2)$ is the set of all configurations in $M(E_1,E_2)$ that have the minimal vortex $P(e)$ in the support.

Our goal is to show that $Z(B(E_1,E_2))$ is small relative to $Z(M(E_1,E_2))$ and that the sum of all configurations in $Z(TG(E_1,E_2))$ is a simple product $\mathcal{D}_{\beta,\kappa}$ of  $Z(G(E_1,E_2))$. The former is the same as what we have done in similar to what we have done earlier. The latter requires some decorrelation estimates.

By a union bound,
the ratio of $Z(B(E_1,E_2))$ to $Z(M(E_1,E_2))$ is bounded from above by $\sum_{e' \in B_K(e)} \mathbb{P}_{\tilde{E}}(P(e') \in \supp(\sigma,\phi)|M(E_1,E_2))$ since all events in $Z(B(E_1,E_2))$ necessarily have some minimal vortex excitation in $B_K(e)$. Each individual term $\mathbb{P}_{\tilde{E}}(P(e') \in \supp(\sigma,\phi)|M(E_1,E_2))$ is bounded from above by $\Phi_{UB}(P(e)) \le \mathcal{B}^6$ from Corollary \ref{col:condbound2}. Thus, we see that the ratio of $Z(B(E_1,E_2))$ and $Z(M(E_1,E_2))$ is less than $\tilde{\mathfrak{c}}:= 12K^4 \mathcal{B}^6$ by the union bound above. 

Like in the proof of Lemma \ref{lem:comparwpoisson}, we would like to say that $\frac{Z(TG(E_1,E_2)}{Z(G(E_1,E_2))}$ is the constant $\mathcal{D}_{\beta,\kappa}$. The only thing preventing this is that the activation of the minimal $P(e)$ does not separate as a product from the remaining activations. However, we specifically use the fact the configurations in $TG(E_1,E_2)$ and $G(E_1,E_2)$ have no excitation in the box $B_K(e)$ aside from the minimal vortex in $P(e)$ for $TG(E_1,E_2)$.

\textit{Proof of a Decorrelation Inequality:}

Let $U$ be the support of some configuration in $G(E_1,E_2)$. Recalling the partition function $\mathcal{Z}$ from equation, we will show that the following is true  \eqref{eq:somepartfunc} 
\begin{equation} \label{eq:someratio}
\left|\frac{\mathcal{Z}(U \cup P(e))}{\mathcal{Z}(U)} - \mathcal{D}_{\beta,\kappa}\right| \le K^{3} e^{-cK} L,
\end{equation}
where $L$ is the constant from equation  \eqref{eq:defL} and our values of $K$ and $c$ satisfy the decorrelation estimate from Lemma \ref{lem:decoresti1}.
Since $Z(G(E_1,E_2)) = \sum_{U} \mathcal{Z}(U)$ and $Z(TG(E_1,E_2))= \sum_U \mathcal{Z}(U \cup P(e))$, this shows that$| \frac{Z(TG(E_1,E_2))}{Z(G(E_1,E_2))} - \mathcal{D}_{\beta,\kappa}| \le K^{3} e^{-cK} L$. 

The proof of \eqref{eq:someratio} is very similar to the proof of Part 1 of Lemma \ref{lem:WeirdCondition}.
Let $B_K(e)$ be the box of size $K$ centered around the minimal vortex $P(e)$. We start by choosing a spanning tree $T(e)$ that is a spanning tree of $B_K(e)$ and its complement. As before, when we gauge out a configuration $\sigma$ with respect to $T(e)$ to $\tilde{\sigma}$, this allows us to split the configuration into those with support $\tilde{\sigma}_1$ in $B_K(e)$ and $\tilde{\sigma}_2$ in $B_K(e)^c$. Furthermore, the only non-trivial edge of $\tilde{\sigma}_1$ would be the edge $e$ that is the center of the minimal vortex $P(e)$. We can then reintroduce all other configuration by introducing a new field $\eta:V_N \to G$.

We have,
\begin{equation}
\begin{aligned}
    &\mathcal{Z}(P(e) \cup U)  = \sum_{\eta_v,\phi_v} \prod_{e'=(v,w) \not \in E(P(e)) \cup B_K(e)^c} \exp[\kappa(f(\eta_v \eta_w^{-1}, \phi_v\phi_w^{-1}) -f(1,1))]\\
    & \hspace{0.0 cm} \times \sum_{\tilde{\sigma}_1 \text{ gauged with }T(e)} \prod_{p \in P(e)} \exp[\beta (\rho((\td \tilde{\sigma}_1)_p) -\rho(1)) ]\\
    & \hspace{0.0 cm}\times \prod_{\substack{e'=(v,w)\\ e' \in E(P(e))}} \exp[\kappa(f(\eta_v^1(\tilde{\sigma}_1)_{e'} (\eta_w^1)^{-1}, \phi_v^1 (\phi_w^1)^{-1}) - f(\eta_v^1 (\eta_w^1)^{-1},\phi_v^1 (\phi_w^1)^{-1}))] \\ \hspace{0.3 cm}& \times \sum_{\tilde{\sigma}_2 \text { gauged with } T(e)} \prod_{p \in U} \exp[\beta (\rho((\td \tilde{\sigma}_2)_p) -\rho(1))] \\ \hspace{0.0 cm} & \times \hspace{0.0 cm}  \prod_{\substack{e'=(v,w)\\ e' \in B_K(e)^c} } \exp[\kappa(f(\eta_v^2(\tilde{\sigma}_2)_{e'} (\eta_w^2)^{-1}, \phi_v^2 (\phi_w^2)^{-1}) - f(\eta_v^2 (\eta_w^2)^{-1},\phi_v^2 (\phi_w^2)^{-1}))],
\end{aligned}
\end{equation}
 where $\eta^{1,2},\phi^{1.2}$ denote the values of $\eta_v$ and $\phi_v$
for vertices that belong to $V(P(e))$ or $B_K(e)^c$. We can divide by the partition function $Z_{K_N}$ to write the product as,
\begin{equation}
\begin{aligned}
    &\frac{\mathcal{Z}(P(e) \cup U) }{Z_{K_N}} =
    \sum_{\substack{\eta_v^1,\phi_v^1,\eta_v^2,\phi_v^2}} \mathbb{P}_{K_N}(\eta_v^1,\phi_v^1,\eta_v^2,\phi_v^2)\\
     & \times \sum_{\tilde{\sigma}_1 \text{ gauged with }T(e)} \prod_{p \in P(e)} \exp[\beta (\rho((\td \tilde{\sigma}_1)_p) -\rho(1)) ] \\ & \times \prod_{\substack{e'=(v,w)\\ e'\in E(P(e))}} \exp[\kappa(f(\eta_v^1(\tilde{\sigma}_1)_{e'} (\eta_w^1)^{-1}, \phi_v^1 (\phi_w^1)^{-1}) - f(\eta_v^1 (\eta_w^1)^{-1},\phi_v^1(\phi_w^1)^{-1}))] \\& \times \sum_{\tilde{\sigma}_2 \text { gauged with } T(e)} \prod_{p \in U} \exp[\beta (\rho((\td \tilde{\sigma}_2)_p) -\rho(1))]\\
     & \times \prod_{\substack{e'=(v,w)\\ e' \in B_K(e)^c}} \exp[\kappa(f(\eta_v^2(\tilde{\sigma}_2)_{e'} (\eta_w^2)^{-1}, \phi_v^2 (\phi_w^2)^{-1}) - f(\eta_v^2 (\eta_w^2)^{-1},\phi_v^2 (\phi_w^2)^{-1}))].
\end{aligned}
\end{equation}

We would like to split $\mathbb{P}_{K_N}(\eta_v^1,\phi_v^1,\eta_v^2,\phi_v^2)$ into $\mathbb{P}_{K_N}(\eta_v^1,\phi_v^1) \mathbb{P}_{K_N}(\eta_v^2,\phi_v^2)$. If we had equality in this splitting, then we would immediately get the product $\mathcal{D}_{\beta,\kappa} \frac{\mathcal{Z}(P(e) \cup U)}{Z_{K_N}}$.

We have the error $$|\mathbb{P}_{K_N}(\eta_v^1,\phi_v^1, \eta_v^2,\phi_v^2) - \mathbb{P}_{K_N}(\eta_v^1,\phi_v^1) \mathbb{P}_{K_N}(\eta_v^2,\phi_v^2)| \le K^{3}e^{-cK} \mathbb{P}_{K_N}(\eta_v^2,\phi_v^2)$$ from Lemma \ref{lem:decoresti1}. Combined with our observation in the previous paragraph,  we see that 
\begin{equation}
\begin{aligned}
   & |\frac{\mathcal{Z}(P(e) \cup U) }{Z_{K_N}} - \mathcal{D}_{\beta,\kappa} \frac{\mathcal{Z}(U)}{Z_{K_N}}| \\
   &\le 
     \sum_{\substack{\eta_v^1,\phi_v^1,\eta_v^2,\phi_v^2}} |\mathbb{P}_{K_N}(\eta_v^1,\phi_v^1,\eta_v^2,\phi_v^2) - \mathbb{P}_{K_N}(\eta_v^1,\phi_v^1) \mathbb{P}_{K_N}(\eta_v^2,\phi_v^2)|\\
     &\times \sum_{\tilde{\sigma}_1 \text{ gauged with }T(e)} \prod_{p \in P(e)} \exp[\beta (\rho((\td \tilde{\sigma}_1)_p) -\rho(1)) ]\\ & \times  \prod_{\substack{e'=(v,w)\\ e' \in E(P(e))}} \exp[\kappa(f(\eta_v^1(\tilde{\sigma}_1)_{e'} (\eta_w^1)^{-1}, \phi_v^1 (\phi_w^1)^{-1}) - f(\eta_v^1 (\eta_w^1)^{-1},\phi_v^1(\phi_w^1)^{-1}))] \\& \times \sum_{\tilde{\sigma}_2 \text { gauged with } T(e)} \prod_{p \in U} \exp[\beta (\rho((\td \tilde{\sigma}_2)_p) -\rho(1))] \\
     & \times \prod_{\substack{e'=(v,w)\\ e' \in B_K(e)^c}} \exp[\kappa(f(\eta_v^2(\tilde{\sigma}_2)_{e'} (\eta_w^2)^{-1}, \phi_v^2 (\phi_w^2)^{-1}) - f(\eta_v^2 (\eta_w^2)^{-1},\phi_v^2 (\phi_w^2)^{-1}))]\\&
     \le K^{3}e^{-cK} \sum_{\tilde{\sigma}_1 \text{ gauged with }T(e)} \prod_{p \in P(e)} \exp[\beta (\rho((\td \tilde{\sigma}_1)_p) -\rho(1)) ] \\ & \times \prod_{\substack{e'=(v,w)\\ e' \in E(P(e)) }} \exp[\kappa(f(\eta_v^1(\tilde{\sigma}_1)_{e'} (\eta_w^1)^{-1}, \phi_v^1 (\phi_w^1)^{-1}) - f(\eta^1_v(\eta^1_w)^{-1},\phi^1_v(\phi^1_w)^{-1}))]\\
     & \times \sum_{\eta_v^2,\phi_v^2} \mathbb{P}_{K_N}(\eta_v^2,\phi_v^2)\sum_{\tilde{\sigma}_2 \text { gauged with } T(e)} \prod_{p \in U} \exp[\beta (\rho((\td \tilde{\sigma}_2)_p) -\rho(1))] \\ & \times \prod_{\substack{e'=(v,w)\\e' \in B_K(e)^c}} \exp[\kappa(f(\eta_v^2(\tilde{\sigma}_2)_{e'} (\eta_w^2)^{-1}, \phi_v^2 (\phi_w^2)^{-1}) - f(\eta^2_v (\eta^2_w)^{-1},\phi^2_v  (\phi^2_w)^{-1}))]\\
     & \le K^{3}e^{-cK} L \frac{\mathcal{Z(U)}}{Z_{K_N}},
\end{aligned}
\end{equation}
where $L$ is the constant from equation 
\eqref{eq:defL}. We can now divide out by $\mathcal{Z(U)} Z_{K_N}^{-1}$ on both sides. This proves the desired intermediary inequality \eqref{eq:someratio}.

\textit{Return to bounding $b_3$}

Recall $\tilde{\mathfrak{c}}:=12K^4 \mathcal{B}^6 $. Provided $\tilde{\mathfrak{c}}$ is less than $1$, we can derive the following. Since we have $\frac{Z(B(E_1,E_2))}{Z(M(E_1,E_))}= \frac{Z(B(E_1,E_2))}{Z(G(E_1,E_2)) + Z(B(E_1,E_2))} \le \tilde{\mathfrak{c}}$, we see that $Z(B(E_1,E_2)) \le \frac{\tilde{\mathfrak{c}}}{1-\tilde{\mathfrak{c}}} Z(G(E_1,E_2))$. 

For a lower bound, we see that we have,

\begin{equation}
\begin{aligned}
   & \frac{Z(M(E_1 \cup \{e\}, E_2)}{Z(M(E_1,E_2))} = \frac{Z(TG(E_1,E_2))}{Z(G(E_1,E_2)) + Z(B(E_1,E_2))} \\
   & \ge  (1- \tilde{\mathfrak{c}}) \frac{Z(TG(E_1,E_2)}{Z(G(E_1,E_2))} \ge (1-\tilde{\mathfrak{c}})[\mathcal{D}_{\beta,\kappa}- K^{3} e^{-cK} L].
\end{aligned}
\end{equation}

For an upper bound, we have that,
\begin{equation}
\begin{aligned}
     &\frac{Z(M(E_1 \cup \{e\}, E_2)}{Z(M(E_1,E_2))} \le \frac{Z(TG(E_1,E_2)) }{Z(G(E_1,E_2)} \\
     & \le [\mathcal{D}_{\beta,\kappa} + K^{3}e^{-cK} L].
\end{aligned}
\end{equation}

These are lower and upper bounds of $\mathbb{E}_{\tilde{E}}[\mathbbm{1}[F_{P(e)}]|M(E_1,E_2)]$ regardless of the sets $E_1$ and $E_2$. We remark that these must be lower and upper bounds on $\mathbb{P}_{\tilde{E}}(F_{P(e)})$. By removing the conditioning on $M(E_1,E_2)$ and summing up overall edges in $\gamma$, we  can derive the following bound on $b_3$,
\begin{equation}
    b_3 \le |\gamma| \left|(1-\tilde{\mathfrak{c}})[\mathcal{D}_{\beta,\kappa}- K^{3} e^{-cK} L] -   [\mathcal{D}_{\beta,\kappa} + K^{3}e^{-cK} L]\right|.
\end{equation}

This is also a bound on $|\lambda - |\gamma| \mathcal{D}_{\beta,\kappa}|$, referring to our earlier lower and upper bounds on $\mathbb{P}_{\tilde{E}}(F_{P(e)})$.

Our bounds on $b_1$ and $b_2$ are much simpler. By simple summation, we see that $|\gamma|12 K^4 \mathcal{B}^{12}$ suffices as an upper bound for both $b_1$. This merely uses that the probability of excitation of a single minimal vortex is less than $\Phi_{UB}(P(e))$ and the probability of excitation of two minimal vortices that do not intersect is less than $\Phi_{UB}(P(e) \cup P(e)') = \Phi_{UB}(P(e))^2 \le \mathcal{B}^{12}$. We can then perform the summation over all $e \in \gamma $ and $e'\in B_K(e)$.

Upon reaching here, the calculation of $d_{TV}(\mathcal{L}_{\tilde{E}}(M_{\gamma}),\text{Poisson}(|\gamma| \mathcal{D}_{\beta,\kappa}))$ is the same as it was in the previous Lemma \ref{lem:comparwpoisson}. We use Theorem \ref{thn:PoisApproxGen2} to compute $d_{TV}(\mathcal{L}_{\tilde{E}}(M_{\gamma}), \text{Poisson}(\lambda))$ and apply $d_{TV}(\text{Poisson}(\lambda), \text{Poisson}(|\gamma| \mathcal{D}_{\beta,\kappa}) \le |\lambda - |\gamma| \mathcal{D}_{\beta,\kappa}|$. We get the desired comparision in \eqref{eq:fincompar} by a triangle inequality.
\end{proof}

By
combining this error estimate with the error estimate from Theorem \ref{thm:WilsonLoopReplowdisorder}, we will be able to prove the following main result. The proof is exactly the same as short proof of Theorem \ref{thm:MainThm1} and we will omit the proof.
\begin{thm}\label{thm:MainThm2}
Assume that the conditions of Theorem \ref{thm:WilsonLoopReplowdisorder} hold and Lemma \ref{lem:comparwpoisson2} hold. Let $X$ be a Poisson random variable with parameter $|\gamma| \mathcal{D}_{\beta,\kappa}$. Then, we have the following estimate on the Wilson loop expectation,
\begin{equation}
\begin{aligned}
   & |\mathbb
    {E}_{\tilde{E}}[W_{\gamma}] -\mathbb{E}[\mathcal{D}_{\beta,\kappa}^X]|  \le 12|\gamma|K^4 \mathcal{B}^{12}\\
    &\hspace{1.5 cm}+ 2|\gamma|\left|(1-\tilde{\mathfrak{c}})[\mathcal{D}_{\beta,\kappa}- K^{3} e^{-cK} L] -   [\mathcal{D}_{\beta,\kappa} + K^{3}e^{-cK} L]\right|
      \\&\hspace{1.5 cm}+ |\gamma|^2 \mathcal{B}^{6} L K^{3} e^{-cK} \exp[\mathcal{B}^{6}|\gamma|L]\\
     &\hspace{1.5 cm}+ \frac{|\gamma| \mathcal{B}^{6} [(1+e^{-c'K})L] e^{-c'K}}{\min_{\substack{\eta_v,\phi_v\\ v \in \text{ minimal vortex }}} \mathbb{P}_{K_N}(\eta_v,\phi_v)} \exp[|\gamma| \mathcal{B}^{6}[(1+e^{-c'}K) L]].
\end{aligned}
\end{equation}
As a consequence, we also have that,
\begin{equation}
\begin{aligned}
    &|\mathbb
    {E}[W_{\gamma}] -\mathbb{E}[\mathcal{D}_{\beta,\kappa}^X]|  \le 12|\gamma|K^4 \mathcal{B}^{12} + \mathbb{P}(\tilde{E}^c)\\
    &\hspace{1.5 cm}+ 2|\gamma|\left|(1-\tilde{\mathfrak{c}})[\mathcal{D}_{\beta,\kappa}- K^{3} e^{-cK} L] -   [\mathcal{D}_{\beta,\kappa} + K^{3}e^{-cK} L]\right|
      \\&\hspace{1.5 cm}+ |\gamma|^2 \mathcal{B}^{6} L K^{3} e^{-cK} \exp[\mathcal{B}^{6}|\gamma|L]\\
     &\hspace{1.5 cm}+ \frac{|\gamma| \mathcal{B}^{6} [(1+e^{-c'K})L] e^{-c'K}}{\min_{\substack{\eta_v,\phi_v\\ v \in \text{ minimal vortex }}} \mathbb{P}_{K_N}(\eta_v,\phi_v)} \exp[|\gamma| \mathcal{B}^{6}[(1+e^{-c'}K) L]].
\end{aligned}
\end{equation}

\end{thm}

\begin{rmk}
Just as in Remark \ref{rmk:onK}, the main point of the introduction of $K$ is that when $K$ is large enough, we can use $e^{-cK}$ to suppress factors of $|\gamma|$ even when $|\gamma|$ is large. In terms that do not have $e^{-cK}$, we only have a polynomial power of $K$ and can suppress this polynomial power of $K$ with $\exp[-\beta]$ for $\beta$ relatively large. Just as in the aforementioned remark, we will choose $K= O(\beta)$.
\end{rmk}

\appendix
\section{Details on Vertex Decompositions}

In this section, we will describe assorted facts on Higgs boson vertex configurations that would be useful for our decomposition theorem.

The following Lemma details the behaviors of the Higgs boson charges on the complement of the support of a distribution.
\begin{lem} \label{lem:compcharg}
Let $P$ be the support of a configuration $\calC$ of Higgs boson and gauge fields. As previously, we use $V(P)$ and $E(P)$ to denote the set of vertices and edges that are associated to $P$. 
Let the complement $V(P)$ in $V_N$ be divided into separate connected components $V(P)^c = B_1 \cup B_2 \cup \ldots \cup B_N$.  Then the following statements hold true,
\begin{itemize}
    \item There is a single charge $c_i$ such that each vertex $v_i$ in $B_i$ is assigned the Higgs boson charge $c_i$, e.g. $\phi_{v_i} = c_i, \forall v_i \in B_i$. 
    \item If $v$ is a vertex in $V(P)$ that is connected by an edge in $E_N$ to a vertex $v_i \in B_i$, then $\phi_v =c_i$, where $c_i$ is the common color of each vertex in $v_i$. 
\end{itemize}
\end{lem}
\begin{proof}
We start with the proof of the first item. 

If $B_i$ is not monocharged, then there are two vertices $v_i,w_i$ in $B_i$ that are not the same color as well as a path $p_i$ consisting entirely of vertices in $B_i$ connecting $v_i$ and $w_i$. At least one of the edges,$e_i$ , on this path $p$ must have opposite charges on its neighboring vertices. This would imply that $e_i$ would be in $\supp(\calC)$; furthermore, it would imply that the vertices of this edge would belong in $V(P)$ rather than the complement. This is a contradiction. This proves the first item.

Now, we prove the second item. If $v_i$ is a vertex in $B_i$ and $w$ is a vertex in $V(P)$ adjacent to $v$ that does not have the same charge $c_i$, then  the edge $e=(v_i,w)$ $w$ has the same charge as $v$ would be in $\supp(\calC)$. Thus, the vertex $v$ must be in $V(P)$. This is a contradiction.
\end{proof}

The behavior of the Higgs boson configurations is much like the Ising model. One can imagine regions of monocharged components containing regions of other monocharged components recursively. It is important to understand the relationship between these monocharged components and the support of the configuration. The following lemma details these relationships.

\begin{lem} \label{lem:bndry}
Let $(\phi,\sigma)$ be a configuration $\calC$ of Higgs boson and gauge fields. Let $V$ be a some connected monocharged set(e.g. all the vertices are assigned the same Higgs boson charge) in $V_N$ and define the set $\mathcal{V}(V)$ as follows,
\begin{equation}
    \mathcal{V}(V):= \{w \in V_N: \exists \text{ path } p(v \to w) \text{ s.t. } \forall \text{ vertices } a \in p, \phi_a = \phi_w \}.
\end{equation}
    Namely, $\mathcal{V}(V)$ is the collection of vertices in $V_N(V)$ that can be connected to $V$ with vertices of the same Higgs boson charge as $V$.
    
    We define $\mathcal{E}(V)$ as the set of vertices connecting vertices in $\mathcal{V}$ to its complement in $V_N$.
\begin{equation}
    \mathcal{E}(V):=\{e=(v,w) \in E_N: v \in \mathcal{V}(V), w \in {\mathcal{V}(V)}^c\}.
\end{equation}   
We finally define $\mathcal{P}(V)$ to be the set of plaquettes in $P_N$ that have one of the edges in $\mathcal{E}$ as a boundary edge. Namely,
\begin{equation}
    \mathcal{P}(V):=\{ p \in P_N: \exists e \in \mathcal{E}(V) \text { s.t. } e \in \delta p \}.
\end{equation}

Then, $\mathcal{P}(V)$ is a subset of $\supp(\calC)$. We can call $\mathcal{P}(V)$ the boundary of $\mathcal{V}(V)$.

 Now assume further that there is no vertex in $\mathcal{V}(V)$ that is a boundary vertex of $P_N$. We can make the following statements on the decomposition of $\mathcal{P}(V)$ into connected components.
 
 There is a unique connected component, which we will call the external boundary of $\mathcal{V}(V)$; we will denote this by $EB(\mathcal{V}(V))$ satisfying the following properties.
 \begin{itemize}
     \item $EB(\mathcal{V}(V))$ is connected.
     \item Embed the subset $\mathcal{V}(V)$ into the full lattice $\mathbb{Z}^d$. Consider the connected components of ${\mathcal{V}(V)}^c$ in $\mathbb{Z}^d$ into connected components ${\mathcal{V}(V)}^c = B_1 \cup B_2 \cup \ldots \cup B_N$. Let $B_1$ be the unique non-compact component, so it extends to $\infty$. If we define $\mathcal{E}(B_1)$ to be the set of edges connecting $B_1$ to its complement and $\mathcal{P}(B_1)$ to be the set of plaquettes having at least one edge from $\mathcal{E}(B_1)$ on its boundary. Then, $EB(\mathcal{V}(V)) = \mathcal{P}(B_1)$.
 \end{itemize}
    
\end{lem}

\begin{rmk}
The statement that $\mathcal{P}$ as we have defined it above is in $\supp{C}$ is by definition. The second part of the above lemma is to formally state the intuition that a set that is monocharged but connected and compact is separated from the outside vertices that do not share its charge by a single connected boundary. All other boundaries of the set $\mathcal{V}$ are internal and separate it from the islands of charge that are internal to $\mathcal{V}$.
\end{rmk}

\textbf{Acknowledgements:} The author would like to thank  Sourav Chatterjee for suggesting this problem as well as for useful advice and Sky Cao for useful discussions.




\bibliographystyle{plain}
\bibliography{references}

\end{document}